\documentclass[acmsmall,screen,authorversion]{acmart}

\usepackage{amsmath}
\usepackage{amsthm}
\usepackage{proof}
\usepackage{bbm}
\usepackage{dsfont}
\usepackage{hyperref}
\usepackage{xcolor}
\usepackage{colortbl}
\usepackage{url}
\usepackage{caption}
\usepackage{subcaption}
\usepackage{graphicx}
\usepackage{paralist}
\usepackage{etoolbox}
\usepackage{enumitem}
\usepackage{empheq}
\usepackage{bussproofs}
\usepackage[ligature, inference]{semantic}
\usepackage{cleveref}
\usepackage{multirow}
\usepackage{makecell}
\usepackage[normalem]{ulem}
\usepackage{listings}

\RequirePackage{tensor}
\RequirePackage{leftindex}

\RequirePackage{paralist}
\RequirePackage{pifont}
\RequirePackage[all=normal,paragraphs]{savetrees}
\RequirePackage{tikz}
\RequirePackage{xspace}

\newcommand*{\code}[1]{{\mathbf{#1}}}

\newcommand*{\Downarrows}{\Downarrow}
\newcommand*{\Downarrowt}{\Downarrow_\mathit{tgt}}
\newcommand*{\Downarrowi}{\Downarrow_\mathit{int}}

\newcommand*{\FAntiChain}{\mathit{FAChain}}

\newcommand*{\RDB}{\mathit{RDB}}
\newcommand*{\Tensor}{\mathit{Tensor}}
\newcommand*{\WTensor}{\mathit{Tensor}^{\dagger}}
\newcommand*{\Index}{\mathit{Index}}

\newcommand*{\State}{\mathit{State}}
\newcommand*{\Str}{\mathbb{S}}

\newcommand*{\Var}{\mathit{Var}}

\newcommand*{\Comm}{\mathit{Com}}

\newcommand*{\extend}{\mathit{extend}}

\newcommand*{\tzero}{\mathit{z}}

\DeclareMathOperator*{\concat}{{+\!\!+}}

\newcommand*{\R}{\mathbb{R}}
\newcommand*{\Z}{\mathbb{Z}}

\newcommand*{\cN}{\mathcal{N}}
\newcommand*{\cP}{\mathcal{P}}

\newcommand*{\db}[1]{\ensuremath{\llbracket #1 \rrbracket}}

\newcommand*{\defeq}{\stackrel{\mathrm{def}}{=}}

\renewcommand*{\dim}{\mathsf{dim}}
\newcommand*{\dom}{\mathsf{dom}}
\newcommand*{\image}{\mathsf{image}}

\newcommand*{\op}{\mathsf{op}}

\newcommand*{\intt}{\mathsf{int}}

\newcommand*{\real}{\mathsf{real}}
\newcommand*{\lintt}{{\mathsf{int}\dagger}}
\newcommand*{\lreal}{{\mathsf{real}\dagger}}

\definecolor{darkorange}{rgb}{0.9, 0.4, 0.0}
\definecolor{darkgreen}{rgb}{0.0, 0.6, 0.0}

\newcommand{\rwl}[1]{{#1}}
\newcommand{\rhl}[1]{{#1}}
\newcommand{\rshl}[1]{{#1}}

\theoremstyle{plain}

\newtheorem{corollary}{Corollary}
\newtheorem{definition}{Definition}
\newtheorem{lemma}{Lemma}
\newtheorem{proposition}{Proposition}
\newtheorem{theorem}{Theorem}
\theoremstyle{definition}
\newtheorem{example}{Example}

\makeatletter
\AfterEndEnvironment{assumption}{\@doendpe}
\AfterEndEnvironment{definition}{\@doendpe}
\AfterEndEnvironment{proposition}{\@doendpe}
\AfterEndEnvironment{theorem}{\@doendpe}
\makeatother

\crefname{enumi}{}{}

\newcommand{\commentout}[1]{}

\newcommand{\bbR}{\mathbb{R}}

\newcommand{\rvar}[1]{\mathbf{#1}}
\newcommand{\rvarstr}[1]{\text{``}\rvar{#1}\text{''}}

\newcommand{\kwscore}{\mathbf{score}}

\newcommand{\mcN}{\mathcal{N}}

\crefname{assumption}{Assumption}{Assumptions}
\crefname{figure}{Fig{.}}{Figs{.}}
\crefname{table}{Table}{Tables}
\crefname{definition}{Definition}{Definitions}
\crefname{theorem}{Theorem}{Theorems}
\crefname{lemma}{Lemma}{Lemmas}
\crefname{proposition}{Proposition}{Propositions}
\crefname{corollary}{Corollary}{Corollaries}
\crefname{problem}{Problem}{Problems}
\crefname{example}{Example}{Examples}
\crefname{fact}{Fact}{Facts}
\crefname{conjecture}{Conjecture}{Conjectures}
\crefname{remark}{Remark}{Remarks}
\crefname{condition}{Condition}{Conditions}
\crefname{requirement}{Requirement}{Requirements}
\crefname{enumi}{}{}
\crefname{equation}{Eq{.}}{Eqs{.}}
\crefformat{section}{\S#2#1#3}
\crefmultiformat{section}{\S#2#1#3}{ and~\S#2#1#3}{, \S#2#1#3}{, and~\S#2#1#3}

\setcopyright{acmlicensed}
\acmJournal{PACMPL}
\acmVolume{10}
\acmNumber{POPL}
\acmArticle{0}
\acmMonth{1}
\acmYear{2026}
\copyrightyear{2026}

\allowdisplaybreaks

\begin{document}

\title{Optimising Density Computations in Probabilistic Programs via
  Automatic Loop Vectorisation}

\thanks{This work was supported by the National Research Foundation of Korea(NRF) grant funded by the Korean Government(MSIT) (No. RS-2023-00279680).}

\author{Sangho Lim}
\authornote{Both authors contributed equally to this paper.}
\email{lim.sang@kaist.ac.kr}
\orcid{0009-0006-6172-0644}
\affiliation{%
  \institution{KAIST}
  \country{Korea}
}

\author{Hyoungjin Lim}
\authornotemark[1]
\email{lmkmkr@kaist.ac.kr}
\orcid{0009-0002-9728-3125}
\affiliation{%
  \institution{KAIST}
  \country{Korea}
}

\author{Wonyeol Lee}
\email{wonyeol.lee@postech.ac.kr}
\orcid{0000-0003-0301-0872}
\affiliation{%
  \institution{POSTECH}
  \country{Korea}}

\author{Xavier Rival}
\email{rival@di.ens.fr}
\orcid{0000-0002-2875-6171}
\affiliation{%
  \institution{INRIA Paris, CNRS, ENS, and PSL University}
  \country{France}
}

\author{Hongseok Yang}
\email{hongseok.yang@kaist.ac.kr}
\orcid{0000-0003-1502-2942}
\affiliation{%
 \institution{KAIST}
 \country{Korea}}

\begin{abstract}
  Probabilistic programming languages (PPLs) are a popular tool for high-level modelling across many fields.
  They provide a range of algorithms for probabilistic inference,
  which analyse models by learning their parameters from a dataset or estimating their posterior distributions.
  However, probabilistic inference is known to be very costly.
  One of the bottlenecks of probabilistic inference stems from the iteration
  over entries of a large dataset or a long series of random samples.
  Vectorisation can mitigate this cost, but manual vectorisation is error-prone,
  and existing automatic techniques are often ad-hoc and limited,
  unable to handle general repetition structures, such as nested loops and loops with data-dependent control flow,
  without significant user intervention.
  To address this bottleneck, we propose a sound and effective method
  for automatically vectorising loops in probabilistic programs.
  Our method achieves high throughput using speculative parallel
  execution of loop iterations, while preserving the semantics
  of the original loop through a fixed-point check.
  We formalise our method as a translation from an imperative PPL into a lower-level target language
  with primitives geared towards vectorisation.
  We implemented our method for the Pyro PPL and evaluated it on a range of probabilistic models.
  Our experiments show significant performance gains against an existing vectorisation baseline,
  achieving \rshl{$1.1$}--$6\times$ speedups and reducing GPU memory usage in many cases.
  Unlike the baseline, which is limited to a subset of models, our method effectively handled all the tested models.
\end{abstract}



\received{10 July 2025}
\received[revised]{23 October 2025}

\maketitle

\setitemize{topsep=0.5em}
\setenumerate{topsep=0.5em}
\section{Introduction}
Over the last two decades, probabilistic programming languages (PPLs) have emerged as important
tools for expressing and reasoning about models in machine learning, statistics, and 
computational science~\cite{Pyro,NumPyro,PyMC,Anglican,Edward2,Turing,Stan}.
They are used to analyse datasets from
epidemiology~\cite{epidemiology:21,epidemiology:22} and ecology~\cite{ecology:22,ecology:24}
to phylogenetics~\cite{phylogenetics:23} and robotics~\cite{robotics:23,robotics:24}.
PPLs enable users to model probabilistic phenomena naturally using primitives for probabilistic sampling 
and conditioning, as well as features from general-purpose programming languages,
such as loops, branches, and more recently, trainable parameters (e.g., neural networks). 
These models are then analysed using generic inference engines of
PPLs, which approximate posterior distributions or learn model parameters 
by performing algorithms from Bayesian statistics and deep learning, such as Hamiltonian Monte Carlo (HMC)~\cite{HMCphy,HMCcs}, 
Stochastic Variational Inference (SVI)~\cite{SVI}, and Stochastic Gradient Descent (SGD)~\cite{SGD:51,SGD:52}. 

This paper aims to reduce the high computational cost of PPL inference engines,
a challenge that persists despite advances in statistical algorithms and compilation. 
\rshl{
Our work focuses on scenarios requiring repeated evaluation of probabilistic programs during parameter learning or posterior inference,
particularly those involving intensive computations of model densities and their gradients arising from repetition structures,
often implemented as for-loops.}
We develop a sound automatic method for optimising such loops through tensor-based vectorisation,
thereby enabling the underlying inference engines to analyse these models more efficiently.

Models with repetition structures are common in practice. Graphical models from statistics
even employs a special notation, called \emph{plate}, to compactly describe such structures. 
The simplest repetition structure is a loop over a dataset
where each iteration is independent of the others
in the sense that it does not have data dependency from the other iterations.
Our optimisation allows the inference engines to process such loops in parallel using tensors,
yielding significant benefits for medium-to-large datasets. 
More challenging examples come from sequence models and hierarchical models,
such as Bayesian hierarchical models~\cite{hierarchical:03,hierarchical:19} and hidden Markov models~\cite{hmm:09},
where repetition structures are implemented as nested loops,
or they lack independence and are coded as loops whose iterations are data-dependent on the previous ones.  
Even in these complex scenarios, our optimisation enables the inference engines to identify 
the available parallelism and exploit it through tensor-based vectorisation,
achieving significant speedups in posterior inference and parameter learning.

Speeding up PPL inference engines using tensor-based vectorisation is a known idea
in the PPL community, supported by a range of existing primitives, 
such as \texttt{vmap}~\cite{jax}, \texttt{vectorized\_map}~\cite{vectorized-map}, 
\texttt{vectorized\_markov}~\cite{funsor:2019}, and the \texttt{MarkovChain} distribution~\cite{tensorflow-probability}.
However, these primitives fall short of fully automatic vectorisation of 
the general repetition structures that we target.
Concretely, they require users
to write repetition structures explicitly in vectorised form, or to provide
vectorisation-related information, e.g., dependency information in the repetition structures.
This task is both cumbersome and error-prone;
for example, manually vectorising the loop in \Cref{fig:programs-source-language} requires care due to
inter-iteration data dependencies, as detailed in \Cref{sec:source-lang}.
Moreover, these primitives are limited, and cannot handle
nested repetition structures or dynamically varying dependencies arising from internal random choices;
they often support only repetition without data dependencies.
In contrast, our vectorisation
is fully automatic, requires no user input,
and correctly handles general repetition structures,
including nested ones and those with complex, randomly-varying dependencies. 

When viewed as a code transformer, our optimisation can be understood as a \emph{non-standard} loop vectoriser. Given a model 
with a repetition structure expressed as a for-loop over $\{0,1,\ldots,n-1\}$, our 
optimisation transforms the loop to parallelism-friendly tensor-based vectorised code.
The vectorised code here is subtle, which goes beyond the simple standard idea of executing  
all the $n$ iterations of the original loop in parallel. To handle the case that 
the iterations of the loop are data-dependent, the code implements two
novel ideas: (i) \emph{speculative execution} of the loop body, and (ii) the use of a \emph{fixed-point check}
for testing the validity of speculation.
\vspace{-1pt}
\begin{itemize}[leftmargin=25pt]
\item[(i)] The code runs all the $n$ iterations of the
original loop in parallel, but some iterations are \emph{speculative} in the sense that their
starting states may differ from those of the corresponding 
iterations in the original loop. This speculation lets the code avoid waiting until 
the starting states of all the iterations of the original loop are known.
But the outcomes of these
speculative computations may diverge from those of the corresponding iterations of the original loop.
\item[(ii)] To address this divergence issue, the code repeats the parallel (and partly speculative) execution 
of all the $n$ iterations of the loop multiple times, and terminates this repetition 
when it finds out that all the speculative executions in the current repetition
are correct, i.e., they correspond to the $n$ iterations of the original loop. 
The termination mechanism here employs a \emph{fixed-point check} to detect
the correctness of speculation, and it is guaranteed to stop the repetition
within at most $n$ iterations
(due to a particular type of speculation we use).
\end{itemize}
\vspace{-1pt}
For instance, as we explain in \Cref{sec:target-lang}, our optimisation transforms the loop of \Cref{fig:programs-source-language} into vectorised code with a fixed-point check, which ensures that the parallel execution of the original loop is repeated
just twice. If the reader is familiar with abstract interpretation, a good intuition is 
that the vectorised code implements an executable version of the
collecting semantics from the abstract interpretation literature; an execution in the collecting semantics
corresponds to a set of executions in the standard semantics, and loops in the collecting semantics are interpreted
in terms of fixed-points.

\rshl{
When a model with a small degree of data dependency (e.g., a low-order Markov model) is expressed in a PPL that supports a powerful tensor framework, tensor broadcasting, and efficient parallel tensor operations, our optimisation can vectorise the computation of the model's density. The vectorised model runs more efficiently than the original one, thanks to parallelism and early termination enabled by our fixed-point check. Moreover, our vectorisation also accelerates gradient computations over the model’s density through the automatic differentiation provided by the PPL’s backend. In this case, automatic differentiation operates on the vectorised computation graph of the density, thereby computing gradients in a vectorised manner as well.
Such models and PPLs are common, and the resulting speedup directly improves the performance of PPL inference engines that repeatedly calculate or differentiate the model's density, such as Stochastic Variational Inference (SVI) and Maximum a Posteriori (MAP) estimation.
}

We prove the correctness of our optimisation using the formalism of expressing and analysing tensor-based
vectorised computations, which we also develop in the paper. To evaluate the effectiveness of our optimisation,
we implemented it for the Pyro PPL, and applied it to posterior inference and parameter learning 
by SVI and MAP estimation
\rshl{as well as to Markov Chain Monte Carlo (MCMC) sampling}.%
\footnote{\rshl{Our implementation is available at \url{https://github.com/Lim-Sangho/auto-vectorise-ppl.public}.}} 
Our experiments on eight Pyro models show that our optimisation significantly enhances inference performance,
achieving \rshl{$1.1$}--$6\times$ speedups over an existing vectorisation baseline,
and reducing GPU memory usage to $14$--$80\%$ 
for seven out of eight models, with one outlier exhibiting 
the increase to \rshl{$109$--$112\%$}.
The vectorisation baseline fails on some of these eight models,
while our optimisation was successfully applied to all the models.

We summarise the contributions of this paper below:
\vspace{-1pt}
\begin{enumerate}[leftmargin=25pt]
\item We propose a language with specialised primitives 
  that enables tensor-based vectorised execution of probabilistic programs.
  It is designed to facilitate vectorisation of loops in PPLs
  and support a broader class of models with control flow and nested structures
  (\Cref{sec:target-lang,sec:target-semantics-translation-soundness}).
\item We present a transformation that converts standard probabilistic
  programs into our proposed language for vectorisation.
  We prove the correctness of this transformation
  (\Cref{sec:target-lang,sec:target-semantics-translation-soundness}).
\item We implemented our transformation on top of the Pyro PPL, and evaluated it on several families of models,
  demonstrating that the transformation preserves program semantics
  while significantly improving execution time and in some cases dramatically reducing memory usage
  especially when the models have short-range data dependencies.
  (\Cref{sec:impl,sec:eval}).
\end{enumerate}

We clarify that our work focuses on optimising density and gradient computations in probabilistic programs,
but does not address the sample-generation cost in sampling-based inference algorithms.
\section{A Simple Programming Language for Unnormalised Probability Densities}
\label{sec:source-lang}

\begin{figure}[t]
  \small
    \centering
    \vspace{-2pt}
    \[
        \mathit{HMM} =
        \left[
        \begin{array}{@{\,}l@{\,}}
          x := 0.0; y := 0.0; \\
          \code{for} \; t \; \code{in} \; \code{range}(10) \; \code{do}\; \{\\
          \ \;\; x := \code{fetch}([(\rvarstr{z},t)]); \\
          \ \;\; \code{score}(\log p_\mcN(x; f(y), 1.0)); \\
          \ \;\; \code{score}(\log p_\mcN(o(t); g(x), 1.0)); \\
          \ \;\; y := x\; \}
        \end{array}
        \right],
        \;\;\,
        \mathit{HMM}' =
        \left[
        \begin{array}{@{\,}l@{\,}}
          x := 0.0; y := 0.0;\\
          \code{for} \; t \; \code{in} \; \code{range}(10) \; \code{do}\; \{\\
          \ \;\; x := \code{sample}([(\rvarstr{z},t)], \mcN(f(y), 1.0)); \\
          \ \;\; \code{observe}(\mcN(g(x), 1.0), o(t)); \\
          \ \;\; y := x\;\} 
        \end{array}
        \right]\!
    \]
    \vspace{-1.3em}
    \caption{Example programs describing a hidden Markov model.
      Here, $\mathit{HMM}$ is written in the language presented in \Cref{sec:source-lang},
      and $\mathit{HMM}'$ is written in a typical probabilistic programming language.
      }
    \label{fig:programs-source-language}
    \vspace{-0.7em}
\end{figure}

This paper is concerned with a class of PPL inference engines that view probabilistic 
programs as \emph{evaluators} of unnormalised posterior densities, instead of
\emph{generators} of weighted samples, and repeatedly run these evaluators during inference or parameter learning.
This class includes Stochastic Variational Inference (SVI), Hamiltonian Monte Carlo (HMC),
Maximum Likelihood Estimation (MLE), and Maximum a Posteriori (MAP) estimation,
forming the core of popular PPLs such as Pyro and Stan.
Our goal is to optimise the density evaluators expressed by probabilistic programs.

In this section, we formalise this density-evaluator 
view of probabilistic programs.
We start with an informal overview (\Cref{ssec:src-overview}), 
and then formally define a simple
programming language, where a program assumes a fixed collection of random 
variables, takes the values of these random variables as inputs, 
and then returns a real number, called total score, as an output (\Cref{ssec:src-syntax,ssec:src-semantics}).

\subsection{Overview of a Score-Computing Language}
\label{ssec:src-overview}

In the imperative language to be introduced in \Cref{ssec:src-syntax}, programs are intended to compute
the unnormalised (posterior) probability densities of some probabilistic models, 
where being unnormalised means the integrals of the densities at all inputs are not $1$ (which 
comes mostly from conditioning in those models). Their inputs are 
the values of the random variables in the models, and their outputs are 
the logarithms of the unnormalised densities at the given input values.

For instance,
the program $\mathit{HMM}$ in \Cref{fig:programs-source-language}
describes the unnormalised posterior density of a hidden Markov model (HMM),
and corresponds to $\mathit{HMM}'$ in the same figure written with the more traditional $\code{sample}$ 
and $\code{observe}$ constructs. The program defines a probabilistic model over random 
variables $\rvar{z}_t$ and $\rvar{o}_t$ for $t \in \{0,\ldots,9\}$, where $\rvar{z}_t$ 
(which is referred to by the notation $[(\rvarstr{z},t)]$ in $\mathit{HMM}$ with the 
string $\rvarstr{z}$) 
denotes the hidden state of a stochastic machine at step $t$, and $\rvar{o}_t$ (which is used
only implicitly by the second $\code{score}$ command in $\mathit{HMM}$) means the 
observation on the state $\rvar{z}_t$ made at step $t$. The probability density of the model is 
\bgroup
\[
p_\mathit{HMM}(\rvar{z}_0,\ldots,\rvar{z}_9,\rvar{o}_0,\ldots,\rvar{o}_9) 
\textstyle
= \prod_{t=0}^{9} \Bigl(p_\mcN(\rvar{z}_t; f(\rvar{z}_{t-1}), 1.0) \cdot p_\mcN(\rvar{o}_t; g(\rvar{z}_t), 1.0)\Bigr),
\]
\egroup
where $f,g$ are the functions in $\mathit{HMM}$, $p_\mcN$ the density of the normal distribution,
and $\rvar{z}_{-1}$ the constant $0.0$. The program conditions this density with observations
$\rvar{o}_t = o(t)$ for all $t$, where the $o(t)$ denote fixed observed values. Then, it computes 
the following quantity using the $\code{score}$ command:
\bgroup
\begin{align}
\label{eqn:hmm-logpdf}
\log p_\mathit{HMM}(\rvar{z}_0,\ldots,\rvar{z}_9,o(0),\ldots,o(9))
\textstyle
= \sum_{t=0}^{9} \Bigl(\log p_\mcN(\rvar{z}_t; f(\rvar{z}_{t-1}), 1.0) + \log p_\mcN(o(t); g(\rvar{z}_t), 1.0)\Bigr),
\!
\end{align}
\egroup
where the values of $\rvar{z}_t$ are given to the program as inputs. In each
iteration of the for-loop, the program reads the input value of $\rvar{z}_t$ using  
the $\code{fetch}$ construct and computes the two terms in the summand of \cref{eqn:hmm-logpdf} using the $\code{score}$ construct.
The outcomes of $\code{score}$ get accumulated, and the final sum is returned as the
output: the logarithm of the unnormalised density in \Cref{eqn:hmm-logpdf}.

In the rest of this section, we formally describe the syntax and semantics of this language of score-computing programs, which will 
be used to define our optimisation and prove its correctness.

\subsection{Syntax}
\label{ssec:src-syntax}

\begin{figure}[t]
  \begin{flushleft}
    \textbf{Syntax:}
  \end{flushleft}
  \vspace{-0.5em}
  \begin{align*}
    \mbox{\it Integer Expressions}\quad Z 
    & \;::=\; n \,\mid\, \smash{x^\intt} \,\mid\, \op_i(Z_1,\ldots, Z_k,E_1, \ldots, E_m) 
    \\
    \mbox{\it Real Expressions}\quad E
    & \;::=\; r \,\mid\, \smash{x^\real} \,\mid\, \op_r(Z_1,\ldots, Z_k,E_1, \ldots, E_m) 
    \\
    \mbox{\it Index Expressions}\quad I 
    & \;::=\; [(\alpha_1, Z_1);\, \ldots;\, (\alpha_k, Z_k)] 
    \quad \text{(for distinct $\alpha_i$'s)}
    \\
    \mbox{\it Commands}\quad C 
    & \;::=\; \code{score}(E) \,\mid\, \smash{x^\real} := \code{fetch}(I) \,\mid\, 
    \code{skip} \,\mid\, \smash{x^\intt} := Z \,\mid\, \smash{x^\real} := E  
    \\
    & \phantom{\;::=\;}
    \,\mid\,
    C_1;C_2 
    \,\mid\,
    \code{ifz} \ Z \ C_1 \ C_2
    \,\mid\,   
    \code{for} \ \smash{x^\intt} \ \code{in} \ \code{range}(n)\ \code{do}\ C \quad \text{(for $n \geq 1$)}  
  \end{align*}
  \\
  \begin{flushleft}
    \textbf{Semantics:}
  \end{flushleft}
  \vspace{-1em}
  \begin{align*}
    \db{n}\sigma & = n
    \qquad\qquad\qquad
    \db{\smash{x^\intt}}\sigma = \sigma(\smash{x^\intt})
    \\
    \db{\op_i(Z_1,\ldots,Z_k,E_1,\ldots,E_m)}\sigma & = \db{\op_i}_a(\db{Z_1}\sigma,\ldots,\db{Z_k}\sigma,\db{E_1}\sigma,\ldots,\db{E_m}\sigma)
    \\
    \db{r}\sigma & = r
    \qquad\qquad\qquad
    \db{\smash{x^\real}}\sigma =  \sigma(\smash{x^\real})
    \\
    \db{\op_r(Z_1,\ldots,Z_k,E_1,\ldots,E_m)}\sigma & = \db{\op_r}_a(\db{Z_1}\sigma,\ldots,\db{Z_k}\sigma,\db{E_1}\sigma,\ldots,\db{E_m}\sigma)
    \\
    \db{[(\alpha_1,Z_1);\ldots;(\alpha_k, Z_k)]}\sigma & = [(\alpha_1,\db{Z_1}\sigma);\,\ldots;\,(\alpha_k,\db{Z_k}\sigma)]
  \end{align*}
  \vspace{0.2ex}
  \(
  \begin{array}{c}
    \infer{
      (x^\real := \code{fetch}(I)), D, \sigma \Downarrows \sigma[x^\real \mapsto D(\db{I}\sigma)],0.0
    }{}
    \qquad
    \raisebox{0.26ex}{
    \infer{
      \code{score}(E), D, \sigma \Downarrows \sigma, \db{E}\sigma
    }{}
    }
    \\[1.25ex]
    \infer{
      (x^\intt := Z), D, \sigma \Downarrows \sigma[x^\intt \mapsto \db{Z}\sigma], 0.0
    }{}
    \qquad
    \infer{
      (x^\real := E), D, \sigma \Downarrows \sigma[x^\real \mapsto \db{E}\sigma], 0.0
    }{}
    \\[1.25ex]
    \infer{
      (C_1;C_2), D, \sigma \Downarrows \sigma'', (r_1 + r_2)
    }{
      C_1, D, \sigma \Downarrows \sigma', r_1 & C_2, D, \sigma' \Downarrows \sigma'', r_2
    }
    \qquad
    \infer{
      (\code{ifz}\ Z\ C_1\ C_2), D, \sigma \Downarrows \sigma',r
    }{
      \db{Z}\sigma = 0 & C_1, D, \sigma \Downarrows \sigma',r
    }
    \qquad
    \infer{
      (\code{ifz}\ Z\ C_1\ C_2), D, \sigma \Downarrows \sigma',r
    }{
      \db{Z}\sigma \neq 0 & C_2, D, \sigma \Downarrows \sigma',r
    }
    \\[1.25ex]
    \infer{
      \code{skip}, D, \sigma \Downarrows \sigma, 0.0
    }{}
    \qquad
    \infer{
      (\code{for}\ x\ \code{in}\ \code{range}(n)\ \code{do}\ C), D, \sigma \Downarrows \sigma_n, \sum_{k = 1}^{n} r_k
    }{
      \sigma_0 = \sigma & C, D, \sigma_k[x \mapsto k] \Downarrows \sigma_{k+1},r_{k+1} \text{ for all $k \in \{0,\ldots,n-1\}$}
    }
  \end{array}
  \)
  \vspace{-0.5em}
  \caption{Syntax and semantics of the score-computing language.}
  \label{f:synsem}
  \vspace{-1em}
\end{figure}

Let $\Str$, $\Z$, $\R$, $\Var$ be the sets of strings, integers, real numbers, and variables, respectively.
We use the following symbols 
to range over elements in these sets: $\alpha,\beta \in \Str$, $n,m \in \Z$, $r \in \R$, and $x,y,z \in \Var$.
Each variable $x \in \Var$ has a type $\tau \in \{\intt,\, \real\}$. We often write $x^\tau$ to denote that $x$ has type $\tau$.

The syntax of the score-computing language is given in \Cref{f:synsem}, based on the above notations.
The language has three types of expressions, which denote
state-dependent integers, reals, and \emph{indices}.
By indices, we mean finite sequences of string-integer pairs
that do not use the same string more than once,
that is, 
elements $i$ in the following set:
\[
i \in \Index = \big\{[(\alpha_1,m_1);\ldots;(\alpha_k,m_k)] \in (\Str \times \Z)^* \,:\, 
\text{$\alpha_j$'s are distinct for $j \in \{1, \ldots, k\}$, $k \ge 0$}
\big\}.
\]
For instance, when evaluated at a state, an index expression
$[(\alpha_1,Z_1);\ldots;(\alpha_k, Z_k)]$ returns
the index $[(\alpha_1,m_1);\ldots;(\alpha_k,m_k)]$, i.e., 
the length-$k$ sequence of string-integer pairs $(\alpha_j,m_j)$,
where $m_j$
is the integer value of $Z_j$ at the state. These indices serve as identifiers
of random variables in the language. When a command in the language needs
the value of a random variable, it uses the index for the variable as a key and
looks up an entry at this key in a table, called \emph{random database},
which stores the values of all the random variables.
The grammars for the expressions are standard,
except for the use of $\op_i$ and $\op_r$, the former being an integer-returning 
operation with appropriate input types and the latter being a real-returning operation.

Commands in the language denote variants of state transformers, which take 
a random database and a starting state as inputs, and return a final state
and a real number called \emph{total score} as outputs.
A {random database} $D$ here refers to a map from all indices
(i.e., finite sequences of string-integer pairs with no repetition of a string) to real numbers, i.e., 
\(
        D \in \RDB = [\Index \to \R].
\)
It specifies the values of random variables (identified by the indices). Thus, besides updating the state, a command $C$ computes the total score of the given random
database. When $C$ represents a probabilistic model, the computed score 
is the logarithm of an unnormalised density of the model.

The command 
$\code{score}(E)$ evaluates $E$ in the current state, and returns the result as the total score,
while keeping the state unchanged. 
The next command $x^\real := \code{fetch}(I)$
reads the value at the index $I$ in the random database, stores the read value 
in  $x^\real$, and returns the resulting state together with zero
score. The other commands are the standard constructs of the imperative languages:
$\code{skip}$, assignments, sequencing, conditionals, and for-loops. Regarding
the state transformation, they have the standard meanings, except that
the conditionals use an implicit zero check on their $Z$ argument, 
so that 
their true branch $C_1$ gets executed if $Z$ evaluates to zero, and the false branch $C_2$ is
followed otherwise. Regarding the score computation, these commands implement 
the following two principles: (i) when a command $C$ is run, it may execute the $\code{score}$ 
command multiple times, and the results of these multiple invocations of $\code{score}$ get 
added to become the final total score of $C$; (ii) none of the other atomic commands, 
such as assignments and $\code{skip}$, contributes to the total score. 
Note that the language is restricted in the sense that it does not have general while-loops. 
This restriction ensures that all commands in the language always terminate,
and this termination property is used when we prove the correctness of our optimisation.
We point out that despite this restriction, the language is expressive enough to allow
a wide range of typical probabilistic models, as indicated by the eight models used in our 
experiments. Lifting this restriction is, however, an interesting and important future work. 
We write $\Comm$ for the set of all commands in the language.

\subsection{Semantics}
\label{ssec:src-semantics}

Let $\State$ be the set of states defined by the type-preserving maps from
variables to integers or reals:
\[ 
\State = \{\sigma \in [\Var \to \Z \cup \R] \,:\, 
\sigma(x^\intt) \in \Z\ \text{and}\ \sigma(y^\real) \in \R\ 
\text{for all variables $x^\intt$ and $y^\real$}\}.
\]

We interpret expressions as maps from states to the values of appropriate types:
$\db{Z} : \State \to \Z$, 
$\db{E} : \State \to \R$, and 
$\db{I} : \State \to \Index$.
We assume that the semantics of integer-returning or real-returning atomic operations 
in the language are given, i.e., we have 
$\db{\op_i}_a : \Z^k \times \R^m \to \Z$ and 
$\db{\op_r}_a : \Z^k \times \R^m \to \R$ for
all $\op_i$ and $\op_r$ taking $k$ integer and $m$ real arguments. Using 
this assumed semantics of atomic operations, we interpret expressions in 
the standard way (\Cref{f:synsem}).

We define the meanings of commands using a 
big-step operational semantics, which specifies 
the relation ${\Downarrows} \subseteq (\Comm \times \RDB \times \State) \times (\State \times \R)$.
Intuitively, $(C,D,\sigma \Downarrows \sigma', r)$ means that running $C$ on 
the state $\sigma$ under
the random database $D$ terminates successfully, producing the state $\sigma'$ and the total score $r$. 
The score $r$ is the logarithm of the unnormalised density 
at $D$ according to a model that $C$ implements. Since the language does not
have general while-loops, and no commands get stuck or generate an error, 
every tuple $(C,D,\sigma)$ of command, random database, and input state leads
to a unique pair $(\sigma',r)$ made of an output state and a total score.
We define $\Downarrows$ in the standard manner using the inference-rule
notation (\Cref{f:synsem}).
\section{Loop Vectorisation}
\label{sec:target-lang}

Our optimisation operates on commands in the language of \Cref{sec:source-lang} 
by vectorising for-loops in those commands. In this section, we describe
this loop vectorisation as a translation from the language to another language that
supports vectorised operations
while continuing to have the primitives for describing unnormalised densities.
We informally explain and formally define this \emph{target language}
of the translation, where all the atomic commands are implicitly vectorised
and new language constructs control various aspects of vectorisation
(\Cref{subsec:overview-target-language,subsec:target-syntax}).
The target language is closely related to the popular tensor-based programming
languages, such as Python with PyTorch tensors, which our implementation is based on (see \Cref{sec:impl}).
After defining the target language, we describe the translation
from the language of \Cref{sec:source-lang} to the target language, and explain 
the use of speculative computations and fixed-point check in the translated commands, 
which we mentioned in the introduction (\Cref{subsec:target-translation}).
In the rest of the paper, we refer to the language of \Cref{sec:source-lang} as the \emph{source language}.

\subsection{Overview of a Target Language for Loop Vectorisation}
\label{subsec:overview-target-language}

The target language is an imperative language with a restricted form of recursion, 
and includes all the constructs from the source language.
Its two main features are as follows.

First, variables and expressions in the target language have \emph{lifted types}
rather than the original types in the source language,
where lifting means changing a type~$\tau$ to the type of maps  
from $\Index$ to $\tau$. 
Concretely, variables in the target language store (finite representations of) maps 
from indices to integers or reals, instead of just integers or reals. Similarly, 
expressions in the target language evaluate to maps from indices to integers, reals, or indices.

Second, commands in the target language cannot use indices explicitly 
in order to access entries in these maps. Instead, they 
maintain a \emph{set of indices} 
implicitly, and these indices are then used when expressions and variables are 
accessed during execution. Note that an index
set, instead of a single index, is maintained here. This means that 
a single execution of a command in the target language
corresponds to multiple standard executions of the same command, one per each index 
in the index set, where these executions access different entries of the 
lifted variables using their indices. This 
execution model of the target language is similar to the  SIMD 
(single instruction multiple data) 
model, and can also be understood as an executable special case 
of the  collecting semantics in the abstract interpretation literature. 

To help the reader to gain a high-level intuition about how this target language works, 
we consider two commands in the language: $C_0$ and $C_1$ 
in \Cref{fig:programs-target-language}.

\begin{figure}[t]
  \small
  \centering
  \[
  \begin{array}{c}
      C_0 =
      \left[
      \begin{array}{l}        
        x^{\lreal} := \code{fetch}([(\rvarstr{z},t^{\lintt})]); \\
        \code{score}(\log p_\cN(x^{\lreal};f(y^{\lreal}),1.0)); \\
        \code{score}(\log p_\mcN(o(t^\lintt); g(x^\lreal), 1.0)); \\
        y^{\lreal} := x^{\lreal} 
      \end{array}
      \right],
      \qquad
      C_1 =
      \left[
      \begin{array}{l}
        \code{extend\_index}(\rvarstr{vec},10)\; \{ \\
        \quad t^{\lintt} := \code{lookup\_index}(\rvarstr{vec}); \\
        \quad C_0\; \}
      \end{array}
      \right]
    \end{array}
    \] 
    \vspace{-1.3em}
  \caption{Example commands in the target language.}
  \label{fig:programs-target-language}
    \vspace{-0.8em}
\end{figure}

\begin{example}[$C_0$ in \Cref{fig:programs-target-language}]
\label{example:targetlang-C0}
The command $C_0$ corresponds to a variant of the loop body of $\mathit{HMM}$ in \Cref{fig:programs-source-language},
where all occurrences of variables are annotated with their types.

\paragraph{Lifted types and index sets.}
The annotations in $C_0$ use the \emph{lifted types}, $\lreal$ and $\lintt$, instead of $\real$ and $\intt$;
these lifted types indicate that the variables store not real numbers or integers,
but {maps} to them.
For instance,  the 
variable $t$ can store the following map from indices to integers:
\begin{gather}
  \label{eqn:examples-target-language:map-x}
  m_t = 
  \lambda i \in \Index.\ 
    \text{if }
    \big(i = [(\rvarstr{vec},\ell)] \concat i' 
         \text{ for some } \ell \in \{0,1,\ldots,9\}
         \text{ and } i'\big) 
    \text{ then } \ell 
    \text{ else } 0.
\end{gather}
Here $j \concat j'$ means the concatenation of indices $j$ and $j'$ 
which are viewed as finite sequences of string-integer pairs.
For this $m_t$, we have $m_t([(\rvarstr{vec}, \ell)]) = \ell$ for all $\ell \in \{0, \ldots, 9\}$.

Some entries of the maps stored in the variables $t,x,y$ are read and updated during the execution of $C_0$. 
But which entries are read or updated at each line of $C_0$ is not specified explicitly. Instead, it 
is determined implicitly by an \emph{index set} that is given
in the beginning of the execution.

Concretely, assume that $C_0$ is run 
under the index set $A_\mathit{vec} = \{i_0,i_1,\ldots,i_9\}$ for $i_\ell = [(\rvarstr{vec},\ell)]$,
as well as a random database $D$ and a starting state $\sigma$ of the target language such that $D([(\rvarstr{z}, \ell)]) = z_\ell$,
$\sigma(t) = m_t$,
$\sigma(x) = m_x$, and 
$\sigma(y) =  m_y$.
Here $z_\ell$ is the value of the random variable $\rvar{z}_\ell$ of $\mathit{HMM}$ 
that we used in \Cref{sec:source-lang},
$m_t$ is the map in \Cref{eqn:examples-target-language:map-x}, 
and $m_x, m_y$ are the maps defined by $m_x(i) = m_y(i) = 0.0$ for all $i$. 
In this setup, $C_0$ is run as follows.

\paragraph{First command.}
The first command $x^{\lreal} := \code{fetch}([(\rvarstr{z},t^{\lintt})])$ reads the ten entries 
of the map $m_t$ at $i_0,\ldots,i_9 \in A_\mathit{vec}$, and forms the ten keys
$[(\rvarstr{z},m_t(i_0))] = [(\rvarstr{z},0)]$, $\ldots$ , $[(\rvarstr{z},m_t(i_9))] = [(\rvarstr{z},9)]$ 
for the random database $D$.
Then, it accesses $D$ using these keys, 
and updates the map $m_x$ with the results of these accesses: 
for each index $i_\ell \in A_\mathit{vec}$, the value
$D([(\rvarstr{z},\ell)]) = z_\ell$
is copied to all the entries of $m_x$ whose indices have   
the form $i_\ell \concat i'$ for some $i'$. Note that the update here is \emph{non-standard} in that
the copy of the read value $z_\ell$
is not restricted to the $i_\ell$-th entry of $m_x$. 
More precisely, the resulting state is
\begin{equation}
  \label{eqn:examples-target-language:state-after-first-command}
  \sigma[x \mapsto m'_x]\;\;\; \text{where}\ 
  m'_x(i) = 
  \text{if }
  \big(i = i_\ell \concat i' 
    \text{ for some } \ell \in \{0,\ldots,9\} 
    \text{ and } i'\big)
  \text{ then } z_\ell
  \text{ else } m_x(i).
\end{equation}
As we will explain in \Cref{sec:target-semantics-translation-soundness}, this 
non-standard variable update models \emph{tensor broadcasting} in PyTorch and other tensor libraries, 
and is one of the factors that let us have an effective tensor-based implementation of maps stored in variables.
Except for the use of this update, 
$\sigma[x\mapsto m'_x]$ is 
precisely the state that we would obtain by 
vectorising the corresponding $\code{fetch}$ command in the source language over $A_\mathit{vec}$ in the usual sense, i.e.,
running the command component-wise over a vector of size ten whose components are accessed
by indices $i_0,\ldots,i_9$.

The first command $x^{\lreal} := \code{fetch}([(\rvarstr{z},t^{\lintt})])$ also returns scores 
for the ten indices in $A_\mathit{vec}$, which happen to be all zero in this case. 
The returned 
scores are packaged into the constant zero map from $A_\mathit{vec}$ to reals,
written as $[i_0 \mapsto 0.0, \ldots, i_9 \mapsto 0.0]$.

\paragraph{Remaining commands.}
The remaining three commands 
are executed in a similarly $A_\mathit{vec}$-vectorised way,
and give rise to the following final state $\sigma'$ and score map $T' : A_\mathit{vec} \to \bbR$:
\begin{align}
  \label{eqn:examples-target-language:final-state}
  \sigma' & {} = \sigma[x \mapsto m'_x, y \mapsto m'_x], &
  T'(i_\ell) 
  & {} = \log p_\cN(z_\ell;f(0.0),1.0) + \log p_\mcN(o(\ell); g(z_\ell), 1.0),
\end{align}
where $m_x'$ is the map in \Cref{eqn:examples-target-language:state-after-first-command}.
These outputs correspond to the combined outcomes of all the ten iterations
of the loop in $\mathit{HMM}$ in the source language,
where for each $\ell \in \{0,\ldots,9\}$,
the $\ell$-th iteration is run after $t$ is set to $\ell$ and both $x$ and $y$ are set to $0.0$.

Note that the sum of the scores in $T'$ differs from $r_\mathit{HMM}$,
the total score computed by $\mathit{HMM}$ in the source language
with random variables $\rvar{z}_\ell$ fixed to the values $z_\ell$ (see \Cref{eqn:hmm-logpdf}):
\begin{align}
\nonumber
\textstyle
\sum_{\ell = 0}^9 T'(i_\ell)
& {} =
\textstyle
\sum_{\ell = 0}^9 \Bigl(\log p_\cN(z_\ell;f(\boxed{0.0}),1.0) + \log p_\mcN(o(\ell); g(z_\ell), 1.0)\Bigr)
\\
\label{eqn:rhmm}
& {} \neq
\textstyle
\sum_{\ell = 0}^9 \Bigl(\log p_\cN(z_\ell;f(\boxed{z_{\ell - 1}}),1.0) + \log p_\mcN(o(\ell); g(z_\ell), 1.0)\Bigr)
= r_\mathit{HMM},
\end{align}
where $z_{-1} = 0.0$.
The different parts are highlighted with boxes. They originate from the fact
that the execution of $C_0$ under $A_\mathit{vec}$ is \emph{speculative} about the starting value of the variable $y$;
at each index $i_\ell \in A_\mathit{vec}$, the $i_\ell$ part of the execution uses $0.0$ as the starting value of $y$,
instead of $z_{\ell -1}$ used by $\mathit{HMM}$. This type of speculation over starting values of variables
at each index is exploited in our loop vectorisation, as we will explain in \Cref{subsec:target-translation}.
\qed
\end{example}

\begin{example}[$C_1$ in \Cref{fig:programs-target-language}]
\label{example:targetlang-C1}
The command $C_1$ extends $C_0$ by adding the $\code{lookup\_index}$ command 
and then wrapping the whole with the $\code{extend\_index}$ construct.
The $\code{extend\_index}$ construct \emph{locally changes} the index set given
in the beginning of the execution, and in so doing, it alters the level of vectorisation of commands inside
its scope temporarily.

Concretely, assume that $C_1$ is executed under the index set
$A_{[]} = \{ [] \}$ containing only the empty index (i.e., empty sequence) and  
the starting state $\sigma_1$ where $\sigma_1(t) = (\lambda i.\,0)$ and 
$\sigma_1(x) = \sigma_1(y) = (\lambda i.\,0.0)$ (i.e., $t,x,y$ store constant-zero maps modulo return types).
Also, assume that the execution is given the same random database $D$ as in \Cref{example:targetlang-C0}.
In this setup, $C_1$ is run as follows.

\paragraph{Entry phase and body of extend\_index.}
The $\code{extend\_index}$ construct in $C_1$ first \emph{extends} the index $[]$ in $A_{[]}$ 
by adding the string-integer pair $(\rvarstr{vec},\ell)$ at the end for every $\ell \in \{0,\ldots,9\}$.
The updated index set is $A_\mathit{vec} = \{i_0,\ldots,i_9\}$ from \Cref{example:targetlang-C0}.

The $\code{extend\_index}$ construct then executes the commands in its scope under $A_\mathit{vec}$, $D$, and $\sigma_1$.
The first command in the scope is 
$t^{\lintt} := \code{lookup\_index}(\rvarstr{vec})$. For every index $i_\ell \in A_\mathit{vec}$,
it finds the value paired with $\rvarstr{vec}$ in $i_\ell$, and copies the value to 
all the entries of the map $\sigma_1(t)$ whose indices are of the form $i_\ell \concat i'$ for some $i'$.
That is, the resulting state is $\sigma_1[t \mapsto m_t]$ where $m_t$ is the map in \Cref{eqn:examples-target-language:map-x},
so that the state is the same as the starting state $\sigma$ of $C_0$ from \Cref{example:targetlang-C0}.
Regarding the score map, 
the command does not perform anything special, and simply returns the zero 
map $[i_0 \mapsto 0.0, \ldots, i_9 \mapsto 0.0]$.
The remaining command in the scope is $C_0$. It is run under $A_\mathit{vec}$, $D$, and $\sigma$,
and results in the final state $\sigma'$ and the score map $T'$ from \Cref{example:targetlang-C0}
(see \cref{eqn:examples-target-language:final-state}).

\paragraph{Exit phase of extend\_index.}
The $\code{extend\_index}$ construct finally \emph{restores} the original index set $A_{[]}$
so that the subsequent computation is run under $A_{[]}$.
Moreover, the construct \emph{post-processes} $\sigma'$ and $T'$ to $\sigma''$ and $T''$, which are then returned 
as the final state and score map.
Intuitively, $\sigma''$ is obtained by making the result at the last index $i_9 \in A_\mathit{vec}$ 
the final outcome of the entire computation, and $T''$ is obtained by summing over all the entries in $A_\mathit{vec}$. 
Concretely, for every variable $w$, the postprocessing builds the new map $m''_w = \sigma''(w)$ from 
the old map $m_w' = \sigma'(w)$, by reading the $i_9$-th entry of $m_w'$ and 
storing the read value at every index:
$m''_w = (\lambda i.\, m'_w(i_9))$.
The score map $T''$ is 
the singleton function mapping the unique index $[] \in A_{[]}$ to $\sum_{\ell = 0}^9 T'(i_\ell)$.
This post-processing comes from our loop vectorisation; when the vectorised parallel execution of
all the loop iterations is completed,
we need to pick only the final state computed by the last iteration
but we have to accumulate the scores computed by all the iterations. 
The details will appear in \Cref{subsec:target-translation}.
\qed
\end{example}

\subsection{Syntax of the Target Language}
\label{subsec:target-syntax}

\begin{figure}[t]
  \begin{align*}
    \mbox{\it Lifted Integer Expr.}\ \ Z 
    & \;::=\; n \,\mid\, x^\lintt \,\mid\, \op_i(Z_1,\, \ldots,\, Z_k,\,E_1,\ldots,\, E_m)
    \\
    \mbox{\it Lifted Real Expr.}\ \ E
    & \;::=\; r \,\mid\, \smash{x^\lreal} \,\mid\, \op_r(Z_1,\, \ldots,\, Z_k,\,E_1,\, \ldots,\, E_m)
    \\
    \mbox{\it Lifted Index Expr.}\ \ I 
    & \;::=\; [(\alpha_1, Z_1);\, \ldots;\, (\alpha_k, Z_k)] 
    \quad \text{(for distinct $\alpha_i$'s)}
    \\
    \mbox{\it Commands}\ \ C 
    & \;::=\; \code{score}(E) \,\mid\, \smash{x^\lreal} := \code{fetch}(I) \,\mid\, 
    \code{skip} \,\mid\, \smash{x^\lintt} := Z \,\mid\, \smash{x^\lreal} := E \,\mid\,
    C_1;C_2  
    \\
    & \phantom{\;::=\;}
    \,\mid\,   
    \code{for} \ \smash{x^\lintt} \ \code{in} \ \code{range}(n_+)\ \code{do}\ C
    \,\mid\,
    \code{ifz} \ Z \ C_1 \ C_2
    \,\mid\, 
    \code{loop\_fixpt\_noacc}(n_+)\ C 
    \\
    & \phantom{\;::=\;}
    \,\mid\,    
    \code{extend\_index}(\alpha,n_+)\ C
    \,\mid\,
    \smash{x^{\lintt}} := \code{lookup\_index}(\alpha)
    \,\mid\, 
    \code{shift}(\alpha)
  \end{align*}
\vspace{-1.8em}
\caption{Syntax of the target language. $n,m$ are non-negative integers, $n_+$ a positive integer,
  and $\alpha$ a string.
}
\label{fig:syntax-target-language}
\vspace{-1em}
\end{figure}

\Cref{fig:syntax-target-language} shows the syntax of the target language, 
which is obtained from that of the source language by annotating variables 
with lifted types $\lreal$ and $\lintt$, instead of $\real$ and $\intt$,
and including four new constructs, one for a special type of looping and three for controlling vectorisation.  
Commands in the target language operate on states with the lifted variables,
i.e., variables that store maps from indices to reals or integers. As input, they
receive a random database $D$, a state $\sigma$ of lifted variables, and 
a finite set $A$ of indices.
Then, they output a final state $\sigma'$ and a score map $T'$ that assigns
a score to each index $i \in A$. We write 
$\Comm_\mathit{tgt}$ for the set of commands in the target language.

Let us go through the cases of a command $C$ in the grammar and explain its behaviour
assuming that $C$ is run under $(D,\sigma,A)$.

\paragraph{Existing commands.}
If the command $C$ is one of the first five cases in the grammar, i.e., an atomic command that also appears in the source language (modulo type annotations of variables), then
its execution in the target language corresponds to the \emph{vectorised parallel execution} of the same command in 
the source language over indices in $A$. For each index $i \in A$, $C$ is run according
to the semantics of the source language, except that all variable
reads are now on the $i$ entries of the maps stored in those variables,
and similarly variable updates are on the $i$-related entries; precisely, all the 
entries at the indices $i\concat i''$ for some $i''$ are updated.
The final score map of $C$ is the function $T'$
that maps each $i \in A$ to the total score computed by the execution under $i$.

If the command $C$ is the sequential
composition $(C_1;C_2)$, then $C_1$ is run first under given $(D, \sigma, A)$ and $C_2$ is 
run afterwards  
under $(D,\sigma'_1,A)$ where $\sigma'_1$ is the final state of $C_1$. The final 
state $\sigma'$ of the entire $C$ is the one
from the execution of $C_2$,
and the final score map $T'$ of $C$ is the \emph{pointwise sum} of the maps $T'_1$ and $T'_2$ 
resulting from the executions of $C_1$ and $C_2$. 
If the command $C$ is $(\code{for}\ x^\lintt\ \code{in}\ \code{range}(n_+)\ \code{do}\ C)$,
then $C$ behaves the same as its loop-unrolled version. 

If the command $C$ is the conditional $(\code{ifz}\ Z\ C_1\ C_2)$,
then its execution starts by 
\emph{partitioning} the index set $A$ into $A_1$ and $A_2$ such that $A_1$ contains 
all the indices $i \in A$ under which $Z$ evaluates to zero, and $A_2$ contains 
the remaining indices. Then, $C_1$ is run under $(D,\sigma,A_1)$,
and $C_2$ is run under $(D,\sigma'_1,A_2)$ afterwards where $\sigma'_1$ is the final 
state of $C_1$.
The final state $\sigma'$ of the conditional command $C$ is the one from 
the execution of $C_2$, and the final 
score map $T'$ of $C$ is 
the following combination of $T'_1$ and $T'_2$ resulting from the runs 
of $C_1$ and $C_2$:
$T'(i) = T'_1(i)$ for all $i \in A_1$ and $T'(i) = T'_2(i)$ for all $i \in A_2$. 
If the reader is familiar with abstract interpretation, she or he may have noticed that
our conditional command behaves according to (a finite executable version of) 
the standard collecting semantics in abstract interpretation.

\paragraph{New commands.}
The remaining four cases are new.
The first case enables a fine-grained control over \emph{looping and score management},
and the other cases enable that over \emph{vectorised parallel execution}.

If the command $C$ is $(\code{loop\_fixpt\_noacc}(n_+)\ C_1)$, it behaves
the same as the sequential composition of $n_+$-many copies of $C_1$'s, except for two points. First, the execution of $C$ skips some copies of $C_1$ and stops early 
if its \emph{fixed-point check} gives a positive answer. Concretely, after the run
of each copy of $C_1$, the execution checks whether the state has been unchanged
during the run of this copy, that is, whether the state is a fixed point of
the run of the copy. If the check says yes, the execution stops, skipping
the remaining copies of $C_1$.
This early stopping based on fixed-point check
leads to the performance improvement in our optimisation. 
Second, the execution of $C$ does not 
follow the \emph{score accumulation scheme} used by the sequential composition, 
but it instead keeps the score of the last executed copy of $C_1$ only, 
ignoring the scores computed by all the other executed copies of $C_1$. This unusual treatment 
of scores is closely tied to the use of speculation in our optimisation; ignoring amounts to discarding the
results of certain unconfirmed speculative computations. 

If the command $C$ is $(\code{extend\_index}(\alpha,n_+)\ C_1)$,
it \emph{extends} every index $i$ in $A$ by adding $(\alpha,k)$ at the end for all $k \in \{0,\ldots,n_+{-}1\}$;
recall that $C$ is assumed to be run under $(D,\sigma,A)$.
This extension may fail if $i \concat [(\alpha,k)]$
for some $i \in A$ and $k \in \{0,\ldots,n_+{-}1\}$ violates the definition of $\Index$
because the string $\alpha$ is already present in $i$.
If this happens, the execution of $C$ gets stuck. 
Otherwise, we have a new index set $A'$ of size $|A| \times n$ consisting of
the extended indices, and $C_1$ is then run under $(D,\sigma,A')$. 
When the execution of $C_1$ is completed, 
the original index set $A$ is \emph{restored}, and the final state $\sigma_1'$ and score map $T_1'$ of $C_1$
are \emph{post-processed} to yield the final results, $\sigma'$ and $T'$, for $C$. 
Here $T'$ is obtained from $T'_1$ via summation: $T'(i) = \sum_{k = 0}^{n_{+}{-}1} 
T'_1(i \concat [(\alpha,k)])$ 
for every $i \in A$. The state $\sigma'$ is obtained by going
through every variable $x$ and updating the map $\sigma'_1(x)$ stored in 
$x$ as follows: 
$\sigma'(x)(i') = \sigma_1'(x)(i \concat [(\alpha,n_+{-}1)])$ if $i' = i \concat i''$
for some $i \in A$ and $i''$, and $\sigma'(x)(i') = \sigma_1'(x)(i')$ otherwise. Intuitively,
this post-processing step ensures that for all $i \in A$, $\sigma'$ keeps only the result of the
last $(i \concat [(\alpha,n_+{-}1)])$ execution of $C_1$,
while $T'$ keeps the scores computed by all  
the $(i \concat [(\alpha,\ell)])$ executions
of $C_1$ for $\ell \in \{0,\ldots,n_+{-}1\}$.

If the command $C$ is
$(x^{\lintt} := \code{lookup\_index}(\alpha))$, it first  
checks whether every index $i$ in the set $A$ includes a pair $(\alpha,\ell_i)$ 
for some $\ell_i \in \Z$. If the answer is no, the execution of $C$ gets stuck.
Otherwise, for every $i \in A$, $C$ stores $\ell_i$ at the $(i\concat i'')$ entry of the map of $x$, 
for all $i''$. This construct allows 
the $i$ execution to depend on information \emph{specific} to $i$, 
thereby enabling different components of vectorised execution to 
perform different computations.

If the command $C$ is $\code{shift}(\alpha)$, it 
updates all the maps stored in the variables
by \emph{moving} values inside each map in a manner analogous to the bitwise-shift operation.
More precisely, assume that $C$ is run under $(D,\sigma,A)$ and $A = \{i \concat [(\alpha,k)]
: i \in A_0,\, k \in \{0,\ldots,n_{+}{-}1\}\}$ for some $n_+ \geq 1$ and $A_0$
such that $\alpha$ is not already present in every index $i \in A_0$.
Then, $\code{shift}(\alpha)$ results
in the state $\sigma'$ by 
going through every variable $x$ and updating the map $\sigma(x)$ stored in 
$x$ as follows: 
\[
\sigma'(x)(i') 
=
\begin{cases}
\sigma(x)(i \concat [(\alpha,k{-}1)])
& \text{if $i' = i \concat [(\alpha,k)] \concat i''$ for some $k \in \{1,\ldots,n_{+}{-}1\}$ and $i''$},
\\
\sigma(x)(i) 
& \text{if $i' = i \concat [(\alpha,0)] \concat i''$ for some $i''$},
\\
\sigma(x)(i')
& \text{otherwise}.
\end{cases}
\]
A good way to understand this definition is to ignore the $i''$ part, and to focus on where the
entries of $\sigma'$
come from:
the $(i \concat [(\alpha,k)])$ entry is obtained from the $(i \concat [(\alpha,k-1)])$ entry 
for all $k \in \{1,\ldots,n_{+}{-}1\}$, and 
the $(i\concat [(\alpha,0)])$ entry is obtained from the $i$ entry. 
This $\code{shift}$ command lets different components of a vectorised parallel
execution communicate with each other by moving values between them,
and enables us to vectorise loops that have dependencies between iterations.

\subsection{Translation}
\label{subsec:target-translation}

\begin{figure}
\begin{align*}
  & \begin{aligned}[t]
  \overline{x^\real := \code{fetch}(I)} & {}\;\equiv\; x^{\lreal} := \code{fetch}(\overline{I}), \\[-1pt]
  \overline{\code{score}(E)} & {}\;\equiv\; \code{score}(\overline{E}), \\[-1pt]
  \overline{x^\intt := Z} & {}\;\equiv\; x^{\lintt} := \overline{Z}, \\[-1pt]
  \overline{x^\real := E} & {}\;\equiv\; x^{\lreal} := \overline{E}, \\[-1pt]
  \overline{\code{skip}} & {}\;\equiv\; \code{skip}, \\[-1pt]
  \overline{C_1; C_2} & {}\;\equiv\; \overline{C_1}; \overline{C_2}, \\[-1pt]
  \overline{\code{ifz}\ Z\ C_1\ C_2} & {}\;\equiv\; \code{ifz}\ \overline{Z}\ \overline{C_1}\ \overline{C_2},
  \end{aligned} &
  & \begin{aligned}[t]
    & \overline{\code{for}\ x^\intt \ \code{in} \ \code{range}(n)\ \code{do}\ C_1} \\
      & \;\; \begin{aligned}[t]
      & {}{\;\equiv\;\;} \code{extend\_index}(\alpha,n)\ \{ \\
      & {}\phantom{\;\equiv\;\;} \quad \code{loop\_fixpt\_noacc}(n)\ \{ \\
      & {}\phantom{\;\equiv\;\;} \quad \quad \code{shift}(\alpha); \\
      & {}\phantom{\;\equiv\;\;} \quad \quad x^{\intt\dagger} := \code{lookup\_index}(\alpha); \\
      & {}\phantom{\;\equiv\;\;} \quad \quad \overline{C_1}\;\; \} \}
    \end{aligned} \\
    & \text{where $\alpha$ is a string not appearing in $\overline{C_1}$}.
  \end{aligned}
\end{align*}
\vspace{-1.1em}
\caption{Translation of commands in the target language.}
\label{fig:translation-commands}
\vspace{-0.9em}
\end{figure}

Our loop vectorisation is implemented as a translation from the source language to the target language.
The translation assumes that these languages use the 
same set of symbols for variables, and for every variable $x$ in that set, 
they give the corresponding types, i.e., $x$ has a type $\tau$ in the source 
language if and only if it has the type $\tau\dagger$ in the target language.
Using this assumption, the translation converts expressions of all three 
types in the source language to themselves except that variables $x^\tau$
in those expressions are replaced by $x^{\tau\dagger}$. We put a line on
top of expressions in the source language to denote their translations. Thus,
$\overline{E}$,
$\overline{Z}$, and $\overline{I}$
mean the translated $E$, $Z$, and $I$.

The translation of commands is defined inductively by the rules in \Cref{fig:translation-commands}. 
The most important rule is the one for translating a for-loop 
and thereby implementing the main part of our loop vectorisation;
the other rules simply apply the translation inductively.

Consider the translation of 
$C \equiv (\code{for}\ x^\intt \ \code{in} \ \code{range}(n)\ \code{do}\ C_1)$ in the figure.
The translated $\overline{C}$ is a $\code{loop\_fixpt\_noacc}$ loop run under
the scope of the $\code{extend\_index}(\alpha,n)$ construct. A single iteration
of this loop corresponds to the 
\emph{vectorised parallel execution} of the body $C_1$ of the original for-loop over all the indices in $\code{range}(n)$.
The translation implements this parallel execution using the $\code{extend\_index}(\alpha,n)$ construct, 
which changes the input index set $A$ to $A' = \{ i \concat [(\alpha,k)] : i \in A,\, k \in \{0,\ldots,n{-}1\}\}$
and increases the level of parallelism by the factor of $n$ since $|A|\times n = |A'|$; 
the increased parallelism is then used to run all the $n$ iterations of the loop body $C_1$ simultaneously.

The parallel execution above is \emph{speculative} in that the execution of the 
translated loop body $\overline{C_1}$ under an index $(i \concat [(\alpha,k{+}1)]) \in A'$
does not wait until the result of the previous iteration under the index $(i \concat [(\alpha,k)])$ is available---%
it just proceeds with a predicted result of the previous iteration.

To preserve the semantics of 
the original for-loop despite the speculative execution, the translation of $C$ \emph{repeats} the parallel execution up to $n$ 
times using the $\code{loop\_fixpt\_noacc}(n)$ construct, 
and forces \emph{communication} between different threads of the parallel execution (which
represent computations under different indices in $A'$) using the $\code{shift}(\alpha)$ command.
Each round of the repetition starts by running $\code{shift}(\alpha)$, which moves values between different 
threads 
such that the $(i\concat [(\alpha,k{+}1)])$-indexed thread for the $(k{+}1)$-th 
iteration of the loop body $\overline{C_1}$ in this round receives, at the start of execution, 
the result of the $(i\concat [(\alpha,k)])$-indexed thread for the $k$-th iteration 
from the previous round. Moving values this way gradually eliminates the speculative nature of the parallel execution:
at the end of the $j$-th round, the threads for the $0,\ldots,j$-th iterations of 
$\overline{C_1}$ are guaranteed to start with correct results for the previous iterations so that they are no longer speculative.

Another important aspect of the translation of $C$ is that the $\code{loop\_fixpt\_noacc}$ loop
is stopped early if the state is not changed by the previous round of the loop, 
which is detected by the \emph{fixed-point check} of the $\code{loop\_fixpt\_noacc}$ construct.
This early stopping occurs often in practice, and it is crucial for the performance 
improvement of our loop vectorisation.

Finally, we point out that \emph{only the last round} of 
the $\code{loop\_fixpt\_noacc}$ loop contributes to the final score map $T'$ computed by the loop;
the score maps computed by all the other rounds of the loop are simply discarded. Using 
the score map of the last round only comes from the fact that only the parallel execution 
of the last round is guaranteed to be the non-speculative parallel execution of all
the iterations of the original for-loop in the source language. 

\begin{figure}[t]
  \small
  \centering
  \[
  \begin{array}{c}
      C_2 =
      \left[
      \begin{array}{l}
      \code{extend\_index}(\rvarstr{vec},10)\; \{ \\
      \quad \code{loop\_fixpt\_noacc}(10)\; \{ \\
      \quad\quad \code{shift}(\rvarstr{vec});
      \; t^{\lintt} := \code{lookup\_index}(\rvarstr{vec});
      \; C_0\; \}\}.
      \end{array}
      \right]
    \end{array}
    \] 
    \vspace{-1.2em}
    \caption{Example translation of $\mathit{HMM}$ in \Cref{fig:programs-source-language}
      to the target language.
      \rwl{The command $C_0$ was defined in \cref{fig:programs-target-language}}.}
  \label{fig:programs-target-language-2}
  \vspace{-1em}
\end{figure}

\begin{example}[$C_2$ in \Cref{fig:programs-target-language-2}]
Recall the $\mathit{HMM}$ example 
from \Cref{fig:programs-source-language}. Translating $\mathit{HMM}$ to the target language
gives the command $C_2$ in \Cref{fig:programs-target-language-2}. In the 
rest of this subsection, we will 
go through the key steps of the execution of $C_2$, and explain how
$C_2$ uses \emph{speculative execution} and \emph{fixed-point check}
to vectorise the loop in $\mathit{HMM}$ even though
the loop iterations have data dependencies between them.

Let $D$ be a random database as defined in \Cref{example:targetlang-C0}.
Also, let $A_{[]} = \{[]\}$ be the singleton set of the empty index, 
and $\sigma_0$ be the state where every variable stores the constant-zero function. 
Suppose that $C_2$ is run with $(D,\sigma_0,A_{[]})$. Then, 
the $\code{loop\_fixpt\_noacc}$ loop inside $C_2$ performs two iterations, 
and $C_2$ produces the state $\sigma'$ and the score map $T': A_{[]} \to \R$ given by
\begin{align*}
& \sigma'(t) = \lambda i.\, 9,\qquad
\sigma'(x) = \sigma'(y) = \lambda i.\, z_9,\qquad
\sigma'(v) = \sigma_0(v) \text{ for all other variables } v,\\[-3pt]
& \textstyle
T'([]) = \sum_{\ell=0}^{9} \Bigl(\log p_\cN(z_\ell; f(z_{\ell-1}), 1.0) + \log p_\mcN(o(\ell); g(z_\ell), 1.0)\Bigr),
\end{align*}
where $z_{-1} = 0.0$. At the index $[]$, this outcome matches that 
of $\mathit{HMM}$ in the source language. Specifically, 
$T'([])$ equals to $r_\mathit{HMM}$ from \Cref{eqn:rhmm}, and for every variable $u$, 
the value $\sigma'(u)([])$ equals to the value of $u$ 
in the final state of $\mathit{HMM}$.
We now examine the execution of $C_2$ in detail.

\paragraph{First round of the loop.}
The execution of $C_2$ begins by 
extending indices in $A_{[]}$ with $(\rvarstr{vec},\ell)$ for
all $\ell \in \{0,\ldots,9\}$,
and updating the current index set to that of the extended indices,
i.e., $A_\mathit{vec} = \{i_0,\ldots,i_9\}$ for $i_\ell = [(\rvarstr{vec},\ell)]$. 
Then, it runs the body
of the loop of $C_2$ with $(D,\sigma_0,A_\mathit{vec})$. 
The first command in the loop body is $\code{shift}(\rvarstr{vec})$,
and it acts as no-op (or $\code{skip}$) this time, because all the variables 
in $\sigma_0$ store constant-zero maps and so moving values between entries
\rhl{does} not change those maps.  
The next command looks up the value bound to $\rvarstr{vec}$ in 
each index $i$ in $A_\mathit{vec}$, stores the looked-up value 
at the $i$-related entries of the map of $t$ (i.e., entries at
indices $i \concat i'$ for all $i'$), and returns  
the following state $\hat{\sigma}_0$:
\[
\hat{\sigma}_0 = \sigma_0[t \mapsto m_t]
\ \ \text{where}\
m_t(i) = 
\text{if } \big(i = i_\ell \concat i' \text{ for some } \ell \in \{0,\ldots,9\} \text{ and } i'\big)
\text{ then } \ell
\text{ else } \sigma_0(t)(i).
\] 
The rest of the loop body executes equivalently to $C_0$ in \Cref{example:targetlang-C0}, producing the same outcome:
a state $\sigma'_0$ and a score map $T'_0 : A_\mathit{vec} \to \R$ given by
\begin{align*}
\sigma'_0 & {} = \sigma_0[t \mapsto m_t, x \mapsto m'_0, y \mapsto m'_0], 
\quad
m'_0(i) {} = 
\begin{cases}
  z_\ell & \text{if $i = i_\ell \concat i'$ for some $\ell \in \{0,\ldots,9\}$ and $i'$},
  \\
  0.0 & \text{otherwise},
\end{cases} 
\\
T'_0(i_\ell) 
& {} = \log p_\cN(z_\ell;f(0.0),1.0) + \log p_\mcN(o(\ell); g(z_\ell), 1.0)
\quad \text{for all $\ell \in \{0,\ldots,9\}$}.
\end{align*}

Note that for all $i_\ell \in A_\mathit{vec}$, the $i_\ell$ part of the execution of the loop body 
corresponds to the $\ell$-th iteration of the original for-loop of $\mathit{HMM}$, 
but this correspondence is not perfect due to the \emph{incorrect speculation} used by the former: the $i_\ell$ 
part of the execution speculates that in the original for-loop of $\mathit{HMM}$, the variable $y$ stores 
$\sigma_0(y)(i_{\ell}) = 0.0$ after the $(\ell{-}1)$-th iteration, but if $\ell \geq 1$ and $z_{\ell-1} \neq 0$, 
the speculation is wrong because the correct value of $y$ is $z_{\ell - 1}$. As a result of incorrect speculation
and imperfect correspondence, $T'(i_\ell)$ is not equal to the score of the $\ell$-th iteration of $\mathit{HMM}$ in general; the former uses $f(0.0)$ instead of $f(z_{\ell-1})$.

The incorrect speculation is 
detected by the \emph{fixed-point check} of the $\code{loop\_fixpt\_noacc}$ construct of $C_2$, and it is corrected
also by the \emph{re-execution} of the body of the $\code{loop\_fixpt\_noacc}$ loop with improved speculation. 
Concretely, the fixed-point check of $\code{loop\_fixpt\_noacc}$ tests whether 
the state $\sigma'_0$ is the same as the initial state $\sigma_0$, and gets a negative answer,
which indicates that the speculated value of $y$ at some iteration of the for-loop of $\mathit{HMM}$ is wrong. 
Because of the negative answer, the body of $\code{loop\_fixpt\_noacc}$ is executed again but this time speculation is 
made based on the current state $\sigma'_0$, instead of the old state $\sigma_0$. Using the more recent state
for speculation always leads to improvement in terms of correctness.

\paragraph{Subsequent rounds.}
In our example, 
the speculation made at the $j$-th round of $\code{loop\_fixpt\_noacc}$ is guaranteed to be 
correct for the $k$-th iteration of the original for-loop of $\mathit{HMM}$ for all $k \leq j$
where correctness means that the speculated value of $y$ at the beginning of the $k$-th iteration equals the actual
value of $y$ at that iteration of the for-loop of $\mathit{HMM}$. In the example, 
we have more than this guaranteed improvement, because all the speculated values of $y$ at the second round of 
the $\code{loop\_fixpt\_noacc}$ loop are correct. The use of the \emph{correct speculation} implies that 
the state $\sigma'_1$ after the second round  equals
$\sigma'_0$, and so the fixed-point check of $\code{loop\_fixpt\_noacc}$ succeeds this time 
and the $\code{loop\_fixpt\_noacc}$ loop terminates after this second round.

The final state $\sigma'$ 
and the score map $T'$ of $C_2$ are then constructed from $\sigma'_1$ and the score map $T'_1$ computed 
by the second round, according to the semantics of $\code{extend\_index}(\rvarstr{vec},10)$:
\begin{align*}
\sigma' & {} = \sigma'_0[t \mapsto m'_t,\, x \mapsto m',\, y \mapsto m'],
\qquad
m'_t = \lambda i.\,9,
\qquad
m' = \lambda i.\, z_9,
\\
T' &
\textstyle
{} = \left[[] \mapsto \sum_{\ell=0}^{9}\Bigl(\log p_\cN(z_\ell; f(z_{\ell-1}), 1.0) + \log p_\mcN(o(\ell); g(z_\ell), 1.0)\Bigr)\right].
\end{align*} 
Note that the final $\sigma'$ and $T'$ of $C_2$ at the index $[]$ are the same as the results of $\mathit{HMM}$ 
in \Cref{fig:programs-source-language}.
\qed
\end{example}

\section{Semantics of the Target Language and Soundness of Translation}
\label{sec:target-semantics-translation-soundness}

In this section, we present an operational semantics of the target language, formalising
the intuitive description of the language in \Cref{subsec:overview-target-language,subsec:target-syntax}
(\Cref{ssec:semantics-target-language}). 
We then prove that the translation from the source language to the target language,
described in \Cref{subsec:target-translation}, preserves the semantics (\Cref{ssec:target-translation-soundness}).

\subsection{Semantics of the Target Language}
\label{ssec:semantics-target-language}

\paragraph{\rwl{Definition of states}}
The semantics relies on the standard prefix order $\sqsubseteq$ 
on indices: for all $i,j \in \Index$,
we have $i \sqsubseteq j$ if and only if  
$|i| \leq |j|$ and $i.k = j.k$ for all $k \in \{1,\ldots,|i|\}$. Here
 $|i|$ means the length of $i$, and $i.k$ means the $k$-th element of $i$.
We denote the upward and downward closures of an index $i$ and an index set $L$ using the following notations:
$i{\uparrow} = \{j \in \Index \,:\, i \sqsubseteq j\}$,
$i{\downarrow} = \{j \in \Index \,:\, j \sqsubseteq i\}$,
$L{\uparrow} = \bigcup_{j \in L} j{\uparrow}$, and
$L{\downarrow} = \bigcup_{j \in L} j{\downarrow}$.
Also, we let
\[ 
\max(L) =
\text{if }
\big(i \text{ is the $\sqsubseteq$-largest element of $L$}\big)
\text{ then } i
\text{ else undefined}. 
\] 
Thus, $i = \max(L)$ if and only if $i \in L$ and for all $i' \in L$, $i' \sqsubseteq i$. 
Note that 
if $L \cap j{\downarrow} \neq \emptyset$ for some index $j$,   
then $\max(L \cap j{\downarrow})$ always exists, because $L \cap j{\downarrow}$ is totally ordered by
$\sqsubseteq$.

For each type $\tau \in \{\intt, \real\}$, let $\db{\intt} = \Z$ and $\db{\real} = \R$.
A {\bf state} $\sigma$ in the target language is a function sending variables to 
finite partial maps from indices to integers or reals such that
for every variable $x^{\tau\dagger}$, 
(i) the range of the partial map $m_x = \sigma(x)$ 
is $\db{\tau}$, and (ii) the domain of $m_x$ contains the empty index $[]$.
A recommended reading is to view the partial map $m_x$ here as
a data structure representing the following total function $\extend(m_x)$:
\begin{equation}
\label{eqn:definition-extend}
\extend(m_x) : \Index \to \db{\tau},
\quad
\extend(m_x)(i) = m_x\big(\max(\dom(m_x) \cap i{\downarrow})\big)
\ \text{for all $i \in \Index$}.
\end{equation}
Note that for every $i \in \Index$, the set $(\dom(m_x) \cap i{\downarrow})$ is nonempty because it 
contains the empty index $[]$ at least, and the set is also totally ordered
because $i{\downarrow}$ is totally ordered. As a result, for all $i$, $\max(\dom(m_x) \cap i{\downarrow})$ is 
a well-defined element in $\dom(m_x)$, which in turn implies that $\extend(m_x)$ is a well-defined total
map. Although the semantics of the target language to be presented shortly is defined in terms of 
finite partial maps stored in variables in states, it depends only on the total functions represented
by those partial maps: for all states $\sigma,\sigma'$, if $\extend(\sigma(x)) = \extend(\sigma'(x))$ for 
all variables $x$, the semantics is unable to distinguish between $\sigma$ and $\sigma'$. 
When we explain the semantics in the rest of the paper, we will often refer to 
the total function $\extend(\sigma(x))$ as the function stored in a variable $x$, although $\extend(\sigma(x))$ is
not stored in $x$.
We write $\State_\mathit{tgt}$ for the set of all the states in 
the target language.  

\rhl{Our $\extend$ operation models {\em tensor broadcasting}~\citep{Walt:11} as implemented in popular tensor libraries such as PyTorch.
To see this, first note that the operation can be applied to any finite partial map $m : \Index \rightharpoonup V$ for a set $V$, even when the empty index $[]$ is not in $\dom(m)$. The result $\extend(m)$ is still well-defined in this general case, although it might not be a total map. Next, as an example, consider two integer-type tensors $t_1$ and $t_2$ with shapes $(3)$ and $(2,3)$, respectively. In our setup, they are represented as partial maps $m_1, m_2 : \Index \rightharpoonup \mathbb{Z}$ with
$\dom(m_1) = \{[(\mathbf{snd},\ell_2)] : \ell_2 \in \{0,1,2\}\}$ and
$\dom(m_2) = \{[(\mathbf{snd},\ell_2);(\mathbf{fst},\ell_1)] : \ell_2 \in \{0,1,2\}, \ell_1 \in \{0,1\}\}$.
Although $t_1$ and $t_2$ have different shapes, they can be added using tensor broadcasting: the addition first broadcasts $t_1$ into a tensor of shape $(2,3)$, and then performs pointwise addition with $t_2$.
In our setup, this broadcasting is modelled by applying $\extend$ to $m_1$, which yields a partial map defined at all indices in $\dom(m_2)$, with each value corresponding to the broadcasted entries of $t_1$. The pointwise addition is then modelled by computing the elementwise sum of $\extend(m_1)$ and $m_2$ at all indices in $\dom(m_2)$. Lastly, note that $\extend(m_1)$ models not only this specific instance of broadcasting $t_1$ with respect to $t_2$, but also all possible instances of broadcasting $t_1$ \mbox{with respect to any tensor $t_2'$ of a larger ``compatible'' shape.}}

\paragraph{\rwl{Semantics of target language}}
The semantics uses the following five semantic domains and their structures:
$\Index$, 
$\State_\mathit{tgt}$,
$\FAntiChain$, 
$\Tensor_\Z$, and  
$\Tensor_\R$.
We have already defined the first two  in this list. The next is
$\FAntiChain$, the domain for finite antichains of indices. A set of indices $L$ 
is an {\bf antichain} if for all $i,j \in L$
with $i \neq j$, we have neither $i \sqsubseteq j$ nor $j \sqsubseteq i$.
The domain $\FAntiChain$ consists of all the {\bf finite antichains} of indices.
When denoting an element of $\FAntiChain$, 
we often use the symbol $A$, instead of the usual symbol $L$ for an index set, 
in order to emphasise that it is not just 
any index set, but an antichain with a finite number of elements. The last two domains are
instances of the set of all {\bf $V$-tensors} for
a set $V$, which are finite partial maps
from $\Index$ to $V$. This $V$-tensor set is written as $\Tensor_V$. We use the symbol $T$
to denote an element of $\Tensor_V$.

\begin{figure}[t]
\begin{align*}
  \db{n}(\sigma, i) & = n, 
  \qquad\qquad
  \db{x^{\lintt}}(\sigma,i) = \extend(\sigma(x^{\lintt}))(i),
  \\
  \db{\op_i(Z_1, \ldots, Z_k,E_1,\ldots,E_m)}(\sigma,i) & =
  \db{\op_i}_a(
      \begin{aligned}[t]
          & \db{Z_1}(\sigma,i),\ldots,\db{Z_k}(\sigma,i),
            \db{E_1}(\sigma,i),\ldots,\db{E_m}(\sigma,i)),
      \end{aligned}
  \\
  \db{r}(\sigma, i) & = r, \qquad\qquad
  \db{x^{\lreal}}(\sigma,i) = \extend(\sigma(x^{\lreal}))(i),
  \\
  \db{\op_r(Z_1, \ldots, Z_k,E_1,\ldots,E_m)}(\sigma,i) & =
  \db{\op_r}_a(
      \begin{aligned}[t]
          & \db{Z_1}(\sigma,i),\ldots,\db{Z_k}(\sigma,i),
            \db{E_1}(\sigma,i),\ldots,\db{E_m}(\sigma,i)),
      \end{aligned}
  \\
  \db{[(\alpha_1,Z_1); \ldots; (\alpha_k, Z_k)]}(\sigma,i) 
  & =
      [(\alpha_1,\db{Z_1}(\sigma,i));\ldots;(\alpha_k,\db{Z_k}(\sigma,i))]. 
\end{align*}
\vspace{-1.8em}
\caption{Semantics of expressions in the target language.}
\label{fig:semantics-expressions-target-language}
\vspace{-0.7em}
\end{figure}

The interpretations of expressions in the target language have the following forms: $\db{Z} : \State_\mathit{tgt} \times \Index \to \Z$,
$\db{E} : \State_\mathit{tgt} \times \Index \to \R$,
and $\db{I} : \State_\mathit{tgt} \times \Index \to \Index$.
The defining clauses for these interpretations appear in \Cref{fig:semantics-expressions-target-language}.
They assume that the interpretations of the integer-returning atomic operations $\op_i$  
and the real-returning atomic operations $\op_r$ are given by $\db{\op_i}_a$ and $\db{\op_r}_a$, 
respectively. The clauses say that an expression is evaluated under a state $\sigma$ and an index $i$
in a standard way using the assumed $\db{\op_i}_a$ and $\db{\op_r}_a$ 
except that all the variable reads in the expression access the $i$ entities of
the represented functions for those variables.

The semantics of commands in the target language is 
specified in terms of a big-step computation relation 
${\Downarrowt}
    \subseteq 
(\Comm_\mathit{tgt} \times \RDB \times \State_\mathit{tgt} \times \FAntiChain) \times (\State_\mathit{tgt} \times \Tensor_\R)$.
The relation $\Downarrowt$ works on states in $\State_\mathit{tgt}$ and takes, as its inputs, 
a command $C$ to run, a random database $D$, a state $\sigma$, and a finite antichain $A$ of indices. 
Intuitively, the indices in $A$ correspond to the ids of the threads that run the same
$C$ in parallel. The id of a thread determines, for each variable, which entries of
the map stored in the variable are read and updated during the 
execution of $C$ by the thread. 

\paragraph{\rwl{Operations on states}}
\rwl{The rules for the relation $\Downarrowt$} use the below notations on
tensors and states. For a set of indices $L$, we write $T^z_L$ for 
the tensor in $\Tensor_\R$ 
that maps every index in $L$ to $0.0$ and is undefined for all other indices.
Also, we write $\oplus$ for the following pointwise addition on $\Tensor_\R$:
\begin{align*}
\oplus & : \Tensor_\R \times \Tensor_\R \to \Tensor_\R,
&
(T \oplus T')(i) & 
=
\begin{cases}
T(i) + T'(i) & \text{if $i \in \dom(T) \cap \dom(T')$},
\\[-1pt]
T(i) & \text{if $i \in \dom(T) \setminus \dom(T')$},
\\[-1pt]
T'(i) & \text{if $i \in \dom(T') \setminus \dom(T)$},
\\[-1pt]
\text{undefined} & \text{otherwise}.
\end{cases}
\end{align*}
Regarding states, the clauses in the figure use the update operation $\sigma[x^{\tau\dagger}: T]$
for a variable $x^{\tau\dagger}$ and a tensor $T \in \Tensor_{\db{\tau}}$, and 
the copy operation $\sigma\langle\rho\rangle$ for a finite partial injective map $\rho$
on indices. These operations result in states $\smash{\sigma[x^{\tau\dagger} : T]}$
and $\sigma\langle\rho\rangle$ in $\State_\mathit{tgt}$ defined as follows:
\begin{align*}
\sigma[x^{\tau\dagger} : T](y^{\tau'\dagger})(i) & =
\begin{cases}
    T(i) & \text{if $y^{\tau'\dagger} \equiv x^{\tau\dagger}$ and $i \in \dom(T)$},
    \\[-1pt]
    \text{undefined} & \text{if $y^{\tau'\dagger} \equiv x^{\tau\dagger}$ and $i \in \dom(T){\uparrow} \setminus \dom(T)$},
    \\[-1pt]
    \sigma(y^{\tau'\dagger})(i) & \text{otherwise},
\end{cases}
\\
\sigma\langle\rho\rangle(y^{\tau'\dagger})(i) & =
\begin{cases}
    \sigma(y^{\tau'\dagger})(\rho^{-1}(i)) & \text{if $i \in \image(\rho)$},
    \\[-1pt]
    \text{undefined} & \text{if $i \in \image(\rho){\uparrow} \setminus \image(\rho)$},
    \\[-1pt]
    \sigma(y^{\tau'\dagger})(i) & \text{otherwise}.
\end{cases}
\end{align*}
If a partial map on the right-hand side is applied to an index that is not in its domain,
the result is undefined, and we do not write this explicitly in the equations.
\rhl{\Cref{fig:tensor-update-copy} shows how update and copy operations work for simple partial maps.}

\usetikzlibrary{arrows, arrows.meta,decorations.pathreplacing, shapes.arrows}
\newcommand{\minitensor}[9]{
    \resizebox{\textwidth}{!}{
        \centering
\scalebox{0.95}{
        \begin{tikzpicture}[
        dot/.style={draw, circle, dotted, fill=gray!30, minimum size=17, font=\small, line width=0.8pt},
        sol/.style={draw, circle, solid, fill=gray!30, minimum size=17, font=\small},
        ]
            \node[font=\small] at (0,1.15) {#9};
            \node[#1] at (0,0.5) {#2};
            \draw [decorate,decoration = {brace,amplitude=6pt,raise=-2pt}] (-1.3,0) --  (1.3,0);
            \node[#3] at (-1,-0.4) {#4};
            \node[#5] at (0,-0.4) {#6};
            \node[#7] at (1,-0.4) {#8};
        \end{tikzpicture}
}
    }
}
\newcommand{\shifttensor}[1]{
    \resizebox{\textwidth}{!}{
        \centering
\scalebox{0.95}{
        \begin{tikzpicture}[
        null/.style={draw, circle, dotted, fill=gray!0, minimum size=17, font=\small, line width=0.8pt},
        ]
            \node[font=\small] at (0,1.15) {#1};
            \node[null] at (0,0.5) (0) {};
            \draw [decorate,decoration = {brace,amplitude=6pt,raise=-2pt}] (-1.3,0) --  (1.3,0);
            \node[null] at (-1,-0.4) (1) {};
            \node[null] at (0,-0.4) (2) {};
            \node[null] at (1,-0.4) (3) {};

            \draw[-{Latex[length=5]}] (0) -- (1);
            \draw[-{Latex[length=5]}] (1) -- (2);
            \draw[-{Latex[length=5]}] (2) -- (3); 
        \end{tikzpicture}
}
    }
}

\newcommand{\shrinktensor}[1]{
    \resizebox{\textwidth}{!}{
        \centering
\scalebox{0.95}{
        \begin{tikzpicture}[
        null/.style={draw, circle, dotted, fill=gray!0, minimum size=17, font=\small, line width=0.8pt},
        ]
            \node[font=\small] at (0,1.15) {#1};
            \node[null] at (0,0.5) (0) {};
            \draw [decorate,decoration = {brace,amplitude=6pt,raise=-2pt}] (-1.3,0) --  (1.3,0);
            \node[null] at (-1,-0.4) (1) {};
            \node[null] at (0,-0.4) (2) {};
            \node[null] at (1,-0.4) (3) {};

            \draw[-{Latex[length=5]}] (3) -- (0);
        \end{tikzpicture}
}
    }
}

\newcommand{\tikzbigdownarrow}{
  \resizebox{0.5\textwidth}{!}{
    \centering
 \scalebox{0.95}{
       \begin{tikzpicture}
          \node[single arrow, draw, shape border rotate=270, fill = black, single arrow tip angle=150, single arrow head extend=3, minimum width = 80, minimum height =15] at (0,0) {};
        \end{tikzpicture}
}
  }
}

\begin{figure}[t]
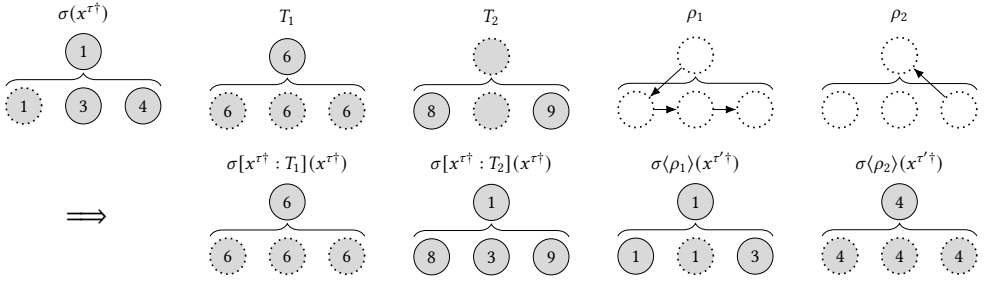

  \rhl{
  \centering
  \begin{minipage}{0.17\textwidth}
    \centering
  \begin{subfigure}{\textwidth}
    \centering
    \minitensor{sol}{1}{dot}{1}{sol}{3}{sol}{4}{$\sigma(x^{\tau\dagger})$}

    \vspace{13pt}
    {\large ~$\implies$}

    \vspace{32pt}
    ~
  \end{subfigure}
  \end{minipage}
  \hspace{0.5em}
  \begin{minipage}{0.17\textwidth}
    \centering
    \begin{subfigure}{\textwidth}
      \centering
      \minitensor{sol}{6}{dot}{6}{dot}{6}{dot}{6}{$T_1$}

      \vspace{-7pt}
      \minitensor{sol}{6}{dot}{6}{dot}{6}{dot}{6}{$\sigma[x^{\tau\dagger}:T_1](x^{\tau\dagger})$}
      \vspace{1em}
    \end{subfigure}
  \end{minipage}
  \hspace{0.5em}
  \begin{minipage}{0.17\textwidth}
    \centering
    \begin{subfigure}{\textwidth}
      \centering
      \minitensor{dot}{ }{sol}{8}{dot}{ }{sol}{9}{$T_2$}

      \vspace{-7pt}
      \minitensor{sol}{1}{sol}{8}{sol}{3}{sol}{9}{$\sigma[x^{\tau\dagger}:T_2](x^{\tau\dagger})$}
      \vspace{1em}
    \end{subfigure}
  \end{minipage}
  \hspace{0.5em}
  \begin{minipage}{0.17\textwidth}
    \centering
    \begin{subfigure}{\textwidth}
      \centering
      \shifttensor{$\rho_1$}

      \vspace{-7pt}
      \minitensor{sol}{1}{sol}{1}{dot}{1}{sol}{3}{$\sigma\langle\rho_1\rangle(x^{\tau'\dagger})$}
      \vspace{1em}
    \end{subfigure}
  \end{minipage}
  \hspace{0.5em}
  \begin{minipage}{0.17\textwidth}
    \centering
    \begin{subfigure}{\textwidth}
      \centering
      \shrinktensor{$\rho_2$}

      \vspace{-7pt}
      \minitensor{sol}{4}{dot}{4}{dot}{4}{dot}{4}{$\sigma\langle\rho_2\rangle(x^{\tau'\dagger})$}
      \vspace{1em}
    \end{subfigure}
  \end{minipage}
  }
  \vspace{-3.0em}
  \caption{\rhl{Examples for update and copy.
  Nodes over braces represent the empty indices, and nodes under braces represent indices $i_\ell = [(\rvarstr{rv},\ell)]$ for $\ell \in \{0,1,2\}$.
  Values in solid nodes are stored in partial maps, while values in dotted nodes are induced by $\extend$.
  Thus, $\sigma(x^{\tau\dagger}) = [[] \mapsto 1, i_1 \mapsto 3, i_2 \mapsto 4]$, $T_1 = [[] \mapsto 6]$, and $T_2 = [i_0 \mapsto 8, i_2 \mapsto 9]$.
  Also, $\extend(\sigma(x^{\tau\dagger}))(i_0) = 1$, and $\extend(T_1)(i_\ell) = 6$ for $\ell \in \{0,1,2\}$.
  The functions $\rho_1$ and $\rho_2$ denote the partial injections on $\Index$
  used in the semantics for $\code{shift}$ and $\code{extend\_index}$ (\cref{fig:semantics-commands-target-language}),
  respectively, and are depicted graphically in the top right;
  e.g., $\rho_2$ is defined only on $i_2$ and maps it to $[]$.}}
  \label{fig:tensor-update-copy}
  \vspace{-0.7em}
\end{figure}

A good way to understand these slightly unusual update and copy operations on states is to 
recall that finite partial maps stored in 
variables represent total functions obtained by the $\extend$ operation,
and examine what the operations do on these represented functions. 
For a function $f : \Index \to V$ and a partial function $g : \Index \rightharpoonup V$,
let $f[g]$ be the outcome of updating $f$ with $g$, that is, the function mapping 
each index $i \in \dom(g)$ to $g(i)$ and every other index $i$ to its original value $f(i)$.
The next lemma characterises the effects of our update and copy operations by means of standard
update and copy operations on the represented functions.
Its proof is in \Cref{subsec:couplings}.
\begin{lemma}
  \label{lem:update-renaming-extend}
  Let $\sigma$ be a state, $x^{\tau\dagger}$, $y^{\tau'\dagger}$ be variables, 
  $T$ be a tensor in $\Tensor_{\db{\tau}}$, and $\rho$ be a finite partial injection on $\Index$. Then,
  we have the following equations:
  \begin{align*}
  \extend(\sigma[x^{\tau\dagger} : T](y^{\tau'\dagger}))
  & = 
  \text{\rm if } \big(y^{\tau'\dagger} \equiv x^{\tau\dagger}\big)
  \text{\rm\ then } \extend(\sigma(x^{\tau\dagger}))[\extend(T)]
  \text{\rm\ else } \extend(\sigma(y^{\tau'\dagger})),
  \\[-3pt]
  \extend(\sigma\langle\rho\rangle(y^{\tau'\dagger}))
  & = \extend(\sigma(y^{\tau'\dagger}))[\extend(T_{y,\rho})], 
\end{align*}
where
  $T_{y,\rho}(i) = \sigma(y^{\tau'\dagger})(\rho^{-1}(i))$ if $i \in \rho(\dom(\sigma(y^{\tau'\dagger})))$,
  and it is undefined otherwise.
\end{lemma}
\noindent
Here we apply the $\extend$ operation in \Cref{eqn:definition-extend} to the tensor $T$, although
$T$ may be undefined at the empty index unlike in the case of the partial maps stored in states. As a result,
this application may lead to a partial function of type $\Index \rightharpoonup \db{\tau}$, 
which we recall below:
\[
\extend(T)(i) = 
\text{if } \big(i \in \dom(T){\uparrow}\big)
\text{ then } T(\max(\dom(T) \cap i{\downarrow}))  
\text{ else undefined}.
\]

Note that the equations in the lemma use ${(\_)}[\extend(T)]$ and 
${(\_)}[\extend(T_{y,\rho})]$, instead of ${(\_)}[T]$ and ${(\_)}[T_{y,\rho}]$. This means that 
both operations change more entries than one typically expects. For instance, in the first equation,  
each $(i,T(i))$ entry of the partial map $T$ results in 
the changes of all the entries in $\extend(\sigma(x^{\tau\dagger}))$
whose indices are in $i{\uparrow}$.
We use this aggressive version of change 
in our operations because it enables us to implement the operations and our target language easily 
and efficiently using the tensor-broadcasting mechanism of PyTorch.

This aggressive change also helps maintain 
the following important invariant of our target language.
For an index set $L$, we say 
a state $\sigma$ is an {\bf $L$-state} if 
$\dom(\sigma(x^{\tau\dagger})) \subseteq L{\downarrow}$
for all variables $x^{\tau\dagger}$.
Intuitively, it
means that for every variable $x^{\tau\dagger}$, 
if the partial map $\sigma(x^{\tau\dagger})$ stored in $x^{\tau\dagger}$ represents a function $f = \extend(\sigma(x^{\tau\dagger}))$,
then the function is determined by what it does on ${L{\downarrow}}$, i.e., $f = \extend(f|_{L{\downarrow}})$.
The invariant is that if a command $C$ is executed under a finite antichain $A$
and an input state is an $A$-state, then the output state is also an $A$-state.
This invariant supports the view that the command $C$ is run in parallel by the threads with their ids in $A$ and the state consists of the disjoint parts for these threads. 
Formal lemmas and proofs are detailed in \Cref{subsec:parallelising-for-loops}.

\begin{figure}[t]
\newcommand{\vspacesem}{1.9ex}
\vspace{-2pt}
\[
\begin{array}{c}
\infer{
    (x^\lreal := \code{fetch}(I)), D, \sigma, A \Downarrowt \sigma[x^\lreal : T], T^{\tzero}_A
}{
    T = [i \mapsto D(\db{I}(\sigma,i)) : i \in A]
}
\qquad
\infer{
    \code{score}(E), D, \sigma, A \Downarrowt \sigma, T
}{
    T = [i \mapsto \db{E}(\sigma,i) : i \in A]
}
\\[\vspacesem]
\infer{
    (x^\lintt := Z), D, \sigma, A \Downarrowt \sigma[x^\lintt : T], T^z_A
}{
    T = [i \mapsto \db{Z}(\sigma,i) : i \in A]
}
\qquad
\infer{
    (x^\lreal := E), D, \sigma, A \Downarrowt \sigma[x^\lreal : T], T^z_A
}{
    T = [i \mapsto \db{E}(\sigma,i) : i \in A]
}
\\[\vspacesem]
\infer{
    \code{skip}, D, \sigma, A \Downarrowt \sigma, T^z_A
}{}
\qquad
\infer{
    (C_1;C_2), D, \sigma, A \Downarrowt \sigma'', (T' \oplus T'')
}{
    C_1, D, \sigma, A \Downarrowt \sigma', T'
    &
    C_2, D, \sigma', A \Downarrowt \sigma'', T''
}
\\[\vspacesem]
\infer{
    (\code{ifz}\ Z\ C_1\ C_2), D, \sigma, A \Downarrowt \sigma'', (T' \oplus T'')
}{
    \begin{array}{ll}
    A_1 = \{i \in A : \db{Z}(\sigma,i) = 0\}
    & 
    C_1, D, \sigma, A_1 \Downarrowt \sigma', T'
    \\[0.0ex]
    A_2 = \{i \in A : \db{Z}(\sigma,i) \neq 0\}
    &
    C_2, D, \sigma', A_2 \Downarrowt \sigma'', T''
    \end{array}
}
\\[\vspacesem]
\infer{
    (\code{for}\ x^\lintt \ \code{in} \ \code{range}(n_+)\ \code{do}\ C), D, \sigma, A \Downarrowt \sigma_{n_+},
    \bigoplus_{k = 1}^{n_+} T_{k}
}{
    \sigma_0 = \sigma
    \qquad
    C, D, \sigma_k[x^\lintt : [i \mapsto k \, :\, i \in A]], A \Downarrowt \sigma_{k+1}, T_{k+1}
    \text{ for all $k \in \{0,{...}\,,n_+{-}1\}$}
}
\\[\vspacesem]
\infer{
    (x^{\lintt} := \code{lookup\_index}(\alpha)), D, \sigma, A \Downarrowt \sigma[x^{\lintt} : T], T^z_A
}{
    \alpha \in \dom(i) \text{ for all $i \in A$} 
    \qquad
    T = [i \mapsto i(\alpha) \,:\, i \in A]  
}
\\[\vspacesem]
\infer{
    (\code{extend\_index}(\alpha,n_+)\ C), D, \sigma, A \Downarrowt \sigma'\langle\rho\rangle, T''
}{
    \begin{array}{@{}l@{}}
    \alpha \not\in \dom(i) \text{ for all $i \,{\in}\, A$}
    \quad
    A' = \{ i \,{\concat}\, [(\alpha,k)] : i \,{\in}\, A,\, k \,{\in}\, \{0,{...}\,,n_+{-}1\} \}
    \quad
    C, D, \sigma, A' \Downarrowt \sigma', T'
    \\[0.0ex]
    \rho = [i \concat [(\alpha,n_+{-}1)] \mapsto i : i \in A ]
    \qquad
    T'' = [i \mapsto \sum_{0 \leq k < n_+} T'(i \concat [(\alpha,k)]) : i \in A]     
    \end{array}
}
\\[\vspacesem]
\infer{
  (\code{loop\_fixpt\_noacc}(n_+)\ C), D, \sigma, A 
  \Downarrowt
  \sigma_m, T_{m{+}1}
}{
  \begin{array}{l}
  0 \leq m < n_+ {-} 1
  \qquad
  \sigma_0 = \sigma
  \hfill {}
  \\[0.0ex]
  C, D, \sigma_k, A \Downarrowt \sigma_{k{+}1}, T_{k{+}1}
  \text{ and }
  \sigma_k \neq \sigma_{k{+}1}
  \text{ for all $k \in \{0,{...}\,,m{-}1\}$}
  \qquad
  C, D, \sigma_m, A \Downarrowt \sigma_m, T_{m{+}1}
  \end{array}
}
\\[\vspacesem]
\infer{
  (\code{loop\_fixpt\_noacc}(n_+)\ C), D, \sigma, A 
  \Downarrowt
  \sigma_{n_+}, T_{n_+}
}{
  \sigma_0 = \sigma
  &
  C, D, \sigma_k, A \Downarrowt \sigma_{k{+}1}, T_{k{+}1}
  \text{ for all $k \in \{0,{...}\,,{n_+}{-}1\}$}
  &
  \sigma_k \neq \sigma_{k{+}1}
  \text{ for all $k \in \{0,{...}\,,{n_+}{-}2\}$}
}
\\[\vspacesem]
\infer{
  \code{shift}(\alpha), D, \sigma, A \Downarrowt \sigma\langle\rho\rangle, T^{z}_A
}{
    \rho(i) =
    \begin{cases}
      i \concat [(\alpha,0)] & 
      \text{if } \alpha \not\in \dom(i) 
      \land i \concat [(\alpha,0)] \in A, 
      \\[-2pt]  
      j \concat [(\alpha,m{+}1)] &
      \text{if } i = j \concat [(\alpha,m)]  \in A{\downarrow}
      \land
      j \concat [(\alpha,m{+}1)] \in A
      \text{ for some $j$ and $m \geq 0$},
      \\[-2pt]
      \text{undefined} & \text{otherwise}
    \end{cases}
}
\end{array}
\]
\vspace{-1.2em}
\caption{Semantics of commands in the target language.}
\label{fig:semantics-commands-target-language}
\vspace{-0.6em}
\end{figure}

\paragraph{\rwl{Semantics of commands}}
The rules for the $\Downarrowt$ relation in \Cref{fig:semantics-commands-target-language} mostly follow 
the general principle of running each command over multiple input data specified
by the input antichain $A$,
except for two unusual parts.
First, updating $x^{\tau\dagger}$ with operation $\sigma[x^{\tau\dagger} : T]$ modifies the partial map $\sigma(x^{\tau\dagger})$ beyond typical SIMD behaviour:
if an entry in $\sigma(x^{\tau\dagger})$ has the index $i \in \dom(T){\uparrow}$, 
the entry gets updated (if $i \in \dom(T)$) or removed (if $i \in \dom(T){\uparrow} \setminus \dom(T)$). 
Second, the rule for the conditional command in the semantics is unusual.
It splits the input index set $A$ into $A_1$ and $A_2$, 
where $A_1$ consists of indices leading to the true branch and $A_2$ those leading to the false branch. 
The command then runs the true branch under $A_1$,
and then the false branch under $A_2$. The disjointedness of $A_1$ and $A_2$, and 
the fact that $A_1 \cup A_2$ is an antichain
prevent interference between the two branches: they do not write
to the same entry of a variable, and neither of these branches reads from an entry of a variable written 
by the other branch.
This non-interference means the execution order of the true and false branches is interchangeable without changing the semantics of the command.

We now go through the last five rules of the $\Downarrowt$ relation in \Cref{fig:semantics-commands-target-language}, which describe the behaviour of the new constructs in the target language.
The rule for the $\code{lookup\_index}$ command says that the command regards each index
$i = [(\alpha_1,m_1);\ldots;(\alpha_k,m_k)]\in A$ as a partial function mapping string $\alpha_i$ to
integer $m_i$, and it runs only when $\alpha$ is in the domain of this partial function
(i.e., $\alpha \in \dom(i)$). In that case, the command looks up the $\alpha$ entry of the partial function $i$ for all $i \in A$, and stores the results of 
these lookups in the variable $x^{\lintt}$. The next rule is concerned with the $\code{extend\_index}$ command,
and says that the command works in two steps.
First, it runs the command $C$ in its body under the set $A'$ of extended indices, where
$A'$ consists of the indices $i\concat [(\alpha,k)]$ for all $i \in A$ and $k \in \{0,\ldots,n{-}1\}$. Second, it 
goes through every variable $x^{\tau\dagger}$, accesses
the partial map $m_x$ stored in $x^{\tau\dagger}$, and copies, for every $i \in A$, the value
$m_x(i \concat [(\alpha,n{-}1)])$ 
to the $i$ entry of $m_x$ while removing the
entries of $m_x$ whose indices are in $i{\uparrow} \setminus \{i\}$.
So, the copy here clears the entries at indices $i \concat [(\alpha,k)]$ for $k \in \{0,\ldots,n{-}1\}$ if such entries exist in $m_x$.
Note that the copy is implemented with the help of the $(\_)\langle \rho \rangle$ operation \rwl{(\cref{fig:tensor-update-copy})}.
Also, note that the tensor $T''$ in the result of the $\code{extend\_index}$ command is constructed from the values of the $i\concat [(\alpha,k)]$ entries in $T'$
for all $i \in A$ and $k \in \{0,\ldots,n{-}1\}$.
If some of these entries are not
defined, the execution does not proceed. But this failure case does not occur,
as indicated by \Cref{lem:updated-indices-cmd} in \Cref{sec:proof-soundness},
because the domain of the tensor in 
the result of execution is always the same as the given input antichain for the execution,
and so $\dom(T') = A'$.

The following two rules describe the behaviour of the $(\code{loop\_fixpt\_noacc}(n_{+})\ C)$
command. They say that the command repeatedly runs its body $C$ up to $n_{+}$ times, but unlike the usual loop, this repetition may stop early at some iteration $m{+}1$, 
and skip the remaining $n_{+}{-}(m{+}1)$ iterations; this early stopping happens 
when a fixed point is reached at the iteration $m$ (for the first time) and so
the fixed-point check at the following iteration succeeds. Another important point
is that the $(\code{loop\_fixpt\_noacc}(n_{+})\ C)$ command does not accumulate 
the score maps computed during iterations, and keeps only the score map at the 
last iteration (which is $T_{m+1}$ in the case of early stopping and $T_{n+}$ otherwise). 
Discarding the score maps from all the non-last iterations comes from our loop vectorisation;
it is exploited by our translation for vectorisation.
The last rule describes the behaviour of the $\code{shift}$ command:
for every variable $x^{\tau\dagger}$ in the current state,
it looks up the partial map $m_x$ stored in 
the variable $x^{\tau\dagger}$, and shifts values
in $m_x$ according to $\rho$ in the rule,
using the $(\_)\langle \rho \rangle$ operation \rwl{(\cref{fig:tensor-update-copy})}. When $\rho(i) = i'$, the definition 
of $\rho$ ensures that $i \in A{\downarrow}$ and $i' \in A$ for the input antichain $A$
so that the source index of the shift is always in $A{\downarrow}$ and
the target index is in $A$.

\subsection{Soundness of Translation}
\label{ssec:target-translation-soundness}

The next theorem states that our parallelising translation preserves the semantics.
\begin{theorem}
  \label{thm:main-theorem-soundness}
  Let $C$ be a command in the source language, and $\overline{C}$ be the parallelising translation of $C$ in the target language.
  Let $D \in \RDB$, $\sigma_0 \in \State_{\mathit{src}}$, and $\sigma_1 \in \State_{\mathit{tgt}}$. If  
  $\sigma_1(x^{\tau\dagger}) = [[] \mapsto \sigma_0(x^\tau)]$
  for all variables $x^\tau$ in the source language,
  there exist $\sigma_0' \in \State_{\mathit{src}}$, $\sigma_1' \in \State_{\mathit{tgt}}$,
  and $r \in \R$ such that 
  \begin{align*}
    &
    C, D, \sigma_0 \Downarrows \sigma_0', r,
    \quad
    \overline{C}, D, \sigma_1, \{[]\} \Downarrowt \sigma_1', [[] \mapsto r],
    \quad\text{and}
    \quad
    \sigma'_1(x^{\tau\dagger}) = 
    [[] \mapsto \sigma'_0(x^\tau)]
    \ \text{ for all $x^\tau$.}
  \end{align*} 
\end{theorem}
Intuitively, \Cref{thm:main-theorem-soundness} says that if a command $C$ in the source language 
and its translation $\smash{\overline{C}}$ in the target language are run in the same states,
then $C$ and $\smash{\overline{C}}$ terminate successfully, and their final states 
and final scores are essentially the same.
The proof of \Cref{thm:main-theorem-soundness} is provided in \Cref{sec:proof-soundness},
where the theorem is restated as \Cref{cor:main-theorem-soundness}.

The key point of the proof
stems from the fact that $\smash{\overline{C}}$ may store in a single target
language state information about multiple source language states. Specifically, 
when $C$ contains a for-loop, $\smash{\overline{C}}$ executes all the iterations of this loop in parallel. Thus,
each of the states that arise during the execution of the translated loop in $\smash{\overline{C}}$ corresponds to 
multiple states that arise during the execution of the original loop in $C$.
The difficulty is to capture this many-to-one correspondence between states
in the source and target languages, and to prove its invariance so as to
establish the theorem.
We address the challenge by first embedding states and commands in the source
language to those in the target language, then defining a form of simulation relation 
between states of the target language, which we call a {\bf coupling} relation, and finally 
proving that modulo embedding, the relation is preserved by the execution of a command in the source
language and that of its translation in the target language.
For the details of the proof, see \Cref{sec:proof-soundness}.
\section{Implementation}
\label{sec:impl}

We implemented our vectorisation technique for Pyro~\citep{Pyro}, a probabilistic programming language built on PyTorch~\citep{pytorch:19}.
Our choice of Pyro is motivated by its powerful support for parallel execution, which leverages PyTorch's efficient tensor operations, \rshl{tensor broadcasting, and automatic differentiation},
and by its rich set of inference engines, including Stochastic Variational Inference (SVI) and variable elimination.
\rshl{We point out that our implementation is not specific to Pyro and can be adapted to other PPLs that support similar features.}
Beyond standard Pyro/PyTorch primitives,
our design incorporates a non-standard runtime to handle variable reading and writing.
This runtime reflects the SIMD-style semantics of expressions and assignment commands based on the antichains in our target language.
We also define two novel primitives specifically for
the target language's control flow ($\code{ifz}$) and loop ($\code{for}$) constructs.
These functionalities are integrated into a library that overrides standard Python variable access and provides new constructors.
The key aspects of our implementation are summarised below.

\paragraph{Representing finite partial maps with PyTorch tensors.}
Our runtime relies on PyTorch tensors to represent finite partial maps from $\Index$ in our target language, such as $V$-tensors $T$ and partial maps $m$ stored in variables.
This approach, however, needs to address several representation problems.
One of these is that PyTorch uses integer sequences as indices, but our target language uses string-integer sequences as indices.
To bridge this gap, our library maintains a global partial map from strings to integers,
allocating a unique dimension for each string, as shown in \Cref{fig:tensor-indices}.
The other problem is that a finite partial map in our target language accepts indices of varying length, while a PyTorch tensor assumes indices of a fixed length.
To address this problem, we track the maximum integer associated with each string present in the partial map's domain.
We then construct a PyTorch tensor where the maximum index of each tensor dimension is one larger than the tracked maximum integer of the corresponding string. 
For instance, in \Cref{fig:tensor-indices},
a partial map $T$ has strings $\alpha$ and $\beta$ with the maximum integers $9$ and $8$, respectively.
Therefore, it is represented by a PyTorch tensor of shape $(11, 10)$,
where the first dimension corresponds to $\alpha$ and the second to $\beta$ by the global mapping.
Here, the one additional component in each dimension of the PyTorch tensor, accessed by $-1$
in the figure, is used to represent the case that the corresponding string is absent in an index of the partial map. 
Each cell of the PyTorch tensor is then assigned the value evaluated by the $\extend(T)$ function at the corresponding index.
For example, in \Cref{fig:tensor-indices}, for all $0 \le i \le 9$ and $0 \le j \le 8$,
the cell of the PyTorch tensor at index $[i,j]$ is assigned the value of $\extend(T)([(\alpha,i);(\beta,j)]) = c_{i,j}$,
the cell at index $[i,-1]$ is set to the value of $\extend(T)([(\alpha,i)]) = b_i$,
the cell at index $[-1,j]$ is given the value of $\extend(T)([(\beta,j)]) = a$,
and the cell at index $[-1,-1]$ is assigned to the value of $\extend(T)([]) = a$.
This strategy allows us to effectively represent any finite partial map in the target language using a single PyTorch tensor.

\begin{figure}[t]
  \centering
  \hspace{-20pt}
  \begin{minipage}{0.40\textwidth}
    \usetikzlibrary{arrows,decorations.pathreplacing}
    \centering
    \begin{tikzpicture}[
      blankstyle/.style={rectangle,draw,fill=gray!20},
      gray/.style={rectangle,draw,fill=gray!60},
      darkgray/.style={rectangle,draw,fill=gray!100},
      lightgray/.style={rectangle,draw,fill=white},
      ]
      \begin{scope}[local bounding box=matrix]
        \node at (-1.65cm, -0.7cm) [align=right, text width = 1cm] {$\begin{matrix} 0 \\[-0.1em] \vdots \\[0.1em] 9 \\[0.5em] -1 \\[-1.1em] ~ \end{matrix}$};
        \node at (-2.3cm, -0.62cm) [align=right, text width = 1cm] {dim 0};
        \draw[decorate,decoration={brace,mirror}] (-1.6cm, 0.4cm) -- (-1.6cm, -1.7cm);

        \node at (0.4cm, 0.8cm) [align=center, text width = 2.3cm] {$0 \hspace{2.5mm} \cdots \hspace{3mm} 8 \hspace{4mm} {-1}$};
        \node at (0.4cm, 1.35cm) [align=center, text width = 1cm] {dim 1};
        \draw[decorate,decoration={brace}] (-0.8cm, 1.0cm) -- (1.6cm, 1.0cm);

        \node at (1.3cm, -1.5cm) [darkgray, minimum width = 0.6cm, minimum height = 0.6cm] {$a$};
        \node at (1.3cm, -0.35cm) [blankstyle, minimum width = 0.6cm, minimum height = 1.8cm] {$\begin{matrix} b_0 \\ \vdots \\ b_9 \end{matrix}$};
        \node at (0cm, -1.5cm) [gray, minimum width = 2cm, minimum height = 0.6cm] {$a \ \: \cdots \ \: a$};
        \node at (0cm, -0.35cm) [lightgray, minimum width = 2cm, minimum height = 1.8cm] {$\begin{matrix}
                                                                                           c_{0,0} \!\!\!\! & \cdots \!\!\!\! & c_{0,8} \\
                                                                                           \vdots  \!\!\!\! & \ddots \!\!\!\! & \vdots  \\
                                                                                           c_{9,0} \!\!\!\! & \cdots \!\!\!\! & c_{9,8}
                                                                                           \end{matrix}$};
        \draw (-1cm,0.55cm)[line width=1pt] rectangle (1.6cm, -1.8cm);
      \end{scope}
    \end{tikzpicture}
  \end{minipage}
  \begin{minipage}{0.45\textwidth}
    \vspace{10pt}
    \begin{align*}
      \iff \hspace{4mm}
      \begin{aligned}
      & \text{strings\_to\_dims} = [\alpha \mapsto 0, \beta \mapsto 1] \\[0.5em]
      & \begin{aligned}
        T = \smash{\big[} & [] \mapsto a, \\
        & [(\alpha, i)] \mapsto b_i, \\
        & [(\alpha,i); (\beta,j)] \mapsto c_{i,j} \\
        & \, {:} \; 0 \le i \le 9 \text{ and } 0\le j \le 8 \smash{\big]}
      \end{aligned}
    \end{aligned}
    \end{align*}
  \end{minipage}
  \vspace{-0.5em}
  \caption{Correspondence between a PyTorch tensor (left) and a finite partial map $T$ in our target language (right).
  The \text{strings\_to\_dims} mapping ensures this correspondence by allocating strings to unique dimensions.}
  \label{fig:tensor-indices}
    \vspace{-1em}
\end{figure}
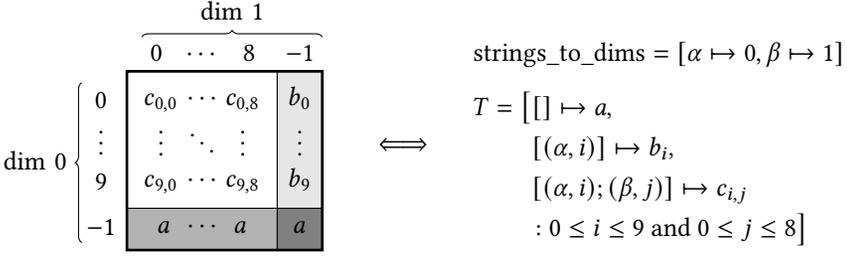

\paragraph{Antichains and PyTorch tensor operations.}
To preserve the correspondence between finite partial maps and PyTorch tensors,
our library overrides the standard variable read and write of Python in a way that respects the antichain.
An antichain, being a finite set of indices, is itself represented by a PyTorch tensor,
where each tensor cell holds a boolean value indicating whether the corresponding index is in the antichain.
Based on this representation,
our library reads a variable by constructing a PyTorch tensor
consisting solely of elements at indices for which the current antichain evaluates to true.
Subsequent operations (e.g., arithmetic) on PyTorch tensors are then performed element-wise,
facilitated by PyTorch's \emph{broadcasting} semantics.
The outcomes of these operations are then used to update the PyTorch tensor assigned to the variable, but in a more careful way:
it updates not only the element at the current antichain's true indices,
but also the elements at the upward closure of those indices,
which can be also efficiently handled by broadcasting.
\rshl{
  These steps are implemented efficiently in PyTorch thanks to its
  support for broadcasting and to the preservation of the encoding of
  finite partial maps into PyTorch tensors across all the base
  operations of the target language.
}

\paragraph{Variable elimination.}
Discrete random variables pose difficulties for probabilistic inference algorithms
that rely on gradient computation, such as SVI and HMC.
A common solution, known as \emph{variable elimination},
enumerates all the possible combinations of the values of these discrete random variables,
computes a total score for each combination, and sums them up.
While variable elimination is costly in general, Pyro efficiently performs it
by allocating a dedicated tensor dimension for the enumeration of each discrete random variable,
and leveraging PyTorch's broadcasting semantics and tensor-based parallelism to compute total scores.
To support these optimised algorithms,
our library internally manages metadata (e.g., dimension-variable mappings)
for the correct handling of these additional dimensions for enumeration.

\section{Experimental Evaluation}
\label{sec:eval}

\newcommand{\mours}{\texttt{ours}\xspace}
\newcommand{\mplate}{\texttt{plate}\xspace}
\newcommand{\mvmarkov}{\texttt{vmarkov}\xspace}
\newcommand{\mdischmm}{\texttt{discHMM}\xspace}
\newcommand{\mdiscHMM}{\mdischmm}
\newcommand{\mseq}{\texttt{seq}\xspace}
\newcommand{\mhand}{\texttt{hand}\xspace}

In this section, we report on the evaluation of probabilistic programs translated using our method,
focusing on the following three research questions:
\begin{enumerate}[label={\bfseries (RQ\arabic*)}, ref=RQ\arabic*]
\item \label{rq:sound}
  {\bf Consistency}:
  Does our approach produce scores and their gradients that are
  consistent with those produced by existing methods (modulo floating-point errors)?
\item \label{rq:time-memory}
  {\bf Time and memory}: %
  How does our approach compare to existing techniques in terms of
  execution time and peak memory usage?
\item \label{rq:scalability}
  {\bf Scalability}:
  Does our approach scale with dataset size, assuming a fixed model?
\end{enumerate}

\paragraph{Benchmarks.}
We used five classes of probabilistic models, with complementary features.

\begin{asparaitem}
\item \emph{Hidden Markov models} (\texttt{hmm-\{ord1,ord2,neural\}})
  \cite{pyro-github:19,tensor-var-elim:19}.
  These models describe piano scores~\cite{polyphonic:12}
  using variants of a hidden Markov model (HMM).
  The models \texttt{hmm-\{ord1,ord2\}} are first- and
  second-order HMMs involving 3 nested loops but no neural networks,
  whereas \texttt{hmm-neural} is a first-order HMM with 2 nested loops
  and neural networks.
  In these models, the outermost loop iterates over 229 (\texttt{hmm-\{ord1,neural\}}) or 50 (\texttt{hmm-ord2}) sequences of music notes,
  the first nested loop stands for the HMM iteration over time steps of length 129,
  and the innermost loop of \texttt{hmm-\{ord1,ord2\}} iterates over piano keys of length 51.
  The first and third loops are easily parallelisable since each iteration is conditionally independent.
\item \emph{Deep Markov model} (\texttt{dmm})
  \cite{dmm-original:17,pyro-github:19}.
  This model is another HMM variant describing the same music dataset.
  It is a first-order Markov model, consisting of two nested loops over 229 sequences of music notes and 129 time steps, respectively.
  It differs from \texttt{hmm-*} in two ways:
  first, \texttt{dmm} uses continuous hidden states, while \texttt{hmm-*} uses discrete states;
  second, \texttt{dmm} uses neural networks for both prior and likelihood, while \texttt{hmm-neural} does so only for likelihood.
\item \emph{Nested hidden Markov models} (\texttt{nhmm-\{train,stock\}}).
  These first-order nested HMMs model
  the number of arrivals at a train station over 44,640 time steps (hours) \cite{bart-train-ridership},
  and the stock price of 10 companies over 2,560 time steps (trading days) \cite{yfinance}, respectively.
  Both models consist of 4 nested loops,
  either over 5 years, 12 months, 31 days, 24 hours,
  or over 10 companies, 10 years, 4 quarters, 64 days.
  Only the outermost loops can be easily parallelised.
\item \emph{Autoregressive model} (\texttt{arm}).
  This variant of autoregressive model (AR)~\cite{autoreg-speech:17} describes speech data taken from \cite{Librispeech}.
  While standard AR models have a constant order, the order of this model is a latent random variable.
  This model contains one loop iterating over 2,000 time steps.
\item \emph{Temperature controller model} (\texttt{tcm})
  \cite{multilevel-mc:17,reparam-grad:18}. 
  This state-space model describes the evolution of the temperature inside a room, equipped with a stochastic temperature controller.
  It includes if-then-else statements, required to convey different types of temperature dynamics based on the controller state.
  It consists of two nested loops over 10 temperature sequences and 200 time steps, respectively,
  where the first loop can be easily parallelised.
\end{asparaitem}

\paragraph{Baseline methods.} 
Pyro features several loop vectorisation methods.
\begin{asparaitem}
\item \label{enum:baseline-plate}
\emph{Plate constructor~\citep{pyro-plate}.}
It is a tensor-based loop constructor that vectorises loops without data dependencies.
However, it may lead to incorrect inference results
if the variables in a loop have data dependencies across iterations.
Plates can be nested but cannot contain if-then-else statements.
\item \label{enum:baseline-vmarkov}
\emph{Vectorised Markov constructor~\citep{pyro-vmarkov}.}
It is a ``funsor''-based loop constructor
that vectorises certain types of loops with data dependencies.
It supports scalar or vector indices determined by the user-defined dependence order,
and includes an efficient automatic variable elimination algorithm for enumeration.
However, it cannot vectorise if-then-else statements and nested loops.
\item \label{enum:baseline-discHMM}
\emph{Discrete HMM distribution~\citep{pyro-dischmm}.}
It is a distribution class in Pyro that expresses a first-order discrete HMM.
Given transition and emission matrices, it efficiently computes log-likelihood and posterior marginals in parallel
using algorithms such as the forward-backward algorithm.
However, it is limited to first-order discrete HMMs without if-then-else statements or nested loops.
\rshl{
\item \label{enum:baseline-manual}
\emph{Hand-coded vectorisation.}
This baseline refers to the manual vectorisation of loops using tensor operations.
As a result, applying it requires expertise in both the Pyro PPL and tensor-based programming,
and often entails non-standard manipulation of the PPL's internal structures.
}
\end{asparaitem}
\vspace{0.25em}

We compare our method (\mours) with \rshl{five} reference methods that combine some of these features:

\begin{asparaitem}
\item \mseq: It is a sequential implementation of the model, where all loops are executed sequentially.
\item \mplate: It uses only plate constructors to vectorise loops without data dependencies.
\item \mvmarkov: It uses plate constructors and vectorised Markov constructors to vectorise loops.
\item \mdiscHMM: It uses plate constructors and discrete HMM constructors to vectorise loops.
\rshl{\item \mhand: It refers to manual vectorisation of loops using tensor operations.}
\end{asparaitem}
Note that all methods except \mseq have limitations and cannot be applied to all models.
In particular, some benchmark models include if-then-else statements, nested loops, or higher-order dependencies that these methods cannot handle.
We consider \mplate, \mvmarkov, or \mdiscHMM inapplicable to a model
if no loop in the model can be vectorised by a plate constructor, a vectorised Markov constructor, or a discrete HMM constructor, respectively.
\rshl{Also, \mhand is considered inapplicable to a model when a loop in the model contains an if-then-else statement.
This is because vectorisation using tensor operations in such cases is highly non-trivial.}
See \Cref{sec:benchmarks-methods} for further details on the applicability.
We also note that the use of \mplate or \mvmarkov requires care since
they are only sound for loops without data dependencies or when the dependence order is correctly specified by the user.
\rshl{See \Cref{sec:code} for the code snippets of the model \texttt{hmm-neural} using our method and the baselines.}

\paragraph{Experimental setup.}
\rshl{We run three types of inference engines, Maximum a Posteriori (MAP) estimation, Stochastic Variational Inference (SVI),
and Markov Chain Monte Carlo (MCMC) sampling.
MAP estimation is performed for the models \texttt{hmm-*} and \texttt{nhmm-*}.
SVI is run using RNN-based amortised variational distributions for \texttt{dmm}
or independent normal variational distributions for \texttt{arm} and \texttt{tcm}.
The inference engines for MAP estimation and SVI require solving an optimisation problem,
where an objective function is the posterior probability density or the evidence lower bound.
We use the Adam optimiser~\cite{adam} to solve this problem, i.e., to train learnable parameters.
MCMC is used for models without neural networks, i.e., \texttt{hmm-\{ord1,ord2\}}, \texttt{nhmm-*}, \texttt{arm}, and \texttt{tcm}.
Specifically, we employ the Hamiltonian Monte Carlo (HMC) sampler~\cite{DiGS:24,HMC-github:24} with fixed number of leapfrog steps and step size
to sample from the posterior distribution of continuous variables,
while a discrete variable in \texttt{arm} is handled via Gibbs sampling~\cite{Geman:84} based on
Metropolis-Hastings~\cite{Metropolis:53,Hastings:70} with a random walk proposal.
All timings are measured using an NVIDIA 2080Ti GPU.
}

\begin{table}[t]
\fontsize{8.5pt}{10pt}\selectfont
\vspace{1pt}
\begin{tabular}{l >{\columncolor{gray!20}}r r >{\columncolor{gray!20}}r r >{\columncolor{gray!20}}r r >{\columncolor{gray!20}}r r}
  \toprule
  & \multicolumn{2}{c}{\mplate} & \multicolumn{2}{c}{\mvmarkov} & \multicolumn{2}{c}{\mdiscHMM} & \multicolumn{2}{c}{\mours} \\
  & \multicolumn{1}{c}{obj} & \multicolumn{1}{c}{grad} & \multicolumn{1}{c}{obj} & \multicolumn{1}{c}{grad} & \multicolumn{1}{c}{obj} & \multicolumn{1}{c}{grad} & \multicolumn{1}{c}{obj} & \multicolumn{1}{c}{grad} \\
\midrule
\texttt{hmm-ord1}   & 0      & 3.5e-7 & 0      & 3.5e-7 & 8.7e-8 & 2.0e-6 & 0      & 3.4e-7 \\
\texttt{hmm-neural} & 2.0e-9 & 4.5e-6 & 5.0e-9 & 4.0e-6 & 4.0e-9 & 4.1e-6 & 5.0e-9 & 4.0e-6 \\
\texttt{hmm-ord2}   & 0      & 2.7e-7 & 0      & 2.7e-7 & -      & -      & 0      & 2.6e-7 \\
\texttt{dmm}        & 7.7e-9 & 3.1e-7 & 2.5e-8 & 3.1e-7 & -      & -      & 2.4e-8 & 3.1e-7 \\
\texttt{nhmm-train} & 0      & 8.1e-6 & -      & -      & -      & -      & 0      & 8.1e-6 \\
\texttt{nhmm-stock} & 0      & 5.4e-6 & -      & -      & -      & -      & 0      & 5.5e-6 \\
\texttt{arm}        & -      & -      & 0      & 3.0e-9 & -      & -      & 0      & 3.0e-9 \\
\texttt{tcm}        & -      & -      & -      & -      & -      & -      & 0      & 0      \\
\bottomrule
\end{tabular}
\caption{
  Average relative errors of objective function and gradient values computed by
  \mplate, \mvmarkov, \mdiscHMM, and \mours with respect to \mseq \rshl{ in SVI or MAP inference}.
  ``-'' indicates the method is not applicable.
  \rshl{Note that \mplate is not applicable to \texttt{arm} as it contains only a single loop with data dependencies.}}
\label{table:soundness}
\vspace{-2em}
\end{table}

\vspace{2mm}
\noindent{\bf [\cref{rq:sound}] Consistency.}\hspace{2mm}
To validate that our method yields correct outputs,
we compared each of the vectorisation methods \rshl{except \mhand} against {\mseq} \rshl{using SVI or MAP inference}, in terms of
(i) the objective function value, computed as the sum of score values at each $\kwscore$ command, and
(ii) the gradient of the objective function \rshl{with respect to learnable parameters computed by automatic differentiation}.
We then compared these results across the different methods to assess their consistency with \mseq.

More precisely, we measured the relative errors between these values in $L_2$-norm.
We computed the objective function or gradient values $v \in \smash{\mathbb{R}^d}$ using \mseq,
and $\hat{v} \in \smash{\mathbb{R}^d}$ using one of the other vectorisation methods,
over the first training step with the same random seed;
and then computed the relative error $\| v - \hat{v} \|_2 / \| v \|_2$.
To alleviate the high computational cost of the \mseq method, we used only 10 of the outermost sequences 
of the datasets used in the \texttt{hmm-*} and \texttt{dmm} benchmarks,
while the entire datasets were used for the other benchmarks.

\Cref{table:soundness} presents these relative errors averaged over 100 random seeds for each model and each vectorisation method.
It demonstrates that the outputs from \mours are the same as, or very close to
the outputs from \mseq for all models considered.
We believe that the non-zero errors come from
the non-associativity of floating-point addition/multiplication.
For objective function values, we manually verified that they are caused by different orderings of floating-point operations:
\texttt{hmm-neural} and \texttt{dmm} involve matrix multiplications,
and PyTorch computes the addition/multiplication operations there
in different orders in \mours and \mseq, due to different tensor dimensions in the computation.
For gradient values, we conjecture that a similar reordering of operations is introduced by PyTorch
during backpropagation, due to different tensor dimensions or non-contiguous memory layouts.
We also note that similar discrepancies caused by these factors appeared across all other vectorisation methods.
Overall, these observations suggest that our method yields consistent outputs, and the remaining errors are
well within the bounds of expected floating-point behaviour.

\begin{figure}[t]
\vspace{-3pt}
\includegraphics[width=\textwidth]{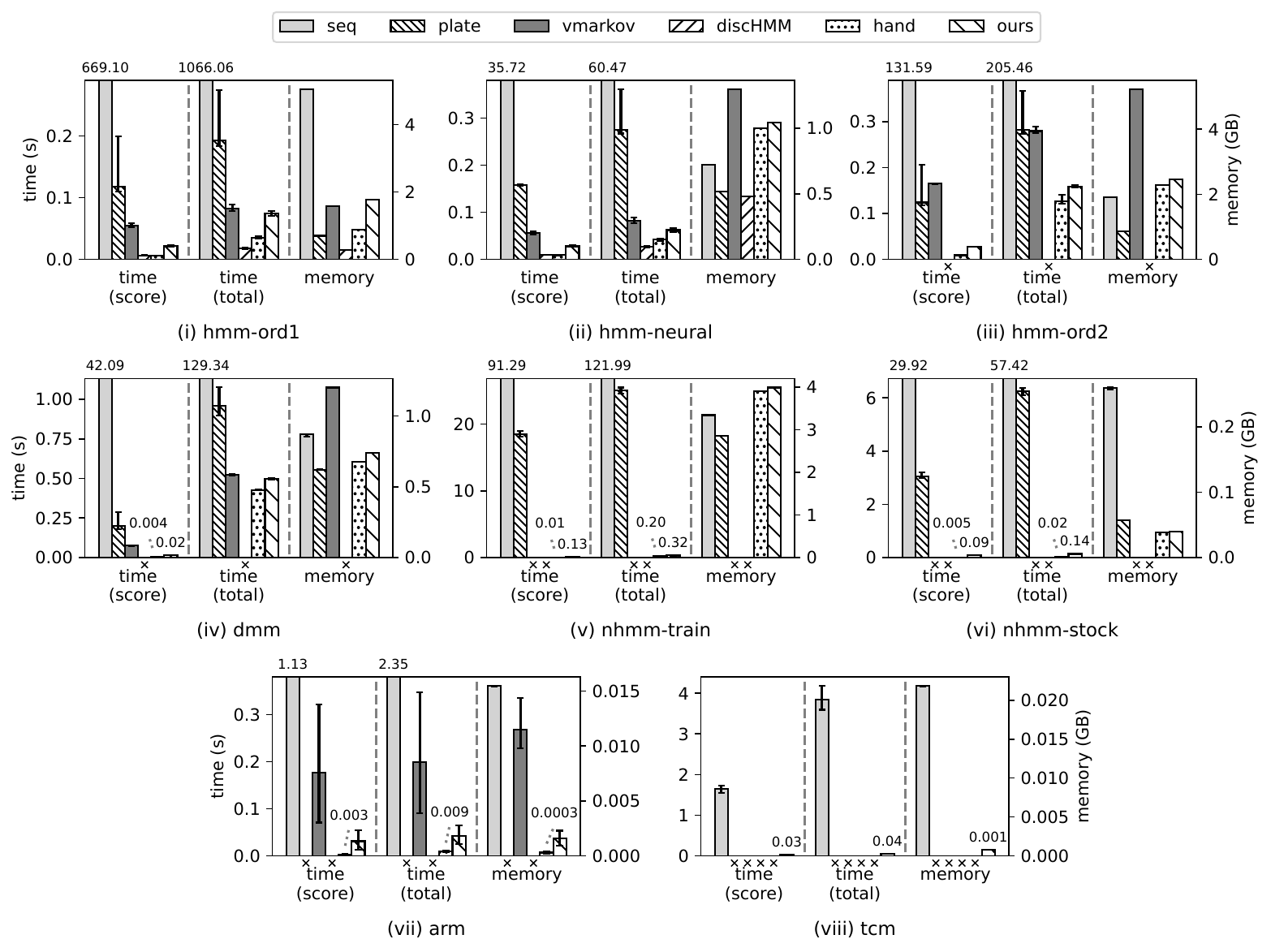}
\vspace{-2.2em}
\caption{\rshl{Score time (s), total time (s), and GPU memory usage (GB) in each training step of SVI/MAP with 95\% confidence intervals.
``$\times$'' below the $x$-axis indicates inapplicability.
Large values are truncated.}
}
\label{fig:performance}
\vspace{-1em}
\end{figure}

\vspace{2mm}
\noindent{\bf [\cref{rq:time-memory}] Time and Memory.}\hspace{2mm}
We compared the wall-clock time and GPU memory usage of our method with other baselines \rshl{using SVI, MAP inference, and MCMC sampling}.
\rshl{For SVI and MAP,} we measured, for each training step, the time to compute score (score time) and the time to run the step entirely (total time).
The total time includes components other than score computation:
e.g., sampling from variational distributions, performing variable elimination, and computing gradients.
\rshl{For MCMC, we measured
the time to execute each sampling step.
We also measured the maximum GPU memory usage over a training step for SVI and MAP,
and over a sampling step for MCMC.}
\rshl{For SVI and MAP,} every reported number is obtained by averaging over 5 runs, where each run consists of 10 training steps (\mseq) or 1,000 training steps (the others),
and the first 1 out of 10 steps (\mseq) or 10 out of 1,000 steps (the others) are excluded from the averages, as considered burn-in steps.
\rshl{Concretely, we averaged over 45 data points for \mseq and 4,950 data points for the other methods.
For MCMC, each reported number is the average over 8 runs,
each with a different number of sampling steps performed within one hour
(i.e., 8 $\times$ the average steps per hour in total).}
In this evaluation, the full datasets are used for all benchmarks.
The main results are shown in \Cref{fig:performance,table:mcmc-results},
the full results are in \Cref{sec:eval-appendix}.

\paragraph{\rshl{Results on SVI/MAP time}}
Compared to \mseq, \mours shows significant speedups
in both score time ($36.2$--$31399.3\times$) and total time ($54.8$--$14302.9\times$).
These speedups are mainly enabled by two factors.
First, in a score computation,
\mours drastically reduces the iteration counts required for score computation (2 to 10.1) 
compared to \mseq's extensive iterations (2,000 to 1,506,591),
achieved by the loop vectorisation with early termination.
Second,
\mours consistently employs tensors, while \mseq uses a list of scalars. 
Many operations in Pyro, including variable elimination and sampling from variational distributions,
are more efficient over tensors than separate scalars.

Compared to \mplate, \mours again shows significant speedups
in score time ($4.5$--$145.8\times$) and total time ($1.8$--$78.4\times$).
These gains are also due to \mplate's high iteration count (129 to 8,928)
and \mours's efficient tensor operations.
Furthermore, unlike \mours, \mplate is not applicable to the models \texttt{arm} and \texttt{tcm},
as all the loops in these models have data dependencies or if-then-else statements.

Compared to \mvmarkov, \mours is faster in both score time
($2.0$--$6.0\times$) and total time ($1.1$--$4.6\times$).
These speedups stem from two main sources.
First, \mours exhibits a fewer number of iterations:
for a loop with dependence order $K$, \mours performs $K+1$ iterations,
while \mvmarkov requires $2K+1$ iterations.
Second, \mvmarkov introduces additional overheads
due to the use of ``funsors'' (functional tensors).
A conversion from PyTorch tensors to Pyro funsors, followed by operations over funsors,
is much slower than direct operations over tensors.
In addition, \mvmarkov is not applicable to the models \texttt{nhmm-*} and \texttt{tcm},
as they contain nested loops with dependencies or if-then-else statements.

Compared to \mdiscHMM, \mours is slower in score time ($0.3$--$0.4\times$) and total time ($0.2$--$0.4\times$)
for the models \texttt{hmm-\{ord1,neural\}}.
However, \mdiscHMM only supports first-order HMMs without nesting and if-then-else statements,
and thus cannot be applied to the other six models.

\rshl{Compared to \mhand, ours is slower in score time ($0.1$--$0.3\times$) and total time ($0.2$--$0.9\times$) for all models except \texttt{tcm}
where \mhand could not be applied due to the presence of if-then-else statements.
}

\paragraph{\rshl{Results on SVI/MAP memory usage}}
Compared to \mseq, \mours shows a decrease of memory usage for some models and an increase for others.
For models \texttt{nhmm-stock}, \texttt{arm}, and \texttt{tcm},
\mours significantly reduces peak GPU memory usage, requiring only 3--15\% of the memory used by \mseq.
For models \texttt{hmm-\{neural, ord2\}} and \texttt{nhmm-train}, \mours requires 119--144\% of the memory used by \mseq.
Two effects play against each other here.
Our approach uses antichains representing indices considered in a given state.
In our current implementation, these antichains are represented as boolean
tensors, which may be sparse, hence causing memory inefficiencies.
However, even then, \mours may benefit from more efficient memory allocation
schemes implemented by GPUs, e.g., with allocation in large fixed-sized chunks.
In contrast, \mseq allocates larger numbers of smaller chunks, which may
increase memory overhead, especially when the maximum vector dimension used
in the innermost loop is small.
These factors explain that \mours can sometimes achieve lower memory usage under
specific conditions, although it seems to have a higher memory requirement in theory.

Compared to \mplate, \mours reduces memory usage to 70\% on \texttt{nhmm-stock} benchmark,
due to the trade-off between usage of antichains and memory allocation schemes discussed above.
On the other models, \mours uses more memory, ranging from 119\% to 284\% of the memory used by \mplate.

Compared to \mvmarkov, \mours uses less memory in all models except \texttt{hmm-ord1}.
As we remarked above, \mours stores antichains which require additional memory.
However, \mvmarkov also induces an important memory overhead,
due to additional execution of iterations with scalar indices and the implementation of funsors.
The results demonstrate that the overall overhead due to funsor in
\mvmarkov is higher than that generated by antichains in \mours.

Compared to \mdiscHMM, \mours uses more memory for the two models that the former can handle.
However, as mentioned earlier, \mdiscHMM is specialised to a narrow class of models,
so it could not be applied to the other six models.

\rshl{Compared to \mhand, \mours uses more memory for all models except \texttt{tcm} where \mhand could not be applied.
Specifically, \mours uses 518\% and 200\% of the memory used by \mhand on \texttt{arm} and \texttt{hmm-ord1}.
For the others, \mours uses 102--109\% of the usage of \mhand, showing comparable memory usage.
}

\begin{table}[t]
\vspace{1pt}
\rshl{
\small
\begin{tabular}{lccccc}
\toprule
& \mseq & \mplate & \mvmarkov & \mdiscHMM & \mhand \\
\midrule
MCMC speedup (other $\to$ \mours) & $67$--$1210\times$ & $1.83$--$65.63\times$ & $1.07$--$2.56\times$ & $0.28\times$ & $0.26$--$0.80\times$ \\
MCMC memory usage (\mours/other) & 3--120\%         & 69--253\%           & 24--109\%          & 562\%        & 102--1290\% \\
\bottomrule
\end{tabular}
}
\caption{
\rshl{Speedup and and relative GPU memory usage per MCMC sampling step,
compared between \mours and the others.
Each entry shows the range over applicable benchmarks.
See \Cref{sec:eval-appendix} for full results.}
}
\label{table:mcmc-results}
\vspace{-2em}
\end{table}

\paragraph{\rshl{Results on MCMC}} 
\rshl{%
As shown in \Cref{table:mcmc-results}, the results are consistent with the observations
from the SVI and MAP inference experiments across all benchmarks,
except for \texttt{hmm-neural} and \texttt{dmm} benchmarks that we do not consider in MCMC experiments.
Overall, \mours exhibits superior sampling speed
(i.e., lower time per sampling step)
and broader applicability compared to the vectorisation baselines,
often while maintaining comparable or lower memory consumption.
For instance, compared to \mvmarkov, \mours performs $1.07$--$2.56\times$ more sampling steps,
while using 24--109\% of the GPU memory consumed by \mvmarkov,
across all benchmarks except \texttt{nhmm-*} and \texttt{tcm}, for which \mvmarkov is inapplicable.
The applicability of each method to the benchmarks is consistent with that observed in the SVI and MAP inference experiments,
except for the case of \mseq on \texttt{hmm-ord1}, where it fails to complete a single step within one hour.
}

\begin{figure}[t]
\includegraphics[width=\textwidth]{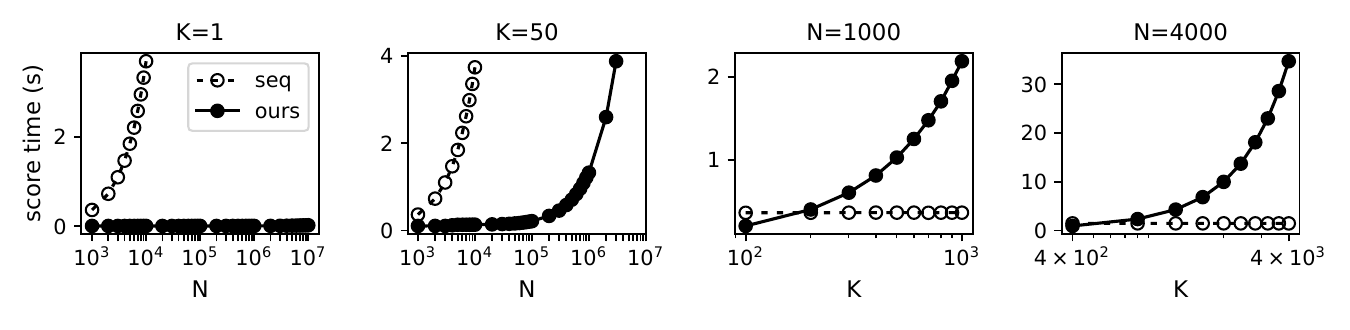}
\vspace{-2.8em}
\caption{Score time (s) of \texttt{ours} and \texttt{seq}
  for the model \texttt{arm'} with different $N$ and $K$.
  Each point denotes the mean over 99 runs.
  95\% confidence intervals are omitted as their relative width to the mean is below 2\%.}
\label{fig:scalability}
\vspace{-0.4em}
\end{figure}

\vspace{2mm}
\noindent{\bf [\cref{rq:scalability}] Scalability.}\hspace{2mm}
To assess the scalability of our method, 
we considered the autoregressive model \texttt{arm'}, where $N$ denotes
the number of samples and $K$ the order of the model:
$y_i \sim \cN(y_{i-1} + \cdots + y_{i-K}, 1)$ for $i = 0,\ldots, N-1$.
Here $y_i$ is the $i$-th sample for $i = 0,\, \ldots,\, N-1$,
and the initial values are $y_{-K} = \cdots = y_{-1} = 0$.
We express the model in Pyro 
using a loop of length $N$ to generate the $y_i$'s,
which takes $K+1$ iterations to reach a fixed point when using our method.

We measured the time to compute score (score time) under two setups:
use a fixed $K \in \{ 1, 50 \}$ and vary $N$;
and use a fixed $N \in \{ 1000, 4000 \}$ and vary $K$.
All samples were generated from the standard normal distribution,
and the time for generating the samples was excluded from the measurement.
For each $(N, K)$ pair, we averaged the time over 99 runs out of 100,
excluding the first burn-in step.
We compared the time taken by the vectorised loop produced by our method
(\mours) against the sequential loop (\mseq).
The results are shown in \Cref{fig:scalability}.

When $K$ is fixed to a value significantly smaller than $N$ (i.e., 1 and 50),
\mours exhibits a notably slower increase rate with respect to $N$
compared to \mseq.
The improvement is more significant when $K$ is small.
This demonstrates that, when the dependence order of the loop is low,
our vectorisation method can improve performance more substantially
by leveraging the GPU parallelism.

In contrast, when $N$ is fixed (i.e., 1000 and 4000) and $K$ is close to $N$,
the runtime of \mseq remains almost constant, while that of \mours is quadratic in $K$.
Indeed, when $K$ and $N$ are close to each other, both the number of
iterations and the size of the vectors processed within each iteration
grow proportionally to $K$.
These findings suggest that our method may not be optimal for higher-order
loops.
Nevertheless, in many practical applications of autoregressive models,
such as speech processing~\cite{autoreg-speech:10,autoreg-speech:17}, 
heart rate variability analysis~\cite{autoreg-heart:02}, and image dynamic texture modelling~\cite{autoreg-texture:07},
the value of $K$ is typically much smaller than $N$.
To mitigate the performance degradation when $K$ approaches $N$,
an adaptive strategy could be employed to switch between sequential and
vectorised execution based on the loop dependence order (e.g., $K$), which
our method can automatically detect.

\section{Related Work}
\label{sec:related}


\paragraph{\rhl{Model translations in PPLs}}
Several prior approaches have translated models from non-tensor-based PPLs into tensor-based languages,
such as compiling Stan~\cite{Stan} models to Pyro or NumPyro models for improved inference performance with provable correctness~\cite{compile-stan:21},
or converting high-level discrete probabilistic programs into NumPy~\cite{numpy} programs with optimised tensor operations (e.g., \texttt{opt\_einsum}~\cite{opt-einsum})
for exact inference over discrete random variables~\cite{compile-to-einsum:23}.
\rhl{
These approaches do not alter the structures of the original models as aggressively as our method does for loops; instead, they mostly rely on generic optimisation of the target language for performance improvements.}

\vspace{-1pt}
\paragraph{\rshl{Vectorisation primitives in PPLs}}
Most tensor-based PPLs have introduced convenient features to vectorise
independent operations in probabilistic models.
A key example is the \texttt{plate} primitive in Pyro and NumPyro,
which originates from graphical models~\cite{plate} and is used in PPLs to vectorise loops,
particularly when the random variables inside the loop are conditionally independent or when the loop has no data dependencies.
Similarly, JAX provides \texttt{vmap}~\citep{jax} and TensorFlow offers \texttt{vectorized\_map}~\citep{vectorized-map},
both of which aid in vectorising function evaluation over independent tensor dimensions.
These primitives also help vectorise MCMC sampling, improving inference
efficiency across multiple chains~\cite{tfp-mcmc}.
While these primitives are effective for simple loops or sequences of independent operations,
our work extends this capability to loops with more complex dependencies.

Other works in PPL aim at extending vectorisation primitives to handle more
complex dependencies, similarly to our approach.
For instance, TensorFlow Probability (TFP)~\cite{tensorflow-probability}
provides \texttt{MarkovChain} distributions in order to vectorise loops
with short-range dependencies in models such as random walks or
autoregressive models.
Similarly, Pyro provides a family of hidden Markov model (HMM)
distributions (e.g., \texttt{DiscreteHMM}, \texttt{GaussianHMM})
to vectorise key computations in HMMs, such as smoothing, filtering
algorithms, and score computation~\cite{prefix-sum-smoother:20}.
However, these approaches usually require the restructuring of a model so
as to make explicit use of specific vectorisation primitives, or have
limitations related to the kind of dependencies they support (e.g.,
supporting only predefined transition distributions).
More closely related to our work, Funsor~\cite{funsor:2019} introduces \texttt{vectorized\_markov} to vectorise loops with dependencies,
but requires users to manually specify dependency lengths,
and is not generally applicable to all dependency structures (e.g., nested structures or conditional branches).
Our work extends these ideas by introducing a more general, automatic
vectorisation method that ensures correctness through fixed-point check and eliminates the need for any additional user input or
extensive model change.
Moreover, it can handle complex dependencies, which may be nested or have dependencies determined at runtime.

\vspace{-1pt}
\paragraph{\rhl{Optimisation of inference algorithms}}
\rhl{Our work is related to approaches that optimise parameter learning or posterior inference algorithms by exploiting properties of PPL models.
A static analysis for the smoothness properties of PPL models has been developed
to propose an SVI algorithm that exploits these properties to reduce variance in gradient estimation~\cite{smoothness}.
Other work focuses on conditional independence properties in PPL models and proposes a type system for capturing these properties,
along with an optimised variable elimination algorithm that leverages them~\cite{cond-indep}.
Meanwhile, methods have been introduced to automatically derive correct and efficient gradient estimators for a given SVI objective when it is expressed as a program,
by transforming it into gradient estimation code~\cite{becker}.
Our method does not directly optimise learning or inference algorithms;
instead, it optimises the execution of probabilistic programs by vectorising loops with dependencies,
which in turn can accelerate various inference algorithms relying on efficient program execution.}

\vspace{-1pt}
\paragraph{\rshl{Automatic vectorisation in general-purpose languages}}
Our work is also related to studies in general-purpose programming languages
that aim at automating vectorisation.
For example, dynamic data dependency analysis techniques have been proposed
to infer SIMD parallelisation potential by identifying all possible
dependency-preserving reorderings, together with further analysis of memory
access patterns~\cite{dynamic-trace:12}.
Additionally, methods have been developed to automatically batch operations in dynamic neural network libraries~\cite{on-the-fly-batching:17}, such as DyNet~\cite{dynet},
and to handle data-dependent control flow and recursion in programs using a pointer per batch element to manage variables and program counters~\cite{auto-batching:20}.
PyTorch's framework for automatic batching~\cite{auto-batching-pytorch:18} leverages mask propagation and control flow rewriting to simplify minibatching in dynamic models.
Recently, Autovesk~\cite{autovesk:23} automates vectorisation by transforming scalar instructions into vectorised instructions over a graph of operations,
using heuristics to minimise the overall instruction numbers.
Although these methods rely on static or dynamic analysis of data dependencies to exploit vectorisation opportunities,
our work avoids the need for complex analysis by an automated fixed-point check that ensures correctness of vectorisation.

\rshl{There are also approaches that use thread-level speculation (TLS) to parallelise general sequential programs.}
R-LRPD~\cite{r-lrpd}, for instance, speculatively executes loop iterations in parallel and records read/write bit arrays to re-execute dependent iterations.
GPU-TLS~\cite{gpu-tls} and STMlite~\cite{stm-lite} both focus on re-execution of mis-speculated iterations, and dynamic analysis to ensure correctness.
Other methods improve performance by integrating TLS with transactional memory (TM)~\cite{tlstm} or adopting in-place speculation~\cite{splip}.
However, these methods typically discard partial results on mis-speculation during re-execution,
which can cause a large number of re-executions and slow convergence,
particularly in loops with many inter-iteration dependencies, as observed in most of the models in our experiments.
In contrast, our method can extract useful partial information even from such mis-speculated iterations through the $\code{shift}$ construct, achieving the faster convergence in those cases, as demonstrated by our experimental results.

\section{Conclusion \rshl{and Limitations}}
\label{sec:conclusion}

We have presented an automatic sound method for the
vectorisation of loops over a fixed range, as often found in PPLs.
It is based on a speculative execution of the body of for-loops, using
a non-standard semantics, and on a fixed-point check to ensure correctness
of the final result.
Its implementation upon Pyro proceeds by translation into a lower level PPL
with primitives for vectorisation.
We proved this correctness of the method and demonstrated its effectiveness
in accelerating inference for several families of complex probabilistic
models including sequence models and hierarchical models, and for
various probabilistic inference engines.

\rshl{%
While our work leads to performance improvements in our experiments,
it also has limitations.
First, its fixed-point check introduces a performance overhead that might
not always be compensated for when the iteration count is not substantially reduced by our method---for instance, due to long data dependencies, as discussed in \Cref{sec:eval}.
Although such cases are not expected to be dominant in practice,
a potential avenue for improvement lies in exploring a hybrid approach that dynamically decides
whether to apply vectorising optimisation or fall back to a standard runtime based on program characteristics.
Second, our work does not consider the vectorisation of sample generation,
which is central to many inference algorithms
(e.g., importance sampling, Metropolis-Hastings, and Sequential Monte Carlo~\cite{smc}).
Our method instead focuses on the vectorisation of density computations,
which is orthogonal to sample generation and is central to other inference algorithms (e.g., SVI and HMC).
Addressing this limitation would yield a more comprehensive vectorisation framework
and further accelerate a broader class of inference algorithms.
}

\bibliography{paper}

\AtEndDocument{%
\appendix
\clearpage
\section{Proof of Soundness of Parallelising Translation}
\label{sec:proof-soundness}
In this section, we prove \Cref{thm:main-theorem-soundness}, the soundness theorem for our parallelising translation in \Cref{subsec:target-translation}.
The re-stated corollary and its complete proof can be found in \Cref{cor:main-theorem-soundness}.
The proof of the corollary proceeds in the following steps:
\begin{enumerate}
\item We define a variant of the evaluation rules in the target language, called the \emph{intermediate evaluation relation}, denoted $\Downarrowi$,
and prove that it preserves the meaning of the original evaluation relation $\Downarrowt$ in the target language
(\Cref{ssec:proof-semantics}, intermediate-to-target).
\item Changing only the types of variables of commands in the source language
so that the commands belong to the target language preserves the semantics under evaluation using $\Downarrow$ and $\Downarrowi$
(\Cref{subsec:simple-embedding-commands}, source-to-intermediate).
\item An optimisation of target language commands from the second step,
which replaces for-loops with indexless loops under the parallelising mechanism,
preserves the semantics under evaluation using $\Downarrowi$
(\Cref{subsec:parallelising-for-loops}, intermediate-to-intermediate).
\item From the observations in preceding steps,
our parallelising translation from the source language to the target language preserves the semantics
under evaluation using $\Downarrow$ and $\Downarrowt$
(\Cref{subsec:parallelising-translation}, source-to-target).
\end{enumerate}
The last three steps rely on a class of binary relations called \emph{couplings},
whose definition and properties are described in \Cref{subsec:couplings}.

\subsection{New Evaluation Rules for Proof of Soundness}
\label{ssec:proof-semantics}
In this section, we introduce a variant of the evaluation rules in the target language defined in \Cref{sec:target-lang}.
This variant will play a crucial role for formally proving the soundness of our parallelising translation in the following sections.
First, we define an \emph{intermediate evaluation relation}, denoted as $\Downarrowi$ and defined on the target language commands,
with which the evaluation of $\code{loop\_fixpt\_noacc}$ command does not terminate by fixed-point check,
but repeats loops for a fixed number of iterations.
Then, we prove that the new evaluation rules defined by $\Downarrowi$
do not change the meaning of commands specified by the original evaluation relation $\Downarrowt$.

The semantics of $\code{loop\_fixpt\_noacc(n_+)\ C}$ in the target language enables the optimisation of the execution of the loop using a fixed point check.
The result score of $(\code{loop\_fixpt\_noacc}(n_+)\ C)$ entirely depends on the last execution of the body $C$.
Once the loop body $C$ no longer changes the state, repeating its execution $n_+$ times has no further effect:
the state remains unchanged until the loop terminates, and the result score remains the same as in the last execution.
We formalise this by defining an intermediate evaluation relation $\Downarrowi$, a variant of $\Downarrowt$.
It is defined on the same semantic domains and structures as $\Downarrowt$:
\begin{align*}
{\Downarrowi} \subseteq (\Comm_\mathit{tgt} \times \RDB \times \State_\mathit{tgt} \times \FAntiChain) \times (\State_\mathit{tgt} \times \Tensor_\R),
\end{align*}
but only replaces the original evaluation rule for the $\code{loop\_fixpt\_noacc}$ command with the following:
\begin{align*}
\infer{
  (\code{loop\_fixpt\_noacc}(n_+)\ C), D, \sigma, A \Downarrowi \sigma_{n_+},
  T_{n_+}
}{
  \sigma_0 = \sigma
  \qquad
  C, D, \sigma_k, A \Downarrowi \sigma_{k+1}, T_{k+1}
  \text{ for all $k \in \{0,\ldots,n_+{-}1\}$}
},
\end{align*}
while preserving the other rules of $\Downarrowt$.
The next theorem shows $\Downarrowi$ does not change the meaning of commands specified by the original $\Downarrowt$ relation. 

\begin{theorem}
\label{thm:target-intermediate-equivalence}
For all commands $C$ in the target language, RDBs $D$, states
$\sigma,\sigma' \in \State_\mathit{tgt}$, antichains $A$, and
real tensors $T \in \Tensor_\R$, we have
\[
C, D, \sigma, A \Downarrowt \sigma', T
\quad\text{if and only if}\quad
C, D, \sigma, A \Downarrowi \sigma', T.
\]
\end{theorem}
\begin{proof}
We denote $n = n_+$ for simplicity.
We will prove the only-if direction first.
The proof proceeds by induction on the depth of the derivation of $C, D, \sigma, A \Downarrowt \sigma', T$.
Since the rules of $\Downarrowt$ and $\Downarrowi$ only differ in
$\code{loop\_fixpt\_noacc}$, it is enough to prove the case where the last derivation
step applies the rule for $\code{loop\_fixpt\_noacc}$ (the proofs for the other cases
are trivial).

When the last derivation is the second rule for $\code{loop\_fixpt\_noacc}$,
the premises for the last derivation already has all premises for the rule of $\code{loop\_fixpt\_noacc}$ in $\Downarrowi$.
Therefore, we can trivially derive the same relation through the original rule.
When the last derivation is the first rule, the premises for the last derivation are as follow:
\[
\begin{array}{l}
  0 \leq m < n - 1
  \qquad
  \sigma_0 = \sigma
  \hfill
  \\[0.2em]
  C, D, \sigma_k, A \Downarrowt \sigma_{k+1}, T_{k+1}
  \text{ and }
  \sigma_k \neq \sigma_{k+1}
  \text{ for all $0 \leq k \leq m-1$}
  \qquad
  C, D, \sigma_m, A \Downarrowt \sigma_m, T_{m+1},
\end{array}
\]
where the inferred relation is as follows:
\[
  (\code{loop\_fixpt\_noacc}(n)\ C), D, \sigma, A 
  \Downarrowt 
  \sigma_m, T_{m+1}
\]
Because of the induction hypothesis, we have
\[
C, D, \sigma_k, A \Downarrowi \sigma_{k+1}, T_{k+1}
\text{  for all $0 \leq k \leq m-1$}
\qquad
C, D, \sigma_m, A \Downarrowi \sigma_m, T_{m+1}.
\]
If we define $\sigma_k = \sigma_m$, $T_k = T_{m+1}$ for all $m+1 \leq k \leq n-1$,
then $C, D, \sigma_k, A \Downarrowi \sigma_{k+1}, T_{k+1}$ for all $m \leq k \leq n-1$ because of the last premises.
Grouping them with the third premise, we can construct $C, D, \sigma_k, A \Downarrowi \sigma_{k+1}, T_{k+1}$ for all $0 \leq k \leq n-1$.
Therefore, we can derive $(\code{loop\_fixpt\_noacc}(n)\ C), D, \sigma, A  \Downarrowi \sigma_n, T_n$, where $\sigma_n = \sigma_m$ and $T_n = T_{m+1}$.

We can prove the if direction likewise.
Similarly to the only-if direction, the proof is by induction on the depth of the derivation of $\Downarrowi$, and it is enough to prove only the $\code{loop\_fixpt\_noacc}$ case.
The premises from the original rule of $\code{loop\_fixpt\_noacc}$ are as follow:
\[
  \sigma_0 = \sigma
  \qquad
  C, D, \sigma_k, A \Downarrowi \sigma_{k+1}, T_{k+1}
  \text{ for all $k \in \{0,\ldots,n-1\}$},
\]
where the inferred relation is as follows:
\[
  (\code{loop\_fixpt\_noacc}(n)\ C), D, \sigma, A 
  \Downarrowi
  \sigma_n, T_n
\]
From the induction hypothesis, we have
\[
  C, D, \sigma_k, A \Downarrowt \sigma_{k+1}, T_{k+1}
  \text{ for all $k \in \{0,\ldots,n-1\}$}
\]
If $\sigma_k \neq \sigma_{k+1}$ for all $0 \leq k < n-1$, we can use the second rule of $\code{loop\_fixpt\_noacc}$ in $\Downarrowt$ and obtain the following derivation:
$(\code{loop\_fixpt\_noacc}(n)\ C), D, \sigma, A \Downarrowt \sigma_n, T_n$.\\
Otherwise, $\sigma_k = \sigma_{k+1}$ for some $0 \leq k < n-1$.
We let $m$ be the first such $k$.
Then, we can use the first rule of $\code{loop\_fixpt\_noacc}$ in $\Downarrowt$ and obtain the following derivation:
\[
(\code{loop\_fixpt\_noacc}(n)\ C), D, \sigma, A \Downarrowt \sigma_m, T_{m+1}.
\]
From the determinism of $\Downarrowi$, we deduce $\sigma_k = \sigma_m$
for any $m \leq k \leq n$ and $T_k = T_{m+1}$ for any $m+1 \leq k \leq n$.
Therefore, $\sigma_m = \sigma_n$, $T_{m+1} = T_n$.
\end{proof}
\subsection{Couplings, Induced Relations, and Their Properties}
\label{subsec:couplings}

One of the key tools in our proof is the notion of coupling, which we define below.
\begin{definition}[Coupling]
A binary relation $R$ on $\Index$ is a {\bf coupling} if the following properties hold:
\begin{enumerate}
\item $R$ is a relation of a partial bijection, that is, for all
$i_0, i_1, i \in \Index$, 
\[
i_0 \, [R] \, i \;\land\; i_1 \, [R] \, i \implies i_0 = i_1
\quad
\text{and}
\quad
i \, [R] \, i_0 \;\land\; i \, [R] \, i_1 \implies i_0 = i_1.
\] 
\item For all $i_0,i_1 \in \Index$ and $\alpha \in \dom(i_0) \cap \dom(i_1)$, 
\[
i_0 \, [R] \, i_1
\implies 
i_0(\alpha) = i_1(\alpha).
\]
\end{enumerate} 
Here, we write $i_0 \, [R] \, i_1$ to mean $(i_0,i_1) \in R$.
\end{definition}

A coupling $R$ induces four relations.
The first defines a relation on states,
and two of the others define relations on tensors,
while the last specifies a relation on sets of indices:
\begin{align*}
\State_\mathit{tgt}(R) & \subseteq \State_\mathit{tgt} \times \State_\mathit{tgt}, \\
\sigma_0 \, [\State_\mathit{tgt}(R)] \, \sigma_1 &
\iff 
\begin{aligned}[t]
  & \text{for all variables $x^{\tau\dagger}$ and $i_0,i_1$ with $i_0 \, [R] \, i_1$}, \\
  & (i)\ i_0 \in \dom(\sigma_0(x^{\tau\dagger})){\uparrow} \iff i_1 \in \dom(\sigma_1(x^{\tau\dagger})){\uparrow}, \\ 
  & (ii)\ i_0 \in \dom(\sigma_0(x^{\tau\dagger})){\uparrow} \implies \extend(\sigma_0(x^{\tau\dagger}))(i_0) = \extend(\sigma_1(x^{\tau\dagger}))(i_1),
\end{aligned} \\ \\
\Tensor_{\db{\tau}}(R) & \subseteq \Tensor_{\db{\tau}} \times \Tensor_{\db{\tau}}, \\
T_0[\Tensor_{\db{\tau}}(R)]T_1 &
\iff
\begin{aligned}[t]
& \text{for all $i_0,i_1$ with $i_0 \, [R] \, i_1$,} \\
& (i)\ i_0 \in \dom(T_0) \iff i_1 \in \dom(T_1), \\ 
& (ii)\ i_0 \in \dom(T_0) \implies T_0(i_0) = T_1(i_1),
\end{aligned} \\ \\
\WTensor_{\db{\tau}}(R) & \subseteq \Tensor_{\db{\tau}} \times \Tensor_{\db{\tau}}, \\
T_0[\WTensor_{\db{\tau}}(R)]T_1 &
\iff
\begin{aligned}[t]
  & \text{for all $i_0,i_1$ with $i_0 \, [R] \, i_1$,} \\
  & (i)\ i_0 \in \dom(T_0){\uparrow} \iff i_1 \in \dom(T_1){\uparrow}, \\ 
  & (ii)\ i_0 \in \dom(T_0){\uparrow} \implies \extend(T_0)(i_0) = \extend(T_1)(i_1),
\end{aligned} \\ \\
\cP(R) & \subseteq \cP(\Index) \times \cP(\Index), \\
L_0 \, [\cP(R)] \, L_1 &
\iff
\begin{aligned}[t]
& \text{for all $i_0,i_1$ with $i_0 \, [R] \, i_1$,} \\
& (i)\ i_0 \in L_0{\uparrow} \iff i_1 \in L_1{\uparrow}, \\ 
& (ii)\ i_0 \in L_0{\uparrow}
\implies
\max(L_0 \cap i_0{\downarrow})  \, [R] \, \max(L_1 \cap i_1{\downarrow}).
\end{aligned}
\end{align*}

One way to get used to these relations is to understand
their common root: all the four relations are derived from 
the same general construction, which we describe next. Consider 
a set $V$
and a relation $S_R$ on $V$ (i.e., $S_R \subseteq V \times V$) 
that may depend on $R$. The general construction is the 
relation $[R \rightharpoonup S_R]$ on partial functions 
of type $\Index \rightharpoonup V$ that is defined as follows:
for all $f_0,f_1 : \Index \rightharpoonup V$,
\[
f_0 \, [R \rightharpoonup S_R] \, f_1
\iff 
\begin{aligned}[t]
  & \text{for all $i_0,i_1$ with $i_0 \, [R] \, i_1$,} \\
  & (i)\ i_0 \in \dom(f_0) \iff i_1 \in \dom(f_1), \\ 
  & (ii)\ i_0 \in \dom(f_0) \implies f_0(i_0) \, [S_R] \, f_1(i_1).
  \end{aligned}
\] 
For a set $W$, let $\Delta_W$ be the equality relation on $W$:
\[
\Delta_W = \{(w,w) : w \in W\},
\]
which is also known as the diagonal relation on $W$.
Then, the definitions of the four relations are equivalent to the following constructions:
\begin{align*}
  \sigma_0 \, [\State_\mathit{tgt}(R)] \, \sigma_1
  & \iff
  \extend(\sigma_0(x^{\tau\dagger})) \, [R \rightharpoonup \Delta_{\db{\tau}}] \, \extend(\sigma_1(x^{\tau\dagger}))
  \ \text{ for all $x^{\tau\dagger}$,}
  \\
  T_0 \, [\Tensor_{\db{\tau}}(R)] \, T_1
  & \iff
  T_0 \, [R \rightharpoonup \Delta_{\db{\tau}}] \,  T_1,
  \\
  T_0 \, [\WTensor_{\db{\tau}}(R)] \, T_1
  & \iff 
  \extend(T_0) \, [R \rightharpoonup \Delta_{\db{\tau}}] \,  \extend(T_1),
  \\
  L_0 \, [\cP(R)] \, L_1
  & \iff
  \extend(\mathrm{id}|_{L_0}) \, [R \rightharpoonup R] \, \extend(\mathrm{id}|_{L_1}),
\end{align*}
where $\mathrm{id}|_L$ is the identity function over $\Index$ restricted to $L \subseteq \Index$.

Before proceeding further, we re-state and prove \Cref{lem:update-renaming-extend} which will be used in this section.

\addtocounter{lemma}{-1}

\begin{lemma}
  \label{lem:update-renaming-extend-appendix}
  Let $\sigma$ be a state, $x^{\tau\dagger}$, $y^{\tau'\dagger}$ be variables, 
  $T$ be a tensor in $\Tensor_{\db{\tau}}$, and $\rho$ be a finite partial injection on $\Index$. Then,
  we have the following equations:
  \begin{align*}
  \extend(\sigma[x^{\tau\dagger} : T](y^{\tau'\dagger}))
  & = 
  \begin{cases}
  \extend(\sigma(x^{\tau\dagger}))[\extend(T)] & \text{if $y^{\tau'\dagger} \equiv x^{\tau\dagger}$},
  \\
  \extend(\sigma(y^{\tau'\dagger})) & \text{otherwise},
  \end{cases}
  \\
  \extend(\sigma\langle\rho\rangle(y^{\tau'\dagger}))
  & = \extend(\sigma(y^{\tau'\dagger}))[\extend(T_{y,\rho})], 
\end{align*}
where
  $T_{y,\rho}(i) = \sigma(y^{\tau'\dagger})(\rho^{-1}(i))$ if $i \in \image(\rho) \cap \rho(\dom(\sigma(y^{\tau'\dagger})))$,
  and it is undefined otherwise.
\end{lemma}
\begin{proof}
We first assume $y^{\tau'\dagger} \equiv x^{\tau\dagger}$ and $i \in \dom(T){\uparrow}$. We have to show
\begin{align}
  \label{eq:update-renaming-extend-eq1}
  \extend(\sigma[x^{\tau\dagger} : T](y^{\tau'\dagger}))(i) = \extend(T)(i).
\end{align}
By the definition of the update operator,
\[L = \dom(\sigma[x^{\tau\dagger} : T](y^{\tau'\dagger})) = (\dom(\sigma(x^{\tau\dagger})) \setminus \dom(T){\uparrow}) \cup \dom(T).\]
We claim that
\[\max(L \cap i{\downarrow}) = \max(\dom(T) \cap i{\downarrow}).\]
To see this, let $i' = \max(\dom(T) \cap i{\downarrow})$.
Then, $i' \in L \cap i{\downarrow}$.
Suppose for contradiction that there exists $i'' \in L \cap i{\downarrow}$ such that $i'' \sqsupseteq i'$ but $i'' \neq i'$.
Then, $i'' \in \dom(T){\uparrow}$, which implies $i'' \in \dom(T)$ since $i'' \in L$.
However, this contradicts the fact that $i'$ is the maximum of $\dom(T) \cap i{\downarrow}$,
since $i''$ is strictly greater than $i'$ while being contained in $\dom(T) \cap i{\downarrow}$.
Thus, $i' = \max(L \cap i{\downarrow})$, which proves the claim.
Now, to see \Cref{eq:update-renaming-extend-eq1}, we have
\begin{align*}
  \extend{(\sigma[x^{\tau\dagger} : T](y^{\tau'\dagger}))}(i)
  &= \sigma[x^{\tau\dagger} : T](y^{\tau'\dagger})(\max(L \cap i{\downarrow})) \\
  &= \sigma[x^{\tau\dagger} : T](y^{\tau'\dagger})(\max(\dom(T) \cap i{\downarrow})) \\
  &= T(\max(\dom(T) \cap i{\downarrow})) \\
  &= \extend(T)(i).
\end{align*}
For the other case when $y^{\tau'\dagger} \not\equiv x^{\tau\dagger}$ or $i \notin \dom(T){\uparrow}$, we have to show
\begin{align}
  \label{eq:update-renaming-extend-eq2}
  \extend(\sigma[x^{\tau\dagger} : T](y^{\tau'\dagger}))(i) = \extend(\sigma(y^{\tau'\dagger}))(i).
\end{align}
We claim that
\[i \notin \dom(T){\uparrow} \implies \max(L \cap i{\downarrow}) \notin \dom(T){\uparrow}.\]
Suppose for contradiction that $\max(L \cap i{\downarrow}) \in \dom(T){\uparrow}$.
Then, $\max(L \cap i{\downarrow}) \in \dom(T)$ since it belongs to $L$.
However, this cannot happen since $\max(L \cap i{\downarrow})$ is also contained in $i{\downarrow}$,
and $\dom(T) \cap i{\downarrow} = \emptyset$ by the assumption that $i \notin \dom(T){\uparrow}$.
Thus, the claim holds.
Now, we can show \Cref{eq:update-renaming-extend-eq2} as follows:
\begin{align*}
  \extend(\sigma[x^{\tau\dagger} : T](y^{\tau'\dagger}))(i)
  &= \sigma[x^{\tau\dagger} : T](y^{\tau'\dagger})(\max(L \cap i{\downarrow})) \\
  &= \sigma(y^{\tau'\dagger})(\max(L \cap i{\downarrow})) \\
  &= \extend(\sigma(y^{\tau'\dagger}))(i),
\end{align*}
where the second equality follows from that $y^{\tau'\dagger} \not\equiv x^{\tau\dagger}$ or $\max(L \cap i{\downarrow}) \notin \dom(T){\uparrow}$.
The second result in the lemma directly follows from that
\begin{align*}
\sigma\langle\rho\rangle(y^{\tau'\dagger}) = \sigma[y^{\tau'\dagger} : T_{y,\rho}](y^{\tau'\dagger})
\end{align*}
for every variable $y^{\tau'\dagger}$.
\end{proof}

Also, we note a few basic facts about these induced relations. 

\begin{lemma}
  \label{lem:indexset-R-membership}
  Let $R$ be a coupling, and $L_0,L_1$ be sets of indices. 
  If $L_0 \, [\cP(R)] \, L_1$, then we have
  the following equivalence for all $i_0,i_1 \in \Index$ with $i_0 \, [R] \, i_1$:
  \[
  i_0 \in L_0 \iff i_1 \in L_1.
  \]
\end{lemma}
\begin{proof}
  Let $R$ be a coupling, $L_0,L_1$ be sets of indices, and $i_0,i_1$ be indices
  such that 
  \[
  L_0 \, [\cP(R)] \, L_1
  \quad\text{and}\quad
  i_0 \, [R] \, i_1.
  \]
  We will show that $i_0 \in L_0$ implies $i_1 \in L_1$; the other implication
  can be proved similarly. Assume that $i_0 \in L_0$. 
  Then, $i_0 \in L_0{\uparrow}$, which implies that $i_1 \in L_1{\uparrow}$.
  Furthermore, 
  \[
  \max(L_0 \cap i_0{\downarrow}) \, [R] \, \max(L_1 \cap i_1{\downarrow}).
  \] 
  But $i_0 = \max(L_0 \cap i_0{\downarrow})$
  since $i_0 \in L_0$.
  As a result, $i_1 = \max(L_1 \cap i_1{\downarrow})$
  since $i_0 \, [R] \, i_1$ and $R$ is a partial bijection, as $R$ is a coupling.
  Thus, $i_1 \in L_1$.
\end{proof}
\begin{corollary}
  \label{cor:Tz-R-from-indexset-R}
  Let $R$ be a coupling, and $L_0,L_1$ be finite sets of indices. Then,
  \[
  L_0 \, [\cP(R)] \, L_1 \implies
  T^z_{L_0} \, [\Tensor_\R(R)] \, T^z_{L_1}.
  \]
\end{corollary}

The following lemma asserts that if $\dom(f_0) \, [\cP(R)] \, \dom(f_1)$,
then $\extend$ preserves the $[R \rightharpoonup \Delta_V]$ relation between $f_0$ and $f_1$.
\begin{lemma}
  \label{lem:extend-and-R}
  Let $R$ be a coupling, and $f_0,f_1$ be partial functions from $\Index$ to a set $V$. Then,
  if $\dom(f_0) \, [\cP(R)] \, \dom(f_1)$, we have
  \[
  f_0 \, [R \rightharpoonup \Delta_V] \, f_1
  \implies
  \extend(f_0) \, [R \rightharpoonup \Delta_V] \, \extend(f_1).
  \]
\end{lemma}
\begin{proof}
  Consider arbitrary indices $i_0,i_1$ with $i_0 \, [R] \, i_1$.
  First, $\dom(f_0) \, [\cP(R)] \, \dom(f_1)$ directly implies that
  \[
    i_0 \in \dom(\extend(f_0)) = \dom(f_0){\uparrow} \iff
    i_1 \in \dom(\extend(f_1)) = \dom(f_1){\uparrow}.
  \]
  Next, we further assume that $i_0 \in \dom(\extend(f_0))$ and $i_1 \in \dom(\extend(f_1))$,
  and show that 
  \[
    f_0(\max(\dom(f_0) \cap i_0{\downarrow})) = 
    f_1(\max(\dom(f_1) \cap i_1{\downarrow})),
  \] 
  which gives the desired $\extend(f_0)(i_0) = \extend(f_1)(i_1)$.
  This directly follows from that
  \[
    f_0 \, [R \rightharpoonup \Delta_V] \, f_1
    \quad\text{and}\quad
    \max(\dom(f_0) \cap i_0{\downarrow}) \, [R] \, \max(\dom(f_1) \cap i_1{\downarrow}).
  \]
\end{proof}
\begin{corollary}
  \label{cor:extend-R-tensor}
  If $\tau$ is either $\intt$ or $\real$,
  and $\dom(T_0) \, [\cP(R)] \, \dom(T_1)$ for $T_0, T_1 \in \Tensor_{\db{\tau}}$,
  we have
  \[
  T_0 \, [\Tensor_{\db{\tau}}(R)] \, T_1
  \implies
  T_0 \, [\WTensor_{\db{\tau}}(R)] \, T_1.
  \]
\end{corollary}

\begin{lemma}
  \label{lem:tensor-R-oplus}
  For all real tensors $T_0,T_1,T_0',T_1' \in \Tensor_\R$ and all couplings $R$,
  we have
  \[
  T_0 \, [\Tensor_\R(R)] \, T_1 \land T_0' \, [\Tensor_\R(R)] \, T_1'
  \implies
  (T_0 \oplus T_0') \, [\Tensor_\R(R)] \, (T_1 \oplus T_1').
  \]
\end{lemma}
\begin{proof}
Pick $i_0,i_1 \in \Index$ such that $i_0 \, [R] \, i_1$. Then, the first 
condition in the definition of $\Tensor_\R(R)$ holds as shown below:
\begin{align*}
  i_0 \in \dom(T_0 \oplus T_0')
  & \iff i_0 \in \dom(T_0) \lor i_0 \in \dom(T_0')
  \\
  & \iff i_1 \in \dom(T_1) \lor i_1 \in \dom(T_1')
  \\
  & \iff i_1 \in \dom(T_1 \oplus T_1'),
\end{align*}
where the second equivalence follows from the assumption that 
\[
T_0 \, [\Tensor_\R(R)] \, T_1
\quad\text{and}\quad
T_0' \, [\Tensor_\R(R)] \, T_1'.
\]
To prove the other condition in the definition of $\Tensor_\R(R)$, 
assume that $i_0 \in \dom(T_0 \oplus T_0')$. If $i_0 \in \dom(T_0) \cap \dom(T_0')$,
then 
\[
i_1 \in \dom(T_1) \cap \dom(T_1'),
\quad T_0(i_0) = T_1(i_1)
\quad \text{and}
\quad T_0'(i_0) = T_1'(i_1),
\]
and so 
\[
(T_0 \oplus T_0')(i_0) 
=  T_0(i_0) + T_0'(i_0)
=  T_1(i_1) + T_1'(i_1) 
= (T_1 \oplus T_1')(i_1),
\]
as desired.
If $i_0 \in \dom(T_0) \setminus \dom(T_0')$, then $i_1 \in \dom(T_1) \setminus \dom(T_1')$
and $T_0(i_0) = T_1(i_1)$, and these facts imply the desired equality:
\[
(T_0 \oplus T_0')(i_0) 
=  T_0(i_0)
=  T_1(i_1) 
= (T_1 \oplus T_1')(i_1).
\]
The remaining case $i_0 \in \dom(T_0') \setminus \dom(T_0)$ can be 
handled similarly to the previous case.
\end{proof}

\begin{lemma}
  \label{lem:state-R-tensor-update}
  Let $R$ be a coupling, and $x^{\tau\dagger}$ be a variable. 
  Consider states $\sigma_0,\sigma_1 \in \State_\mathit{tgt}$ and tensors
  $T_0,T_1 \in \Tensor_{\db{\tau}}$ such that 
  \[
  \sigma_0 \, [\State_\mathit{tgt}(R)] \, \sigma_1
  \quad\text{and}\quad
  T_0 \, [\WTensor_{\db{\tau}}(R)] \, T_1.
  \]
  Then,  
  \[
  \sigma_0[x^{\tau\dagger} : T_0]
 \, [\State_\mathit{tgt}(R)] \,
  \sigma_1[x^{\tau\dagger} : T_1].
  \] 
\end{lemma}
\begin{proof}
  Let $x^{\tau\dagger}$ be a variable.
  Consider a coupling $R$, states $\sigma_0,\sigma_1$, and tensors $T_0,T_1$,
  such that
  \[
  \sigma_0 \, [\State_\mathit{tgt}(R)] \, \sigma_1
  \quad\text{and}\quad
  T_0 \, [\WTensor_{\db{\tau}}(R)] \, T_1.
  \] 
  Pick an arbitrary variable $y^{\tau'\dagger}$ and 
  arbitrary indices $i_0,i_1$ such that $i_0 \, [R] \, i_1$.
  Then, we have
  \begin{align*}
    \extend(\sigma_0[x^{\tau\dagger} : T_0](y^{\tau'\dagger}))(i_0)
    & {} =
    \begin{cases}
      \extend(\sigma_0(x^{\tau\dagger}))[\extend(T_0)](i_0) & \text{if $y^{\tau'\dagger} \equiv x^{\tau\dagger}$,} \\
      \extend(\sigma_0(y^{\tau'\dagger}))(i_0) & \text{otherwise,}
    \end{cases}
    \\
    & {} =
    \begin{cases}
      \extend(T_0)(i_0) & \text{if $y^{\tau'\dagger} \equiv x^{\tau\dagger}$ and $i_0 \in \dom(T_0){\uparrow}$,} \\
      \extend(\sigma_0(x^{\tau\dagger}))(i_0) & \text{if $y^{\tau'\dagger} \equiv x^{\tau\dagger}$ and $i_0 \not\in \dom(T_0){\uparrow}$,} \\
      \extend(\sigma_0(y^{\tau'\dagger}))(i_0) & \text{otherwise,}
    \end{cases}
    \\
    & {} =
    \begin{cases}
      \extend(T_1)(i_1) & \text{if $y^{\tau'\dagger} \equiv x^{\tau\dagger}$ and $i_1 \in \dom(T_1){\uparrow}$,} \\
      \extend(\sigma_1(x^{\tau\dagger}))(i_1) & \text{if $y^{\tau'\dagger} \equiv x^{\tau\dagger}$ and $i_1 \not\in \dom(T_1){\uparrow}$,} \\
      \extend(\sigma_1(y^{\tau'\dagger}))(i_1) & \text{otherwise,}
    \end{cases}
    \\
    & {} =
    \begin{cases}
      \extend(\sigma_1(x^{\tau\dagger}))[\extend(T_1)](i_1) & \text{if $y^{\tau'\dagger} \equiv x^{\tau\dagger}$,} \\
      \extend(\sigma_1(y^{\tau'\dagger}))(i_1) & \text{otherwise,}
    \end{cases}
    \\
    & {} =
    \extend(\sigma_1[x^{\tau\dagger} : T_1](y^{\tau'\dagger}))(i_1).
  \end{align*}
  The first and the last equalities follow from \Cref{lem:update-renaming-extend-appendix}.
  The third equality holds because $i_0 \, [R] \, i_1$ and
  $\sigma_0 \, [\State_\mathit{tgt}(R)] \, \sigma_1$ and
  $T_0 \, [\WTensor_{\db{\tau}}(R)] \, T_1$.
\end{proof}

When $R = \Delta_L$ for a subset $L$ of $\Index$,
we write $=_L$ for $\State_\mathit{tgt}(R)$.
We spell out the definition of $=_L$ explicitly below since it is used frequently:
\[
\sigma_0 =_L \sigma_1 \iff
  \extend(\sigma_0(x^{\tau\dagger}))(i)
  = \extend(\sigma_1(x^{\tau\dagger}))(i)
  \text{ for all variables $x^{\tau\dagger}$ and indices $i \in L$}. 
\]
Intuitively, it says that the $\extend$ functions represented by
$\sigma_0(x^{\tau\dagger})$ and $\sigma_1(x^{\tau\dagger})$ coincide on $L$.

\begin{corollary}[Effects of Variable Writes]
  \label{cor:effects-write}
  For all states $\sigma \in \State_\mathit{tgt}$, variables $x^{\tau\dagger}$,
  and tensors $T \in \Tensor_{\db{\tau}}$ such that $L = \dom(T)$, we have
  \[
  \sigma[x^{\tau\dagger} : T] =_{(L{\uparrow})^c} \sigma.
  \]
\end{corollary}
\begin{proof}
  It follows from \Cref{lem:state-R-tensor-update} with the following facts:
  \[
  \sigma =_{(L{\uparrow})^c} \sigma
  \quad\text{and}\quad
  T \, [\WTensor_{\db{\tau}}(\Delta_{(L{\uparrow})^c})] \, T_\emptyset
  \quad\text{and}\quad
  \sigma[x^{\tau\dagger} : T_\emptyset] = \sigma.
  \]
\end{proof}

\begin{lemma}[Updated Map Entries: Commands]
  \label{lem:updated-indices-cmd}
  Let $C$ be a command in the target language, $D \in \RDB$, and $\sigma, \sigma' \in \State_\mathit{tgt}$.
  Suppose that $C, D, \sigma, A \Downarrowi \sigma', T$.
  Then,   
  \[
    \sigma =_{(A{\uparrow})^c} \sigma'
    \qquad\text{and}\qquad
    \dom(T) = A.
  \]
\end{lemma}
Intuitively, the lemma expresses that if an entry of the map stored in a variable is updated during execution,
its index should be in the upward closure of the given $A$ for the execution, and also that
the tensor recording the outcome of scoring during execution should have domain $A$. 
\begin{proof}
  We prove the lemma by induction on the structure of $C$.
  We do case analysis on $C$, and prove that the two conditions of the lemma hold for each case.
  \begin{itemize}[itemsep=5pt]
  \item {\bf Cases of 
    $C \equiv (x^\lreal := \code{fetch}(I))$, 
    $C \equiv (x^{\lintt} := \code{lookup\_index}(\alpha))$,
    $C \equiv (x^\lintt := Z)$, and 
    $C \equiv (x^\lreal := E)$:} \\
    We have $\dom(T) = \dom(T^z_A) = A$. Thus, the second condition of the lemma holds.
    For the other condition, we note that
    by the rules for these commands in the semantics,
    $\sigma' = \sigma[x^\lreal : T']$ or $\sigma' = \sigma[x^\lintt : T']$ for some tensor $T'$ with $\dom(T') = A$.
    Thus, by \Cref{cor:effects-write}, 
    we have the desired $\sigma =_{(A{\uparrow})^c} \sigma'$.

  \item {\bf Cases of $C \equiv (\code{score}(E))$ and $C \equiv (\code{skip})$:} \\
    In these two cases, $\dom(T) = A$, and so the second condition of the lemma holds.
    The other condition also holds because the input state $\sigma$ is not modified in these cases.
  
  \item {\bf Case of $C \equiv (C_1;C_2)$:} \\
    The second condition follows from the induction hypothesis on $C_1$ and $C_2$,
    and the fact that $\dom(T_1 \oplus T_2) = \dom(T_1) \cup \dom(T_2)$ for any $T_1, T_2 \in \Tensor_\R$.
    The other condition holds because of the induction hypothesis on $C_1$ and $C_2$
    and the transitivity of the $=_{(A{\uparrow})^c}$ relation.

  \item {\bf Case of $C \equiv (\code{ifz}\ Z\ C_1\ C_2$):} \\
    Let $A_1 = \{i \in A : n_i = 0\}$ and $A_2 = \{i \in A : n_i \neq 0\}$ 
    where $n_i = \db{Z}(\sigma,i)$ for all $i \in A$.
    By the induction hypothesis on $C_1$ and $C_2$, we have
    \begin{align*}
    &
    C_1, D, \sigma, A_1 \Downarrowi \sigma_1, T_1,
    &&
    \dom(T_1) = A_1,
    &&
    \sigma =_{(A_1{\uparrow})^c} \sigma_1,
    \\
    &
    C_2, D, \sigma_1, A_2 \Downarrowi \sigma_2, T_2,
    &&
    \dom(T_2) = A_2,
    &&
    \sigma_1 =_{(A_2{\uparrow}^c)} \sigma_2,
    \end{align*}
    for $\sigma' = \sigma_2$ and $T = T_1 \oplus T_2$.
    The second condition of the lemma follows from these relationships as shown below:
    \[
    \dom(T) = \dom(T_1 \oplus T_2) = \dom(T_1) \cup \dom(T_2) = A_1 \cup A_2 = A.
    \]
    The other condition also holds because $(A{\uparrow})^c \subseteq (A_1{\uparrow})^c$ and $(A{\uparrow})^c \subseteq (A_2{\uparrow})^c$ imply
    \begin{align*}
    \sigma =_{(A{\uparrow})^c} \sigma_1 =_{(A{\uparrow})^c} \sigma_2 = \sigma',
    \end{align*}
    and the relation $=_{(A{\uparrow})^c}$ is transitive.

  \item {\bf Cases of $C \equiv (\code{for}\ x^\lintt \ \code{in} \ \code{range}(n)\ \code{do}\ C')$:} \\
    $C$ in this case can be expressed as a sequence of $2n$ commands.
    By repeating our proofs for assignment and sequential composition and using the induction hypothesis on $C'$,
    we can prove that the lemma holds in this case.

  \item {\bf Case of $C \equiv (\code{extend\_index}(\alpha,n)\ C')$:} \\
    Let 
    \begin{align*}
      A_0 &= \{i \concat [(\alpha,k)] : i \in A,\ k \in \{0,\ldots,n-1\}\}, \\
      \rho_0 &= [i \concat [(\alpha,n-1)] \mapsto i : i \in A]. 
    \end{align*}
    By the semantics of the $\code{extend\_index}$ command, there exist
    $\sigma_0 \in \State_\mathit{tgt}$ and $T_0 \in \Tensor_\R$ such that
    \[
      C', D, \sigma, A_0 \Downarrowi \sigma_0, T_0,
      \quad
      \sigma' = \sigma_0\langle\rho_0\rangle,
      \quad
      T = \left[i \mapsto \sum_{k=0}^{n-1} T_0(i\concat [(\alpha,k)]) \;:\; i \in A\right].
    \]
    By the induction hypothesis on $C'$,
    \[
    \sigma =_{(A_0{\uparrow})^c} \sigma_0
    \quad\text{and}\quad
    \dom(T_0) = A_0.
    \]
    Note that the second relationship from above ensures that 
    the rule for the $\code{extend\_index}$ command in the operational semantics
    never experiences the failure due to the undefinedness of the tensor $T_0$ 
    at some $i\concat [(\alpha,k)]$, as we explained right after the definition 
    of the operational semantics.    
    The second condition required by the lemma holds because $\dom(T) = A$.
    To prove the other condition, we note that 
    for all variables $x^{\tau\dagger}$ and indices $j \in (A{\uparrow})^c$,
    \begin{align*}
      \extend(\sigma'(x^{\tau\dagger}))(j)
      & {} = 
      \extend(\sigma_0\langle \rho_0 \rangle (x^{\tau\dagger}))(j) \\
      & {} =
      \extend(\sigma_0(x^{\tau\dagger}))(j) \\
      & {} =
      \extend(\sigma(x^{\tau\dagger}))(j).
    \end{align*}
    In the derivation, the second equality holds because $j \in (A{\uparrow})^c = (\image(\rho_0){\uparrow})^c$,
    and the third equality holds because
    $\sigma =_{(A_0{\uparrow})^c} \sigma_0$ and $(A{\uparrow})^c \subseteq (A_0{\uparrow})^c$.
  
  \item {\bf Case of $C \equiv \code{shift}(\alpha)$:} \\
    In this case, $\dom(T) = \dom(T^z_A) = A$, as required by the second condition of the lemma.
    To prove the other condition, pick an arbitrary variable $x^{\tau\dagger}$ and index $i \in (A{\uparrow})^c$.
    Note that $\sigma' = \sigma\langle\rho\rangle$ 
    for some $\rho$ such that $\image(\rho) \subseteq A$.
    \begin{align*}
      \extend(\sigma'(x^{\tau\dagger}))(i)
      & {} =
      \extend(\sigma\langle\rho\rangle(x^{\tau\dagger}))(i) \\
      & {} =
      \extend(\sigma(x^{\tau\dagger}))(i),
    \end{align*} 
    where the second equality holds because $i \in (A{\uparrow})^c \subseteq (\image(\rho){\uparrow})^c$.

  \item {\bf Case of $C \equiv (\code{loop\_fixpt\_noacc}(n)\ C')$:} \\
    By the semantics of the $\code{loop\_fixpt\_noacc}$ command,
    there exist $\sigma_k \in \State_\mathit{tgt}$ and $T_k \in \Tensor_\R$ for $k = 0,\ldots,n-1$ such that
    \begin{align*}
    \sigma_0 = \sigma, &&
    \sigma_n = \sigma', &&
    T_n = T, &&
    C', D, \sigma_k, A \Downarrowi \sigma_{k+1}, T_{k+1}
    \quad\text{for $k = 0,\ldots,n-1$}.
    \end{align*}
    By the induction hypothesis on $C'$ and the transitivity of the $=_{(A{\uparrow})^c}$ relation,
    we have $\sigma_k =_{(A{\uparrow})^c} \sigma_{k+1}$ for $k = 0,\ldots,n-1$,
    and thus $\sigma =_{(A{\uparrow})^c} \sigma'$ and $\dom(T) = A$.
  \end{itemize}
\end{proof}

\begin{corollary}[Identity Property under Empty Set of Indices]
  \label{lem:emp-identity}
  Let $C$ be a command in the target language, $D \in \RDB$, and $\sigma, \sigma' \in \State_\mathit{tgt}$.
  Then, if $C, D, \sigma, \emptyset \Downarrowi \sigma', T$, then
  \[
  \sigma =_{\Index} \sigma' \quad\text{and}\quad T = T_\emptyset
  \]
  where $T_\emptyset$ is the empty tensor, i.e., a tensor with the empty domain.
\end{corollary}
\begin{proof}
  This directly follows from \Cref{lem:updated-indices-cmd} since
  \[
  (\emptyset{\uparrow})^c = \Index
  \quad\text{and}\quad
  \dom(T) = \emptyset \iff T = T_\emptyset.
  \]
\end{proof}
\subsection{Embedding of Source-Language Commands into Target-Language Commands}
\label{subsec:simple-embedding-commands}

The simple straightforward way to embed expressions and commands in the 
source language to those in the target language is simply to change
the type of each variable so that a variable $x$ of type $\tau$
in the source language becomes the variable with the same name but of
the lifted type $\tau\dagger$ in the target language. We denote the outcomes
of this embedding using the $-'$ notation. For instance, $E'$ and $C'$
denote the embeddings of an expression $E$ and a command $C$ in the source language.

In this subsection, we prove that this embedding preserves the behaviour of the commands.
Specifically, we show that if the initial states of a source language command and its embedding in the target language are equivalent,
then their final states and scores are also equivalent after evaluating the commands using the evaluation relations $\Downarrow$ and $\Downarrowi$, respectively.

Consider the following two relations, the first between states of the source language and those of the target language,
and the next between real numbers and real-valued tensors, both denoted with the
same notation $\sim$:
\begin{align*}
\sim & \subseteq \State_\mathit{src} \times \State_\mathit{tgt},
&
\sigma \sim \sigma' & 
\iff 
\begin{aligned}[t]
& \text{$\sigma(x^\tau) = \extend(\sigma'(x^{\tau\dagger}))([])$}
\\
& \text{for all variables $x^\tau$ in the source languages,} 
\end{aligned}
\\
\sim & \subseteq \R \times \Tensor_\R,
&
r \sim T & \iff r = \sum_{i \in \dom(T)} T(i).
\end{align*}

The next proposition says that the embedding 
preserves the semantics of expressions.
\begin{proposition}[Embedding of Expressions]
  \label{prop:exp-cor}
    Let $\sigma \in \State_\mathit{src}$ and $\sigma' \in \State_\mathit{tgt}$ such that 
    $\sigma \sim \sigma'$. Then, for all integer expressions $Z$, real
    expressions $E$, and index expressions $I$ in the source language, 
    if we let $Z'$, $E'$, and $I'$ be the ones obtained from $Z$, $E$, and $I$
    by replacing all variables $x^{\tau}$ in them with $x^{\tau\dagger}$, then
    \[
    \db{Z'}(\sigma', []) = \db{Z}(\sigma),
    \qquad
    \db{E'}(\sigma', []) = \db{E}(\sigma),
    \qquad
    \db{I'}(\sigma', []) = \db{I}(\sigma),
    \]
    where the left-hand side of each equation uses the semantics of the target language,
    and the right-hand side uses 
    the semantics of the source language.   
\end{proposition}
\begin{proof}
    The proof is by induction on the structures of $E$, $Z$, and $I$.
    In the proof, we use the primed version of an expression of any type
    to mean the one obtained from it by replacing all variables $x^{\tau}$
    with $x^{\tau\dagger}$.
    \begin{itemize}[itemsep=5pt]
      \item {\bf Case of $Z \equiv n$
        and $E \equiv r$}: \\
        The desired equality follows immediately
        from the semantics of the source and target languages.
      \item {\bf Case of $Z \equiv x^{\intt}$
        and $E \equiv x^{\real}$}: \\
        We derive the desired equality as follows:
        \[
        \db{x^{\tau\dagger}}(\sigma',[]) =
        \extend(\sigma'(x^{\tau\dagger}))([]) = 
        \sigma(x^\tau) = \db{x^{\tau}}(\sigma)
        \]
        for $\tau \in \{\intt, \real\}$.
        Here the second equality follows from the assumption 
        that $\sigma \sim \sigma'$, and the other equalities from
        the semantics of the source and target languages.
      \item {\bf Case of $Z \equiv \op_i(Z_1,\ldots,Z_n,E_1,\ldots,E_m)$
        and $E \equiv \op_r(Z_1,\ldots,Z_n,E_1,\ldots,E_m)$}: \\
        For $\op \in \{\op_i, \op_r\}$, we show the desired equality as follows:
        \begin{align*}
        \db{\op(Z'_1,\ldots,Z'_n,E'_1,\ldots,E'_m)}(\sigma',[]) 
        & {} =
        (\db{\op}_a)(
            \begin{aligned}[t]
                & \db{Z'_1}(\sigma',[]),\ldots,\db{Z'_m}(\sigma',[]),
                \\
                & \db{E'_1}(\sigma',[]),\ldots,\db{E'_n}(\sigma',[]))
            \end{aligned}
        \\
        & {} =
        (\db{\op}_a)(
            \begin{aligned}[t]
                & \db{Z_1}(\sigma),\ldots,\db{Z_m}(\sigma),
                \\
                & \db{E_1}(\sigma),\ldots,\db{E_n}(\sigma))
            \end{aligned}
        \\
        & {} = 
        \db{\op(Z_1,\ldots,Z_n,E_1,\ldots,E_m)}(\sigma).
        \end{align*}
        The second and third equalities use the induction hypothesis,
        and the other equalities follow from the semantics of the source
        and target languages.
      \item {\bf Case of $I \equiv [(\alpha_1, Z_1);\, \ldots;\, (\alpha_n, Z_n)]$}: \\
        We prove the goal equality as shown below:
        \begin{align*}
        \db{[(\alpha_1, Z'_1);\, \ldots;\, (\alpha_n, Z'_n)]}(\sigma',[])
        & {} =
        [(\alpha_1,\db{Z'_1}(\sigma',[]));\ldots;(\alpha_n,\db{Z'_n}(\sigma',[]))]
        \\
        & {} =
        [(\alpha_1,\db{Z_1}(\sigma));\ldots;(\alpha_n,\db{Z_n}(\sigma))]
        \\
        & {} =
        \db{[(\alpha_1, Z_1);\, \ldots;\, (\alpha_n, Z_n)]}(\sigma).
        \end{align*}
        The first equality follows from the induction hypothesis and the
        semantics of the index expression in the target language. 
        The second equality uses the induction hypothesis,
        and the last equality the semantics of the source language.
  \end{itemize}
\end{proof}

We next show that the embedding preserves the semantics of commands.
\begin{lemma}[Preservation of Equivalence by Update]
  \label{lem:upd-cor}
  Let $\sigma \in \State_\mathit{src}$ and $\sigma' \in \State_\mathit{tgt}$.
  If $\sigma \sim \sigma'$, then for all variables $x^\tau$ in the source language 
  and all $r \in \R$, we have
  \[
  \sigma [ x^\tau \mapsto r ]
  \sim
  \sigma' [ x^{\tau\dagger} : T ],
  \quad \text{ for $T = [ [] \mapsto r ]$}.
  \]
\end{lemma}
\begin{proof}
  Let $\sigma$, $\sigma'$, $x^{\tau}$, $r$, and $T$ be as in the lemma.
  Assume that $\sigma \sim \sigma'$.
  For all variables $y^{\tau'}$ in the source language, we have
  \begin{align*}
  \sigma[x^\tau \mapsto r](y^{\tau'}) 
  & {} =
  \begin{cases}
    r & \text{if $y^{\tau'} \equiv x^\tau$}
    \\
    \sigma(y^{\tau'}) & \text{otherwise}
  \end{cases}
  \\
  & {} =  
  \begin{cases}
    \extend(\sigma'[x^{\tau\dagger} : T](x^{\tau\dagger}))([]) & \text{if $y^{\tau'} \equiv x^\tau$}
    \\
    \extend(\sigma'(y^{\tau'\dagger}))([]) & \text{otherwise}
  \end{cases}
  \\
  & {} =
  \extend(\sigma'[x^{\tau\dagger} : T](y^{\tau'\dagger}))([]).
  \end{align*}
\end{proof}

\begin{proposition}
  \label{prop:sem-equiv}
  Let $C$ be a command in the source language,
  and $C'$ be the command in the target language obtained from $C$
  by replacing all variables $x^{\tau}$ with $x^{\tau\dagger}$.
  Then, for all $\sigma_0 \in \State_\mathit{src}$ and $\sigma_0' \in \State_\mathit{tgt}$ 
  such that $\sigma_0 \sim \sigma_0'$,
  we have the following two properties:
  \begin{enumerate}
  \item
    if $C, D, \sigma_0 \Downarrows \sigma_1, r$ for some $\sigma_1 \in \State_\mathit{src}$ and $r \in \R$,
    then there exist $\sigma_1' \in \State_\mathit{tgt}$ and $T \in \Tensor_\R$ such that
    \[
    C', D, \sigma_0', \{[]\} \Downarrowi \sigma_1', T,
    \qquad
    \sigma_1 \sim \sigma_1',
    \qquad
    r \sim T,
    \]
  \item
    if $C', D, \sigma_0', \{[]\} \Downarrowi \sigma_1', T$ for some $\sigma_1' \in \State_\mathit{tgt}$ and $T \in \Tensor_\R$,
    then there exist $\sigma_1 \in \State_\mathit{src}$ and $r \in \R$ such that 
    \[
    C, D, \sigma_0 \Downarrows \sigma_1, r,
    \qquad
    \sigma_1 \sim \sigma_1',
    \qquad
    r \sim T.
    \]
  \end{enumerate}
\end{proposition}
\begin{proof}
  Let $A = \{[]\} \in \FAntiChain$.
  We prove the proposition by structural induction on $C$.
  \begin{itemize}[itemsep=5pt]
  \item {\bf Case of $C \equiv (x^\real := \code{fetch}(I))$
  and $C' \equiv (x^\lreal := \code{fetch}(I'))$}:
    \begin{enumerate}
    \item Assume that $C, D, \sigma_0 \Downarrows \sigma_1, r$ for some $\sigma_1$ and $r$.
    By the semantics of the source language,
    \begin{align*}
    \sigma_1 = \sigma[x^\real \mapsto D(\db{I}(\sigma_0))],
    &&
    r = 0.0.
    \end{align*}
    By \Cref{prop:exp-cor}, the assumption $\sigma_0 \sim \sigma_0'$ entails that
    $\db{I'}(\sigma_0',[]) = \db{I}(\sigma_0)$.
    By the semantics of the target language,
    \begin{align*}
    C', D, \sigma_0', A \Downarrowi \sigma_1', T, &&
    \sigma_1' = \sigma_0'[x^\lreal : [[] \mapsto D(\db{I}(\sigma_0))]], &&
    T = T^z_A.
    \end{align*}
    \Cref{lem:upd-cor} entails that $\sigma_1 \sim \sigma_1'$.
    Last, $r \sim T$ clearly holds.

    \item Assume that $C', D, \sigma_0', A \Downarrowi \sigma_1', T$ for some $\sigma_1'$ and $T$.
    By the semantics of the target language,
    \begin{align*}
    \sigma_1' = \sigma[x^\lreal : [[] \mapsto D(\db{I'}(\sigma',[]))]], &&
    T = T^z_A.
    \end{align*}
    By \Cref{prop:exp-cor}, $\db{I'}(\sigma',[]) = \db{I}(\sigma)$.
    By the semantics of the source language,
    \begin{align*}
    C, D, \sigma \Downarrows \sigma_0', r, &&
    \sigma_0' = \sigma[x^\real \mapsto D(\db{I}(\sigma))], &&
    r = 0.0
    \end{align*}
    \Cref{lem:upd-cor} entails that $\sigma_1 \sim \sigma_1'$.
    Last, $r \sim T$ clearly holds.
    \end{enumerate}

  \item {\bf Case of $C \equiv (\code{score}(E))$ and $C' \equiv (\code{score}(E'))$}:
    \begin{enumerate}
    \item Assume that $C, D, \sigma_0 \Downarrows \sigma_1, r$ for some $\sigma_1$ and $r$. 
    By the semantics of the source language,
    \begin{align*}
    \sigma_0 = \sigma_1 &&
    r = \db{E}(\sigma_0).
    \end{align*}
    By \Cref{prop:exp-cor}, $\db{E'}(\sigma_0',[]) = \db{E}(\sigma_0) = r$.
    Let $\sigma_1' = \sigma_0'$ and $T = [ [] \mapsto r ]$.
    Then, we have
    \begin{align*}\
    C', D, \sigma_0', A \Downarrowi \sigma_1', T, &&
    \sigma_1 \sim \sigma_1', &&
    r \sim T.
    \end{align*}

    \item Assume that $C', D, \sigma_0', A \Downarrowi \sigma_1', T$ for some $\sigma_1'$ and $T$.
    By the semantics of the target language
    \begin{align*}
    \sigma_0' = \sigma_1', &&
    T = [ [] \mapsto \db{E'}(\sigma_0',[])].
    \end{align*}
    By \Cref{prop:exp-cor}, $\db{E}(\sigma_0) = \db{E'}(\sigma_0',[])$.
    Define $\sigma_1 = \sigma_0$ and $r = \db{E}(\sigma_0)$.
    Then, we have
    \begin{align*}
    C, D, \sigma_0 \Downarrows \sigma_1, r, &&
    \sigma_1 \sim \sigma_1', &&
    r \sim T.
    \end{align*}
    \end{enumerate}

  \item {\bf Case of $C \equiv (x^\intt := Z)$ and $C' \equiv (x^\lintt := Z')$}:
    \begin{enumerate}
    \item Assume that 
    $C, D, \sigma_0 \Downarrows \sigma_1, r$ for some $\sigma_1$ and $r$.
    By the semantics of the source language,
    \begin{align*}
    \sigma_1 = \sigma_0[x^\intt \mapsto \db{Z}(\sigma_0)], &&
    r = 0.0.
    \end{align*}
    By \Cref{prop:exp-cor}, $\db{Z'}(\sigma_0',[]) = \db{Z}(\sigma_0)$.
    By the semantics of the target language,
    \begin{align*}
    C', D, \sigma_0', A \Downarrowi \sigma_1', T, &&
    \sigma_1' = \sigma_0'[x^\lintt : [[] \mapsto \db{Z'}(\sigma_0',[])]], &&
    T = T^z_A.
    \end{align*}
    \Cref{lem:upd-cor} entails $\sigma_1 \sim \sigma_1'$,
    Last, $r \sim T$ is obvious.

    \item Assume that 
    $C', D, \sigma_0', A \Downarrowi \sigma_1', T$ for some $\sigma_1'$ and $T$.
    By the semantics of the target language,
    \begin{align*}
    \sigma_1' = \sigma_0'[ x^\lintt : [[]\mapsto \db{Z'}(\sigma_0',[])]], &&
    T = T^z_A.
    \end{align*}
    By \Cref{prop:exp-cor}, $\db{Z}(\sigma_0) = \db{Z'}(\sigma_0',[])$.
    By the semantics of the source language,
    \begin{align*}
    C, D, \sigma_0 \Downarrows \sigma_1, r, &&
    \sigma_1 = \sigma[ x^\intt \mapsto \db{Z}(\sigma_0) ], &&
    r = 0.0.
    \end{align*}
    \Cref{lem:upd-cor} entails $\sigma_1 \sim \sigma_1'$.
    Last, $r \sim T$ is obvious.
    \end{enumerate}

  \item {\bf Case of $C \equiv (x^\real := E)$ and $C' \equiv (x^\lreal := E')$}: \\
    This case is similar to the previous one.
    
  \item {\bf Case of $C \equiv C' \equiv \code{skip}$}: \\
    In this case, by the semantics of the source language and that of the target language, we have
    \begin{align*}
    \code{skip}, D, \sigma_0 \Downarrows \sigma_0, 0.0, &&
    \code{skip}, D, \sigma_0', A \Downarrowi \sigma_0', T^z_A.
    \end{align*}
    Thus, the claimed properties of the proposition hold. 

  \item {\bf Case of $C \equiv (C_1;C_2)$ and $C' \equiv (C_1';C_2')$}:
    \begin{enumerate}
    \item Assume that $C_1; C_2, D, \sigma_0 \Downarrows \sigma_1, r$ for some $\sigma_1$ and $r$.
    By the semantics of the source language, there exist $r_1, r_2 \in \R$ and a state $\sigma_2$ such that
    \begin{align*}
    C_1, D, \sigma_0 \Downarrows \sigma_2, r_2, &&
    C_2, D, \sigma_2 \Downarrows \sigma_1, r_1, &&
    r = r_2 + r_1.
    \end{align*}
    By the induction hypothesis on $C_1$ and $C_1'$,
    there exist $\sigma_2'$ and $T_2$ such that
    \begin{align*}
    C_1', D, \sigma_0', A \Downarrowi \sigma_2', T_2, &&
    \sigma_2 \sim \sigma_2', &&
    r_2 \sim T_2.
    \end{align*}
    Similarly, by the induction hypothesis on $C_2$ and $C_2'$,
    there exist $\sigma_1'$ and $T_1$ such that 
    \begin{align*}
    C_2', D, \sigma_2', A \Downarrowi \sigma_1', T_1, &&
    \sigma_1 \sim \sigma_1', &&
    r_1 \sim T_1.
    \end{align*}
    Letting $T = T_2 \oplus T_1$ gives us
    \begin{align*}
    C_1'; C_2', D, \sigma_0', A \Downarrowi \sigma_1', T, &&
    \sigma_0 \sim \sigma_1', &&
    r = r_2 + r_1.
    \end{align*}
  
    \item Assume that $C_1'; C_2', D, \sigma_0', A \Downarrowi \sigma_1', T$ for some $\sigma_1'$ and $T$.
    By the semantics of the target language, there exist $T_1, T_2$ and a state $\sigma_2'$ such that
    \begin{align*}
    C_1', D, \sigma_0', A \Downarrowi \sigma_2', T_2, &&
    C_2', D, \sigma_2', A \Downarrowi \sigma_1', T_1, &&
    \sigma_0' \sim \sigma_2', &&
    T = T_2 \oplus T_1.
    \end{align*}
    By the induction hypothesis on $C_1$ and $C_1'$,
    there exist $\sigma_2$ and $r_2$ such that
    \begin{align*}
    C_1, D, \sigma_0 \Downarrows \sigma_2, r_2, &&
    \sigma_2 \sim \sigma_2', &&
    r_2 \sim T_2.
    \end{align*}
    Similarly, by the induction hypothesis on $C_2$ and $C_2'$,
    there exist $\sigma_1$ and $r_1$ such that
    \begin{align*}
    C_2, D, \sigma_2 \Downarrows \sigma_1, r_1, &&
    \sigma_1 \sim \sigma_1', &&
    r_1 \sim T_1.
    \end{align*}
    Letting $r = r_2 + r_1$ gives us
    \begin{align*}
    C_1; C_2, D, \sigma_0 \Downarrows \sigma_1, r, &&
    \sigma_1 \sim \sigma_1', &&
    r \sim T.
    \end{align*}
    \end{enumerate}
    
  \item {\bf Case of $C \equiv (\code{ifz}\ Z\ C_1\ C_2)$ and $C' \equiv (\code{ifz}\ Z'\ C_1'\ C_2')$}:
    \begin{enumerate}
    \item Assume that $(\code{ifz}\ Z\ C_1\ C_2), D, \sigma_0 \Downarrows \sigma_1, r$ for some $\sigma_1$ and $r$.
    There are two subcases depending on the value of $\db{Z}(\sigma_0)$:
    \begin{itemize}
    \item {\bf Subcase of $\db{Z}(\sigma_0) = 0$}: \\
    By the semantics of the source language, we have $C_1, D, \sigma_0 \Downarrows \sigma_1, r$.
    Since $\sigma_0 \sim \sigma_0'$, we have 
    $\db{Z'}(\sigma_0', []) = \db{Z}(\sigma_0) = 0$ by \Cref{prop:exp-cor}.
    Therefore, with the notations of the rule for $\code{ifz}$ in the target language semantics, 
    $A_1 = \{ [] \}$ and $A_2 = \emptyset$.
    By applying the induction hypothesis to $C_1$ and $C_1'$,
    we derive the existence of $\sigma_2'$ and $T$ such that
    \begin{align*}
    C_1', D, \sigma_0', A_1 \Downarrowi \sigma_2', T, &&
    \sigma_1 \sim \sigma_2', &&
    r \sim T,
    \end{align*}
    Thus, by applying \Cref{lem:emp-identity} to $C_2'$, there exists $\sigma_1'$ such that
    \begin{align*}
    C_2', D, \sigma_2', A_2 \Downarrowi \sigma_1', T_\emptyset, &&
    \sigma_2' =_{\Index} \sigma_1'.
    \end{align*}
    Therefore, we conclude $(\code{ifz}\ Z'\ C'_1\ C'_2), D, \sigma_0', A \Downarrowi \sigma_1', T$
    with $\sigma_1 \sim \sigma_1'$ and $r \sim T$.

    \item {\bf Subcase of $\db{Z}(\sigma_0) \neq 0$}: \\
    By the semantics of the source language, we have $C_2, D, \sigma_0 \Downarrows \sigma_1, r$.
    Since $\sigma_0 \sim \sigma_0'$, we have 
    $\db{Z'}(\sigma_0', []) = \db{Z}(\sigma_0) \neq 0$ by \Cref{prop:exp-cor}.
    Therefore, with the notations of the rule for $\code{ifz}$ in the
    semantics of the target language, $A_1 = \emptyset$ and $A_2 = \{ [] \}$.
    By applying \Cref{lem:emp-identity} to $C_1'$, there exists $\sigma_2'$ such that
    \begin{align*}
    C_1', D, \sigma_0', A_1 \Downarrowi \sigma_2', T_\emptyset, &&
    \sigma_0' =_{\Index} \sigma_2'.
    \end{align*}
    This implies $\sigma_0 \sim \sigma_2'$, so we can apply the induction hypothesis to $C_2$ and $C_2'$
    to derive the existence of $\sigma_1'$ and $T$ such that
    \begin{align*}
    C'_2, D, \sigma_2', A_1 \Downarrowi \sigma_1', T, &&
    \sigma_1 \sim \sigma_1', &&
    r \sim T.
    \end{align*}
    Therefore, we conclude $(\code{ifz}\ Z'\ C'_1\ C'_2), D, \sigma_0', A \Downarrowi \sigma_1', T$
    with $\sigma_1 \sim \sigma_1'$ and $r \sim T$.
    \end{itemize}

    \item Now assume that 
    $(\code{ifz}\ Z'\ C'_1\ C'_2), D, \sigma_0', A \Downarrowi \sigma_1', T$ for some $\sigma_1'$ and $T$.
    As before, we distinguish two subcases, depending on the value of $\db{Z'}(\sigma_0', [])$:
    \begin{itemize}
    \item {\bf Subcase of $\db{Z'}(\sigma_0', []) = 0$}: \\
    By the semantics of the target language, there exist $\sigma_2'$ and $T_2, T_1$ such that 
    \begin{align*}
    C_1', D, \sigma_0', \{ [] \} \Downarrowi \sigma_2', T_2, &&
    C_2', D, \sigma_2', \emptyset \Downarrowi \sigma_1', T_1. &&
    T = T_2 \oplus T_1,
    \end{align*}
    Since $\sigma_0 \sim \sigma_0'$, we have $\db{Z}(\sigma_0', []) = \db{Z'}(\sigma_0) = 0$ by \Cref{prop:exp-cor}. 
    By the induction hypothesis on $C_1$ and $C_1'$, there exist $\sigma_1$ and $r$ such that
    \begin{align*}
    C_1, D, \sigma_0 \Downarrows \sigma_1, r, &&
    \sigma_1 \sim \sigma_2', &&
    r \sim T_2.
    \end{align*}
    By the semantics of the source language, this in turn implies that 
    $(\code{ifz}\ Z\ C_1\ C_2), D, \sigma_0 \Downarrows \sigma_1, r$.
    \Cref{lem:emp-identity} can be applied to $C_2'$, and doing so gives the following facts:
    \begin{align*}
    \sigma_2' =_{\Index} \sigma_1', &&
    T_1 = T_\emptyset.
    \end{align*}
    So, $T_2 = T$ and $r \sim T$. Also, $\sigma_1 \sim \sigma_1'$.

    \item {\bf Subcase of $\db{Z'}(\sigma_0', []) \neq 0$}: \\
    By the semantics of the target language, there exist $\sigma_2'$ and $T_2, T_1$ such that
    \begin{align*}
    C_1', D, \sigma_0', \emptyset \Downarrowi \sigma_2', T_2, &&
    C_2', D, \sigma_2', \{ [] \} \Downarrowi \sigma_1', T_1, &&
    T = T_2 \oplus T_1.
    \end{align*}
    Since $\sigma_0 \sim \sigma_0'$, we have $\db{Z}(\sigma_0) = \db{Z'}(\sigma_0', []) \not= 0$ by \Cref{prop:exp-cor}. 
    Thus, we can apply \Cref{lem:emp-identity} to $C_1'$, and derive that
    \begin{align*}
    \sigma_0' =_{\Index} \sigma_2', &&
    \sigma_0 \sim \sigma_2', &&
    T_2 = T_\emptyset, &&
    T = T_1.
    \end{align*}
    By the induction hypothesis applied to $C_2$ and $C_2'$, there exist $\sigma_0'$ and $r$ such that 
    \begin{align*}
    C_2, D, \sigma_0 \Downarrows \sigma_1, r, &&
    \sigma_1 \sim \sigma_1', &&
    r \sim T.
    \end{align*}
    By the semantics of the source language,
    $(\code{ifz}\ Z\ C_1\ C_2), D, \sigma_0 \Downarrows \sigma_1, r$.
    \end{itemize}
    \end{enumerate}
  
  \item {\bf Case of
    $C \equiv (\code{for}\ x^\intt \ \code{in} \ \code{range}(n)\ \code{do}\ C_1)$
    and 
    $C' \equiv (\code{for}\ x^\lintt \ \code{in} \ \code{range}(n)\ \code{do}\ C_1')$}: \\
    We observe that $C$ is semantically equivalent to \( x^\intt := 0; C_1; x^\intt := 1; C_1;
    \ldots x^\intt := n-1; C_1 \) and that $C'$ is semantically equivalent to \( x^\lintt := 0;
    C'_1; x^\lintt := 1; C'_1; \ldots x^\lintt := n-1; C'_1 \).
    Therefore the proof follows from the induction hypothesis on $C_1$ and $C'_1$, and on the
    above arguments for the integer-assignment and sequential-composition cases, it follows
    that these loop cases satisfy the claimed properties of the proposition.
  \end{itemize}
\end{proof}

\subsection{Soundness of Transformation of For-loops in the Target Language}
\label{subsec:parallelising-for-loops}
In this subsection, we prove the soundness of the optimisation of for-loops within the target language,
but with the evaluation rules defined by $\Downarrowi$.
Our optimisation replaces each for-loop by a $\code{loop\_fixpt\_noacc}$ loop 
such that the body of the $\code{loop\_fixpt\_noacc}$ loop works on a larger set of indices than
the original loop and also runs without accumulating the score maps for all but the last iteration.
We will prove that this optimisation preserves the original loop's semantics.

Before delving into the soundness theorem of our optimisation, we first introduce the concept of $L$-states.
Let $L$ be a set of indices.
A state $\sigma$ is called {\bf $L$-state} if for every variable $x^{\tau\dagger}$,
\[
\dom(\sigma(x^{\tau\dagger})) \subseteq L{\downarrow}.
\]
For all $L$-states $\sigma$, variables $x^{\tau\dagger}$, and indices i, we have
\[
\extend(\sigma(x^{\tau\dagger}))(i) = 
\extend(\sigma(x^{\tau\dagger}))(\max(L{\downarrow} \cap i{\downarrow})).
\]
That is, the values of $\extend(\sigma(x^{\tau\dagger}))$ at indices 
in $L{\downarrow}$ fully determines the entire function $\extend(\sigma(x^{\tau\dagger}))$.
The next lemma gives that the update of a variable preserves the property of being an $L$-state.
\begin{lemma}
  \label{lem:update-L-state}
  Let $L$ be a set of indices. Then,
  for all states $\sigma \in \State_\mathit{tgt}$, variables $x^{\tau\dagger}$,
  and tensors $T \in \Tensor_{\db{\tau}}$, if $\sigma$ is an $L$-state and $\dom(T) \subseteq L{\downarrow}$, then
  $\sigma[x^{\tau\dagger} : T]$ is also an $L$-state.
\end{lemma}
\begin{proof}
  For every variable $y^{\tau'\dagger}$, we have
  \[
  \dom(\sigma[x^{\tau\dagger} : T](y^{\tau'\dagger}))
  \subseteq 
  \dom(\sigma(y^{\tau'\dagger})) \cup \dom(T)
  \subseteq 
  L{\downarrow}.
  \]
  The second inclusion holds since $\sigma$ is an $L$-state
  and $\dom(T) \subseteq L{\downarrow}$. Thus, $\sigma[x^{\tau\dagger} : T]$ is an $L$-state.
\end{proof}

The following lemma says that the execution under a finite antichain $A \in \FAntiChain$ preserves
the state being an $L$-state for every $L$ with $A \subseteq L{\downarrow}$.
\begin{lemma}[$L$-state Preservation]
  \label{lem:preservation-L-states}
  Let $C$ be a command, $D$ be an RDB, $A$ be a finite antichain of indices, and
  $L$ be a set of indices. For all states $\sigma,\sigma' \in \State_\mathit{tgt}$ and tensors $T \in \Tensor_\R$,
  if 
  \[ 
  A \subseteq L{\downarrow},
  \quad 
  \text{$\sigma$ is an $L$-state},
  \quad\text{and}\quad 
  C, D, \sigma, A \Downarrowi \sigma', T,
  \]  
  then $\sigma'$ is also an $L$-state. 
\end{lemma}
\begin{proof}
  The proof is by induction on the structure of $C$.
  \begin{itemize}[itemsep=5pt]
  \item {\bf Cases of
    $C \equiv (x^\lreal := \code{fetch}(I))$,
    $C \equiv (\code{score}(E))$,
    $C \equiv (x^\lintt := Z)$,
    $C \equiv (x^\lreal := E)$,
    $C \equiv (\code{skip})$, and
    $C \equiv (x^{\lintt} := \code{lookup\_index}(\alpha))$:} \\
    In these cases, $\sigma'$ is $\sigma$ or $\sigma[x^{\tau\dagger} : T]$ 
    for some $T$ with $\dom(T) = A \subseteq L{\downarrow}$. Thus, $\sigma'$ is also an $L$-state
    by \Cref{lem:update-L-state}.

  \item {\bf Cases of $C \equiv (C_1;C_2)$:} \\
    The desired conclusion follows from the induction hypothesis on $C_1$ and $C_2$.

  \item {\bf Case of $C \equiv (\code{ifz}\ Z\ C_1\ C_2$):} \\
    Let $A_1 = \{i \in A : \db{Z}(\sigma,i) = 0\}$ and $A_2 = \{i \in A : \db{Z}(\sigma,i) \neq 0\}$.
    Since $A_1 \subseteq L{\downarrow}$
    and $A_2 \subseteq L{\downarrow}$, the desired conclusion follows from the induction hypothesis
    on $C_1$ and $C_2$.

  \item {\bf Cases of $C \equiv (\code{for}\ x^\lintt \ \code{in} \ \code{range}(n)\ \code{do}\ C')$:} \\
    We remark that $C$ is semantically equivalent to \( x^\lintt := 0; C'; x^\lintt := 1; C';
    \ldots x^\lintt := n-1; C' \).
    Therefore, the result follows from the induction hypothesis for $C'$ and the above
    arguments for integer-assignments and sequential-composition.

  \item {\bf Case of $C \equiv (\code{extend\_index}(\alpha,n)\ C')$:} \\
    By the rule of the $\code{extend\_index}$ command, we have
    \begin{align*}
    C', D, \sigma, A_0 \Downarrowi \sigma_0, T_0, 
    &\qquad
    A_0 = \{i\concat [(\alpha,k)] : i \in A, k \in \{0,\ldots,n-1\}\}, \\
    \sigma' = \sigma_0\langle\rho_0\rangle,
    &\qquad 
    \rho_0 = [i\concat [(\alpha,n-1)] \mapsto i : i \in A].
    \end{align*}
    Then, $\sigma$ is an $(L\cup A_0)$-state, and by the induction hypothesis on $C'$,
    $\sigma_0$ is also an $(L\cup A_0)$-state. 
    For an arbitrary variable $x^{\tau\dagger}$,
    define a tensor $T_{x, \rho_0} \in \Tensor_{\db{\tau}}$ by
    \[
      T_{x, \rho_0}(i) = 
      \begin{cases}
        \sigma_0(x^{\tau\dagger})(\rho_0^{-1}(i)), & \text{if $i \in \image(\rho_0) \cap \rho_0(\dom(\sigma_0(x^{\tau\dagger})))$},
        \\
        \text{undefined}, & \text{otherwise}.
      \end{cases}
    \]
    Then, 
    \begin{align*}
      \dom(\sigma'(x^{\tau\dagger}))
      & {} = \dom(\sigma_0\langle\rho_0\rangle(x^{\tau\dagger}))
      \\
      & {} = \dom(\sigma_0[x^{\tau\dagger}: T_{x,\rho_0}](x^{\tau\dagger}))
      \\
      & {} = (\dom(\sigma_0(x^{\tau\dagger})) \setminus \dom(T_{x,\rho_0}){\uparrow}) \cup \dom(T_{x,\rho_0})
      \\
      & {} \subseteq ((L \cup A_0){\downarrow} \setminus \dom(T_{x,\rho_0}){\uparrow}) \cup \dom(T_{x,\rho_0})
      \\
      & {} = ((L \cup A_0){\downarrow} \setminus A{\uparrow}) \cup A
      \\
      & {} \subseteq (L{\downarrow} \cup (A_0{\downarrow} \setminus A{\uparrow})) \cup A
      \\
      & {} \subseteq (L{\downarrow} \cup A{\downarrow}) \cup A = L{\downarrow}.
    \end{align*}
    The first subset relationship holds because $\sigma_0$ is an $(L\cup A_0)$-state,
    the fifth equality is true since $\dom(T_{x^{\tau\dagger}}) = \image(\rho_0) = A$,
    and the third subset relationship holds because
    $A_0{\downarrow} = A{\downarrow} \cup A_0$ and $A_0 \subseteq A{\uparrow}$. 

  \item {\bf Case of $C \equiv \code{shift}(\alpha)$:} \\
    In this case, $\sigma' = \sigma\langle\rho\rangle$ for some
    finite partial function $\rho$ such that $\image(\rho) \subseteq A$. 
    For an arbitrary variable $x^{\tau\dagger}$, define a tensor 
    $T_{x,\rho} \in \Tensor_{\db{\tau}}$ by
    \[
    T_{x,\rho}(i) = 
    \begin{cases}
      \sigma(x^{\tau\dagger})(\rho^{-1}(i)), & \text{if $i \in \image(\rho) \cap \rho(\dom(\sigma(x^{\tau\dagger})))$},
      \\
      \text{undefined}, & \text{otherwise}.
    \end{cases}
    \]
    Then,
    \begin{align*}
      \dom(\sigma'(x^{\tau\dagger}))
      & {} = \dom(\sigma\langle\rho\rangle(x^{\tau\dagger}))
      \\
      & {} = \dom(\sigma[x^{\tau\dagger}: T_{x,\rho}](x^{\tau\dagger}))
      \\
      & {} = (\dom(\sigma(x^{\tau\dagger})) \setminus \dom(T_{x,\rho}){\uparrow}) \cup \dom(T_{x,\rho})
      \\
      & {} \subseteq L{\downarrow} \cup \image(\rho)
      \\
      & {} \subseteq L{\downarrow} \cup A = L{\downarrow}.
    \end{align*}
    We have just shown that $\sigma'$ is an $L$-state, as required.

  \item {\bf Case of $C \equiv (\code{loop\_fixpt\_noacc}(n)\ C')$:} \\
    By the semantics of the $\code{loop\_fixpt\_noacc}$ command,
    there exist $\sigma_k \in \State_\mathit{tgt}$ and $T_k \in \Tensor_\R$ for $k = 0,\ldots,n-1$ such that
    \begin{align*}
    \sigma_0 = \sigma, &&
    \sigma_n = \sigma', &&
    T_n = T, &&
    C', D, \sigma_k, A \Downarrowi \sigma_{k+1}, T_{k+1}
    \quad\text{for $k = 0,\ldots,n-1$}.
    \end{align*}
    By the induction hypothesis on $C'$, $\sigma_{k+1}$ is an $L$-state for all $k = 0,\ldots,n-1$,
    and thus $\sigma'$ is also an $L$-state.
  \end{itemize}
\end{proof}

The following proposition states the preservation of lifted integer, real, and index expressions
under a coupling $R$ and the induced relations.
\begin{proposition}
  \label{prop:expression-R-preservation}
  Let $R$ be a coupling, and $Z$, $E$, $I$ be lifted integer, real, and index expressions 
  in the target language. Then, for all states $\sigma_0,\sigma_1 \in \State_\mathit{tgt}$ 
  and indices $i_0,i_1 \in \Index$, if 
  \[
  \sigma_0 \, [\State_\mathit{tgt}(R)] \, \sigma_1
  \quad\text{and}\quad
  i_0 \, [R] \, i_1,
  \]
  then 
  \begin{align*}
  \db{Z}(\sigma_0,i_0) & = \db{Z}(\sigma_1,i_1),
  & 
  \db{E}(\sigma_0,i_0) & = \db{E}(\sigma_1,i_1),
  &
  \db{I}(\sigma_0,i_0) & = \db{I}(\sigma_1,i_1).
  \end{align*}
\end{proposition}
\begin{proof}
The proof is by structural induction.
\begin{itemize}[itemsep=5pt]
  \item {\bf Cases of $Z \equiv n$ and $E \equiv r$:} \\
    The desired equalities in both cases follow 
    immediately from the semantics of the target language.

  \item {\bf Cases of $Z \equiv x^\lintt$ and $E \equiv x^\lreal$:} \\
    The desired conclusions in these cases follow from the assumption that
    $\sigma_0 \, [\State_\mathit{tgt}(R)] \, \sigma_1$ and $i_0 \, [R] \, i_1$.

  \item {\bf Case of $Z \equiv \op_i(Z_1,\ldots,Z_n,E_1,\ldots,E_m)$:} \\
    By the induction hypothesis, $\db{Z_k}(\sigma_0,i_0) = \db{Z_k}(\sigma_1,i_1)$
    and $\db{E_\ell}(\sigma_0,i_0) = \db{E_\ell}(\sigma_1,i_1)$ for all $k \in \{1,\ldots,n\}$
    and $\ell \in \{1,\ldots,m\}$. Thus, we can prove
    the desired equality as follows:
    \begin{align*}
    \db{\op_i(Z_1,\ldots,Z_n,E_1,\ldots,E_m)}(\sigma_0,i_0)
    & {} = 
    \db{\op_i}_a(
      \begin{aligned}[t]
        & \db{Z_1}(\sigma_0,i_0),\ldots,\db{Z_n}(\sigma_0,i_0),
        \\
        & \db{E_1}(\sigma_0,i_0),\ldots,\db{E_m}(\sigma_0,i_0))
      \end{aligned}
    \\
    & {} = 
    \db{\op_i}_a(
      \begin{aligned}[t]
        & \db{Z_1}(\sigma_1,i_1),\ldots,\db{Z_n}(\sigma_1,i_1),
        \\
        & \db{E_1}(\sigma_1,i_1),\ldots,\db{E_m}(\sigma_1,i_1))
      \end{aligned}
    \\
    & {} =
    \db{\op_i(Z_1,\ldots,Z_n,E_1,\ldots,E_m)}(\sigma_1,i_1).
    \end{align*}

  \item {\bf Case of $E \equiv \op_r(Z_1,\ldots,Z_n,E_1,\ldots,E_m)$:} \\
    The case can be proved by a similar argument to the one 
    for $Z \equiv \op_i(Z_1,\ldots,Z_n,E_1,\ldots,E_m)$.

  \item {\bf Case of $I \equiv [(\alpha_1, Z_1);\, \ldots;\, (\alpha_n, Z_n)]$:} \\
    By the induction hypothesis,  
    $\db{Z_k}(\sigma_0,i_0) = \db{Z_k}(\sigma_1,i_1)$ 
    for all $k \in \{1,\ldots,n\}$. Thus,
    \begin{align*}
    \db{[(\alpha_1, Z_1);\, \ldots;\, (\alpha_n, Z_n)]}(\sigma_0,i_0)
    & {} = 
    [(\alpha_1,\db{Z_1}(\sigma_0,i_0));\ldots;(\alpha_n,\db{Z_n}(\sigma_0,i_0))]
    \\
    & {} = 
    [(\alpha_1,\db{Z_1}(\sigma_1,i_1));\ldots;(\alpha_n,\db{Z_n}(\sigma_1,i_1))]
    \\
    & {} = 
    \db{[(\alpha_1, Z_1);\, \ldots;\, (\alpha_n, Z_n)]}(\sigma_1,i_1),
    \end{align*}
    as desired.
\end{itemize} 
\end{proof}

Let $C$ be a command in the target language that does not use any construct absent in the source language.
We define the translation of $C$ to a command $\overline{C}$ in the target language as follows:
\begin{align*}
  & \begin{aligned}[t]
  \overline{x^\lreal := \code{fetch}(I)} & {}\;\equiv\; x^{\lreal} := \code{fetch}(I), \\
  \overline{\code{score}(E)} & {}\;\equiv\; \code{score}(E), \\
  \overline{x^\lintt := Z} & {}\;\equiv\; x^{\lintt} := Z, \\
  \overline{x^\lreal := E} & {}\;\equiv\; x^{\lreal} := E, \\
  \overline{\code{skip}} & {}\;\equiv\; \code{skip}, \\
  \overline{\code{ifz}\ Z\ C_1\ C_2} & {}\;\equiv\; \code{ifz}\ Z\ \overline{C_1}\ \overline{C_2}, \\
  \overline{C_1; C_2} & {}\;\equiv\; \overline{C_1}; \overline{C_2},
  \end{aligned} &
  & \begin{aligned}[t]
    & \overline{\code{for}\ x^\lintt \ \code{in} \ \code{range}(n)\ \code{do}\ C_1} \\
      & \quad \begin{aligned}[t]
      & {}\;\equiv\; \code{extend\_index}(\alpha,n)\ \{ \\
      & {}\phantom{\;\equiv\;} \quad \code{loop\_fixpt\_noacc}(n)\ \{ \\
      & {}\phantom{\;\equiv\;} \quad \quad \code{shift}(\alpha); \\
      & {}\phantom{\;\equiv\;} \quad \quad x^{\intt\dagger} := \code{lookup\_index}(\alpha); \\
      & {}\phantom{\;\equiv\;} \quad \quad \overline{C_1}\;\; \} \}
    \end{aligned} \\
    & \text{where $\alpha \in \Str \setminus \mathcal{S}(\overline{C_1})$}.
  \end{aligned}
\end{align*}
Here $\mathcal{S}(C)$ denotes the set of all strings that appear in $C$.

The following theorem states the soundness of the translation of commands in the target language
when the commands are evaluated under the evaluation rules defined by $\Downarrowi$.
\begin{theorem}
  \label{thm:soundness-translation}
  Let $C$ be a command in the target language that does not use any construct absent in the source language.
  Let $D \in \RDB$, $A_0, A_1 \in \FAntiChain$, $\sigma_0, \sigma_1, \sigma_0', \sigma_1' \in \State_\mathit{tgt}$, $T_0, T_1 \in \Tensor_\R$,
  $\sigma_0$ be an $A_0$-state, and $\sigma_1$ be an $A_1$-state. Suppose that
  \begin{align*}
    C, D, \sigma_0, A_0 \Downarrowi \sigma_0', T_0
    \quad\text{and}\quad
    \overline{C}, D, \sigma_1, A_1 \Downarrowi \sigma_1', T_1
  \end{align*}
  Then, for a coupling $R$ such that
  \begin{align*}
    \sigma_0 \, [\State_\mathit{tgt}(R)] \, \sigma_1
    \quad\text{and}\quad
    A_0 \, [\cP(R)] \, A_1,
  \end{align*}
  we have that
  \begin{align*}
    \sigma_0' \, [\State_\mathit{tgt}(R)] \, \sigma_1'
    \quad\text{and}\quad
    T_0 \, [\Tensor_\R(R)] \, T_1.
  \end{align*}
\end{theorem}
\begin{proof}
  We do the case analysis on $C$.
  \begin{itemize}
    \item {\bf Case of $C \equiv (x^\lreal := \code{fetch}(I))$:} \\
      By the semantics of commands, $T_0 = T^z_{A_0}$ and $T_1 = T^z_{A_1}$.
      Since $A_0 \, [\cP(R)] \, A_1$, by \Cref{cor:Tz-R-from-indexset-R}, we have
      \[
        T_0 = T^z_{A_0} \, [\Tensor_\R(R)] \, T^z_{A_1} = T_1.
      \] 
      It remains to show that 
      \[
        \sigma'_0 \, [\State_\mathit{tgt}(R)] \, \sigma'_1.
      \]
      By the semantics of the target language,
      \begin{align*}
      \sigma'_0 &= \sigma_0[x^\lreal : T_0'], & T_0' &= [i_0 \mapsto D(\db{I}(\sigma_0,i_0)) : i_0 \in A_0], \\
      \sigma'_1 &= \sigma_1[x^\lreal : T_1'], & T_1' &= [i_1 \mapsto D(\db{I}(\sigma_1,i_1)) : i_1 \in A_1].
      \end{align*}
      We claim that $T_0' \, [\WTensor_\R(R)] \, T_1'$.
      Note that if true, the claim gives the desired conclusion by \Cref{lem:state-R-tensor-update}.
      Also, note that by \Cref{cor:extend-R-tensor},
      the claim follows if we show the following two conditions:
      \begin{enumerate}[label=(\roman*), leftmargin=2\parindent]
        \item $\dom(T_0') \, [\cP(R)] \, \dom(T_1')$,
        \item $T_0' \, [\Tensor_\R(R)] \, T_1'$.
      \end{enumerate}
      The first condition follows immediately from the facts that $\dom(T_0') = A_0$, $\dom(T_1') = A_1$.
      To prove the second condition, consider arbitrary indices $i'_0 \in \dom(T_0')$ and $i'_1 \in \dom(T_1')$ 
      with $i'_0 \, [R] \, i'_1$. Then, $\db{I}(\sigma_0,i'_0) = \db{I}(\sigma_1,i'_1)$ 
      by \Cref{prop:expression-R-preservation}. Thus, 
      \begin{align*}
        T'_0(i'_0) = D(\db{I}(\sigma_0,i'_0)) = D(\db{I}(\sigma_1,i'_1)) = T'_1(i'_1),
      \end{align*}
      as desired.
      
    \item {\bf Case of $C \equiv (\code{score}(E))$:} \\
      By the semantics of the target language, 
      \begin{align*}
      \sigma_0' &= \sigma_0 & T_0 &= [i_0 \mapsto \db{E}(\sigma_0,i_0) : i_0 \in A_0], \\
      \sigma_1' &= \sigma_1 & T_1 &= [i_1 \mapsto \db{E}(\sigma_1,i_1) : i_1 \in A_1].
      \end{align*}
      Thus, $\sigma'_0 \, [\State_\mathit{tgt}(R)] \, \sigma'_1$.
      We claim that $T_0 \, [\Tensor_\R(R)] \, T_1$.
      Suppose $i_0 \, [R] \, i_1$.
      By the assumption $A_0 \, [\cP(R)] \, A_1$, we have $\dom(T_0) \, [\cP(R)] \, \dom(T_1)$.
      By \Cref{lem:indexset-R-membership}, it implies
      \[
      i_0 \in \dom(T_0) \iff i_1 \in \dom(T_1).
      \]
      Also, by the assumption $\sigma_0 \, [\State_\mathit{tgt}(R)] \, \sigma_1$ and \Cref{prop:expression-R-preservation},
      \[
      T_0(i_0) = \db{E}(\sigma_0,i_0) = \db{E}(\sigma_1,i_1) = T_1(i_1).
      \]
      The desired relationship $T_0 \, [\Tensor_\R(R)] \, T_1$ follows from these two facts.

    \item {\bf Case of $C \equiv (x^\lintt := Z)$:} \\
      By the semantics of the target language,
      \begin{align*}
      \sigma_0' & = \sigma_0[x^\lintt : T_0'], & T'_0 = [i_0 \mapsto \db{Z}(\sigma_0, i_0) : i_0 \in A_0], \\
      \sigma_1' & = \sigma_1[x^\lintt : T_1'], & T'_1 = [i_1 \mapsto \db{Z}(\sigma_1, i_1) : i_1 \in A_1],
      \end{align*}
      and $T_0 = T^z_{A_0}$ and $T_1 = T^z_{A_1}$.
      Since $A_0 \, [\cP(R)] \, A_1$, by \Cref{cor:Tz-R-from-indexset-R}, we have
      \[
      T_0 = T^z_{A_0} \, [\Tensor_\R(R)] \, T^z_{A_1} = T_1.
      \]
      It remains to prove that 
      $\sigma'_0 \, [\State_\mathit{tgt}(R)] \, \sigma'_1$. Note that by \Cref{lem:state-R-tensor-update},
      we can discharge this proof obligation by showing that $T_0' [\WTensor_\R(R)] T_1'$.
      But by \Cref{cor:extend-R-tensor}, this relationship between $T_0'$ and $T_1'$ follows 
      from the following two conditions:
      \begin{enumerate}[label=(\roman*), leftmargin=2\parindent]
        \item $\dom(T_0') \, [\cP(R)] \, \dom(T_1')$,
        \item $T_0' \, [\Tensor_\R(R)] \, T_1'$.
      \end{enumerate}
      The first condition holds because $\dom(T_0') = A_0$, $\dom(T_1') = A_1$, and $A_0 \, [\cP(R)] \, A_1$.
      To prove the second condition, pick arbitrary $i'_0 \in \dom(T_0')$ 
      and $i'_1 \in \dom(T_1')$ with $i'_0 \, [R] \, i'_1$. Then, since
      $i'_0 \, [R] \, i'_1$ and $\sigma_0 \, [\State_\mathit{tgt}(R)] \, \sigma_1$, we have
      $\db{Z}(\sigma_0,i'_0) = \db{Z}(\sigma_1,i'_1)$ by \Cref{prop:expression-R-preservation}. Thus,
      \[
        T'_0(i'_0) 
        = 
        \db{Z}(\sigma_0,i'_0)
        =
        \db{Z}(\sigma_1,i'_1)
        = 
        T'_1(i'_1),
      \]
      as desired. 

    \item {\bf Case of $C \equiv (x^{\lreal} := E)$:} \\
     The proof is similar to the one for $C \equiv (x^\lintt := Z)$.

    \item {\bf Case of $C \equiv (\code{skip})$:} \\
      In this case, we have $\sigma_0' = \sigma_0$, $\sigma'_1 = \sigma_1$, 
      $T_0 = T^z_{A_0}$, and $T_1 = T^z_{A_1}$. Since $A_0 \, [\cP(R)] \, A_1$
      and $\sigma_0 \, [\State_\mathit{tgt}(R)] \, \sigma_1$, by \Cref{cor:Tz-R-from-indexset-R}, we have
      \[
      \sigma_0' \, [\State_\mathit{tgt}(R)] \, \sigma_1'
      \quad\text{and}\quad
      T_0 \, [\Tensor_\R(R)] \, T_1.
      \]

    \item {\bf Case of $C \equiv (C_1;C_2)$:} \\
      By the semantics of sequencing, 
      there exist $\sigma_{0}'',\sigma_{1}'' \in \State_\mathit{tgt}$ and $T'_{0},T''_{0},T'_1,T''_1 \in \Tensor_\R$ 
      such that
      \begin{align*}
      C_1, D, \sigma_0, A_0 & {} \Downarrowi \sigma''_{0}, T''_{0}, &
      C_2, D, \sigma''_{0}, A_0 & {} \Downarrowi \sigma'_{0}, T'_{0}, &
      T_0 = T''_{0} \oplus T'_{0}, \\
      \overline{C_1}, D, \sigma_1, A_1 & {} \Downarrowi \sigma''_{1}, T''_{1}, &
      \overline{C_2}, D, \sigma''_{1}, A_1 & {} \Downarrowi \sigma'_{1}, T'_{1}, &
      T_1 = T''_{1} \oplus T'_{1}.
      \end{align*}
      By the induction hypothesis on $C_1$, 
      we have 
      \[
      \sigma''_{0} \, [\State_\mathit{tgt}(R)] \, \sigma''_{1}
      \quad
      \text{and}
      \quad
      T''_{0} \, [\Tensor_\R(R)] \, T''_{1}.
      \]
      The first relationship from above enables us to apply the induction hypothesis to $C_2$,
      and to obtain
      \[
      \sigma'_{0} \, [\State_\mathit{tgt}(R)] \, \sigma'_{1}
      \quad
      \text{and}
      \quad
      T'_{0} \, [\Tensor_\R(R)] \, T'_{1}.
      \]
      Thus, by \Cref{lem:tensor-R-oplus},
      \[
        \sigma'_{0} \, [\State_\mathit{tgt}(R)] \, \sigma'_{1}
        \quad
        \text{and}
        \quad
        T_0 = \Big(T''_{0} \oplus T'_{0}\Big) \, [\Tensor_\R(R)] \, \Big(T''_{1} \oplus T'_{1}\Big) = T_1.
      \]

  \item {\bf Case of $C \equiv (\code{ifz}\ Z\ C_1\ C_2)$:}\\
  Let
  \begin{align*}
    A'_{0} & = \{i_0 \in A_0 : \db{Z}(\sigma_0, i_0) = 0\}, &
    A''_{0} & = \{i_0 \in A_0 : \db{Z}(\sigma_0, i_0) \neq 0\}, \\
    A'_{1} & = \{i_1 \in A_1 : \db{Z}(\sigma_1, i_1) = 0\}, &
    A''_{1} & = \{i_1 \in A_1 : \db{Z}(\sigma_1, i_1) \neq 0\}.
  \end{align*}
  By the semantics, there exist $\sigma''_0,\sigma''_1 \in \State_\mathit{tgt}$
  and $T'_{0},T''_{0},T'_{1},T''_{1} \in \Tensor_\R$ such that
  \begin{align*}
    & C_1, D, \sigma_0, A'_{0} \Downarrowi \sigma''_0, T''_{0}, &
    & C_2, D, \sigma''_0, A''_{0} \Downarrowi \sigma'_0, T'_{0}, &
    & T_0 = T''_{0} \oplus T'_{0}, \\
    & \overline{C_1}, D, \sigma_1, A'_{1} \Downarrowi \sigma''_1, T''_{1}, &
    & \overline{C_2}, D, \sigma''_1, A''_{1} \Downarrowi \sigma'_1, T'_{1}, &
    & T_1 = T''_{1} \oplus T'_{1}.
  \end{align*}  
  By \Cref{lem:updated-indices-cmd}, 
  \begin{equation*}
    \dom(T''_0) = A'_0, \quad
    \dom(T'_0) = A''_0, \quad
    \dom(T''_1) = A'_1, \quad
    \dom(T'_1) = A''_1.
  \end{equation*}
  We claim that 
  \begin{equation*}
    A'_0 \, [\cP(R)] \, A'_1
    \quad\text{and}\quad
    A''_0 \, [\cP(R)] \, A''_1.
  \end{equation*}
  We only prove the first relationship in the claim; the second can be proved similarly.
  Consider arbitrary indices $i_0,i_1 \in \Index$ with $i_0 \, [R] \, i_1$.
  We need to prove that
  \begin{enumerate}[label=(\roman*), leftmargin=2\parindent]
    \item $i_0 \in A_0'{\uparrow} \iff i_1 \in A_1'{\uparrow}$,
    \item $i_0 \in A_0'{\uparrow} \implies \max(A'_0 \cap i_0{\downarrow}) \, [R] \, \max(A'_1 \cap i_1{\downarrow})$.
  \end{enumerate}
  First, we show the only-if direction of the first condition; the if direction can be proved similarly.
  Assume that $i_0 \in A_0'{\uparrow}$. Then, $i_0 \in A_0{\uparrow}$ and $i_1 \in A_1{\uparrow}$
  since $A'_0 \subseteq A_0$ and $A_0 \, [\cP(R)] \, A_1$.
  Thus, since $A_1{\uparrow} = A_1'{\uparrow} \cup A_1''{\uparrow}$, it suffices to show that $i_1 \not\in A''_1{\uparrow}$.
  For the sake of contradiction, suppose that $i_1 \in A''_1{\uparrow}$. Then, again since $A_0 \, [\cP(R)] \, A_1$, 
  \[
  \max(A'_0 \cap i_0{\downarrow}) = \max(A_0 \cap i_0{\downarrow})
  \, [R] \, 
  \max(A_1 \cap i_1{\downarrow}) = \max(A''_1 \cap i_1{\downarrow}).
  \]
  Since $\sigma_0 \, [\State_\mathit{tgt}(R)] \, \sigma_1$, by \Cref{prop:expression-R-preservation}, we have
  \[
  \db{Z}(\sigma_0, \max(A'_0 \cap i_0{\downarrow})) 
  = 
  \db{Z}(\sigma_1, \max(A''_1 \cap i_1{\downarrow})).
  \]
  But this cannot happen because the left hand side of the equation is $0$ and the right hand side is not $0$.
  Next, we prove the second condition.
  Assume that $i_0 \in A_0'{\uparrow}$, and so $i_1 \in A_1'{\uparrow}$.
  Then, since $A_0 \, [\cP(R)] \, A_1$ and $A'_0 \subseteq A_0$ and $A'_1 \subseteq A_1$ and both $A_0$ and $A_1$ are antichains, we have
  \[
  \max(A'_0 \cap i_0{\downarrow}) = \max(A_0 \cap i_0{\downarrow})
  \, [R] \, 
  \max(A_1 \cap i_1{\downarrow}) = \max(A'_1 \cap i_1{\downarrow}),
  \]
  and thus, the claim holds.
  By $A_0' \, [\cP(R)] \, A_1'$ and the assumption $\sigma_0 \, [\State_\mathit{tgt}(R)] \, \sigma_1$,
  we can apply the induction hypothesis to $C_1$. This gives us 
  \begin{equation*}
  \sigma''_0 \, [\State_\mathit{tgt}(R)] \, \sigma''_1
  \quad\text{and}\quad
  T''_0 \, [\Tensor_\R(R)] \, T''_1.
  \end{equation*}
  The first conjunct from above and $A_0'' \, [\cP(R)] \, A_1''$ let us apply the induction hypothesis to $C_2$, which gives us 
  \begin{equation*}
  \sigma'_0 \, [\State_\mathit{tgt}(R)] \, \sigma'_1
  \quad\text{and}\quad
  T'_0 \, [\Tensor_\R(R)] \, T'_1.
  \end{equation*}
  By \Cref{lem:tensor-R-oplus}, we have
  \[
  T_0 = (T''_0 \oplus T'_0) \, [\Tensor_\R(R)] \, (T''_1 \oplus T'_1) = T_1.
  \]
  
  \item {\bf Case of $C \equiv (\code{for} \ x^\lintt \ \code{in} \ \code{range}(n)\ \code{do}\ C')$:}
  \paragraph{\bf Dynamics of $C$.}
  By the semantics of the sequential composition and the assignment command, there exist 
  \[
    \sigma_{(0,k)}, \sigma_{(0,k)}' \in \State_\mathit{tgt} \quad\text{and}\quad
    T_{(0,k)} \in \Tensor_\R \quad\text{for all $k \in \{0,\ldots,n-1\}$}
  \]
  such that for all $k \in \{0,\ldots,n-1\}$,
  \begin{align*}
    \sigma_{(0,k)}' & {} = \sigma_{(0,k-1)}[x^\lintt : [i \mapsto k : i \in A_0]], &
    C', D, \sigma_{(0,k)}', A_0 & {} \Downarrowi \sigma_{(0,k)}, T_{(0,k)}, 
  \end{align*}
  where
  \begin{align*}
    \sigma_{(0,-1)} & = \sigma_0, &
    \sigma_{(0,n-1)} & = \sigma_0', &
    T_0 & = \left[
      i \mapsto \sum_{k = 0}^{n-1} T_{(0,k)}(i) \;:\; i \in A_0
    \right].
  \end{align*}
  \paragraph{\bf Dynamics of $\overline{C}$.}
  Let $\alpha$ be a fresh string such that $\alpha \notin \dom(i)$ for all $i \in A_0 \cup A_1$.
  We will use $\alpha$ to translate $C$ to $\overline{C}$.
  Define
  \[
  A_1' = \{i \concat [(\alpha,k)] : i \in A_1 \text{ and } k \in \{0,\ldots,n-1\}\}.
  \] 
  Let us use $\_$ to denote the real tensor $T_{A_1'}^z$.
  Then, by the semantics of the target language, there exist 
  \begin{align*}
    \sigma_{(1,k)},\sigma'_{(1,k)},\sigma''_{(1,k)} \in \State_\mathit{tgt} \quad\text{and}\quad
    T_{(1,k)} \in \Tensor_\R \quad\text{ for all $k \in \{0,\ldots,n-1\}$}
  \end{align*} 
  such that 
  \begin{enumerate}
  \item for all $k \in \{0,\ldots,n-1\}$,
  \begin{align*}
    \code{shift}(\alpha), D, \sigma_{(1,k-1)}, A_1' & {} \Downarrowi \sigma'_{(1,k)}, \_, \\
    (x^{\lintt} := \code{lookup\_index}(\alpha)), D, \sigma'_{(1,k)}, A_1' & {} \Downarrowi \sigma''_{(1,k)}, \_, \\
    \overline{C'}, D, \sigma''_{(1,k)}, A_1' & {} \Downarrowi \sigma_{(1,k)}, T_{(1,k)},
  \end{align*}
  \item for all $k \in \{0,\ldots,n-1\}$,
  \begin{align*}
    \rho_1'(i) = \begin{cases}
    i \concat [(\alpha,0)] & 
    \text{if } \alpha \not\in \dom(i)
    \text{ and } i \concat [(\alpha,0)] \in A_1',
    \\  
    j \concat [(\alpha,m+1)] &
    \begin{aligned}[t]
    & \text{if } i = j \concat [(\alpha,m)] \in A_1'{\downarrow} \\
    & \text{and } j \concat [(\alpha,m+1)] \in A_1' \text{ for some $j \in \Index$ and $m \geq 0$},
    \end{aligned}
    \\
    \text{undefined} & \text{otherwise},
    \end{cases}
  \end{align*}
  \begin{align*}
    \sigma_{(1,-1)} & {} = \sigma_1, &
    \sigma'_{(1,k)} & {} = \sigma_{(1,k-1)}\langle\rho_1'\rangle, &
    \sigma''_{(1,k)} & {} = \sigma'_{(1,k)}[x^\lintt : [i \mapsto i(\alpha) : i \in A_1']].
  \end{align*}
  \item finally, 
  \begin{align*}
    \begin{aligned}
    \rho_1 & {} = [i \concat [(\alpha,n-1)] \mapsto i : i \in A_1], \\
    \sigma_1' & {} = \sigma_{(1,n-1)}\langle\rho_1\rangle,
    \end{aligned} &&
    \begin{aligned}
    T_1 = \left[
      i \mapsto 
      \sum_{k = 0}^{n-1}
        T_{(1,n-1)}(i\concat [(\alpha,k)]) : i \in A_1
    \right].
    \end{aligned}
  \end{align*}
  \end{enumerate}
  
  \paragraph{\bf Overview.}
  The rest of the proof is about relating the executions of $C$ and $\overline{C}$ and deriving the conclusion of the theorem.
  For this purpose, we define relations $R^{(-1)}$ and $R^{(k)}$ on indices for all $k \in \{0,\ldots,n-1\}$ as follows:
  for all $i_0,i_1,i_1' \in \Index$,
  \begin{align*}
    i_0 \, [R^{(-1)}] \, i_1 & \iff i_0 \in A_0 \text{ and } i_1 \in A_1 \text { and } i_0 \, [R] \, i_1, \\
    i_0 \, [R^{(k)}] \, i_1' & \iff
    \exists \, i_1 \in A_1. \,
    \Big(
      i_0 \, [R^{(-1)}] \, i_1 \text{ and } i_1' = i_1 \concat [(\alpha,k)]
    \Big).
  \end{align*}
  Then, $R^{(-1)}$ and $R^{(k)}$ are couplings for all $k \in \{0,\ldots,n-1\}$.
  Our proof relies on the application of the induction hypothesis to $C'$ and $\overline{C'}$.
  Our proof consists of three parts.
  \begin{enumerate}
  \item The first part shows that
  \begin{equation*}
    A_0 \, [\cP(R^{(k)})] \, A_1' \quad \text{for all $k \in \{0,1,\ldots,n-1\}$}.
  \end{equation*}
  \item The second part shows that for all $k \in \{0,\ldots,n-1\}$ and $\ell \in \{0,\ldots,k\}$,
  \begin{equation*}
    \sigma_0 \, [\State_\mathit{tgt}(R^{(-1)})] \, \sigma''_{(1,k)},
    \quad\text{and}\quad
    \sigma'_{(0,\ell)} \, [\State_\mathit{tgt}(R^{(\ell)})] \, \sigma''_{(1,k)}.
  \end{equation*}
  \item The third part derives the claimed conclusion of the theorem.
  \end{enumerate}
   
  \paragraph{\bf First part.}
  We spell out the first part of the proof, namely, 
  \[
  A_0 \, [\cP(R^{(k)})] \, A_1' \quad\text{for all $k \in \{0,\ldots,n-1\}$}.
  \]
  Pick $k \in \{0,\ldots,n-1\}$, indices $i_0, i_1'$ such that $i_0 \, [R^{(k)}] \, i_1'$ arbitrarily.
  Then,
  \[
    \exists \, i_1 \in A_1. \,
    \Big(
      i_0 \, [R^{(-1)}] \, i_1 \text{ and } i_1' = i_1 \concat [(\alpha,k)]
    \Big).
  \]
  We only need to show that
  \[
    \max(A_0 \cap i_0{\downarrow}) \, [R^{(k)}] \, \max(A_1' \cap i_1'{\downarrow}).
  \]
  This is obvious since $\max(A_0 \cap i_0{\downarrow}) = i_0$ and $\max(A_1' \cap i_1'{\downarrow}) = i_1'$.
  
  \paragraph{\bf Second part.}
  We describe the second part of our proof:
  \[
  \sigma_0 \, [\State_\mathit{tgt}(R^{(-1)})] \, \sigma''_{(1,k)},
  \quad\text{and}\quad
  \sigma'_{(0,\ell)} \, [\State_\mathit{tgt}(R^{(\ell)})] \, \sigma''_{(1,k)}
  \quad
  \text{for all $\ell \in \{0,\ldots,k\}$}.
  \]
  The proof is based on the inductive argument on $k$. 
  
  \paragraph{Base case.}
  The base case is that
  \begin{equation*}
    \sigma_0 \, [\State_\mathit{tgt}(R^{(-1)})] \, \sigma_{(1,0)}''
    \quad\text{and}\quad
    \sigma_{(0,0)}' \, [\State_\mathit{tgt}(R^{(0)})] \, \sigma_{(1,0)}''.
  \end{equation*} 
  Recall the definitions of $\sigma_{(1,0)}''$ and $\sigma_{(0,0)}'$:
  \begin{align*}
    \sigma_{(1,0)}''
    & {} = \sigma_{(1,0)}'[x^{\lintt} : [i \mapsto i(\alpha) : i \in A_1']] \\
    & {} = \sigma_1\langle\rho_1'\rangle[x^{\lintt} : [i \mapsto i(\alpha) : i \in A_1']], \\
    \sigma_{(0,0)}' & {} = \sigma_0[x^{\lintt} : [i \mapsto 0 : i \in A_0]].
  \end{align*}
  To show the first relationship from above,
  pick a variable $y^{\tau\dagger}$ and indices $i_0$ and $i_1$ with $i_0 \, [R^{(-1)}] \, i_1$ arbitrarily.
  Then, $i_0 \in A_0$ and $i_1 \in A_1$ and $i_0 \, [R] \, i_1$.
  These imply that
  \begin{align*}
    \extend(\sigma_{(1,0)}''(y^{\tau\dagger}))(i_1)
    & {} =
    \extend(\sigma_1\langle\rho_1'\rangle[x^{\lintt} : [i \mapsto i(\alpha) : i \in A_1']](y^{\tau\dagger}))(i_1) \\
    & {} =
    \extend(\sigma_1\langle\rho_1'\rangle(y^{\tau\dagger}))(i_1) \\
    & {} = 
    \extend(\sigma_1(y^{\tau\dagger}))(i_1) \\
    & {} =
    \extend(\sigma_0(y^{\tau\dagger}))(i_0).
  \end{align*}
  The second equality holds by \Cref{lem:update-renaming-extend-appendix} because $i_1 \in A_1$ implies $i_1 \notin A_1'{\uparrow}$.
  The third equality also holds by \Cref{lem:update-renaming-extend-appendix}
  because $\image(\rho_1'){\uparrow} \subseteq A_1'{\uparrow}$ and $i_1 \notin A_1'{\uparrow}$ imply $i_1 \notin \image(\rho_1'){\uparrow}$.
  To prove the second relationship from above,
  pick a variable $y^{\tau\dagger}$ and indices $i_0, i_1'$ with $i_0 \, [R^{(0)}] \, i_1'$ arbitrarily.
  Then, $i_0 \in A_0$ and $i_1' \in A_1'$ and
  \[
    \exists \, i_1 \in A_1. \, \big( i_1' = i_0 \, [R] \, i_1 \text{ and } i_1' = i_1 \concat [(\alpha,0)] \big).
  \]
  We have to show that
  \[
  \extend(\sigma_{(0,0)}'(y^{\tau\dagger}))(i_0) = 
  \extend(\sigma_{(1,0)}''(y^{\tau\dagger}))(i_1).
  \]
  We discharge this proof obligation as follows:
  \begin{align*}
    \extend(\sigma_{(1,0)}''(y^{\tau\dagger}))(i_1')
    & {} = \extend(\sigma_1\langle \rho'_1\rangle[x^{\lintt} : [i \mapsto i(\alpha) : i \in A_1']](y^{\tau\dagger}))(i_1') \\
    & {} =
    \begin{cases}
      \extend([i \mapsto i(\alpha) : i \in A_1'])(i_1') & \text{if $y^{\tau\dagger} \equiv x^\lintt$} \\
      \extend(\sigma_1\langle \rho'_1\rangle(y^{\tau\dagger}))(i_1') & \text{otherwise}
    \end{cases} \\
    & {} =
    \begin{cases}
      0 & \text{if $y^{\tau\dagger} \equiv x^\lintt$} \\
      \extend(\sigma_1(y^{\tau\dagger}))(i_1) & \text{otherwise}
    \end{cases} \\
    & {} =
    \begin{cases}
      \extend([i \mapsto 0: i \in A_0])(i_0) & \text{if $y^{\tau\dagger} \equiv x^\lintt$} \\
      \extend(\sigma_0(y^{\tau\dagger}))(i_0) & \text{otherwise}
    \end{cases} \\
    & {} = \extend(\sigma[x^{\lintt} : [i \mapsto 0: i \in A_0]](y^{\tau\dagger}))(i_0) \\
    & {} = \extend(\sigma_{(0,0)}'(y^{\tau\dagger}))(i_0).
  \end{align*}
  The second equality holds by \Cref{lem:update-renaming-extend-appendix} because $i_1' \in A_1'$.
  The third equality also holds by \Cref{lem:update-renaming-extend-appendix} because $\rho_1'(i_1) = i_1 \concat [(\alpha,0)] = i_1'$.
  The fourth equality holds because $i_0 \, [R] \, i_1$ and $\sigma_0 \, [\State_\mathit{tgt}(R)] \, \sigma_1$.
  The fifth equality holds by \Cref{lem:update-renaming-extend-appendix} because $i_0 \in A_0$.

  \paragraph{Inductive case.}
  Let's move on to the inductive case. Consider $k \in \{1,\ldots,n-1\}$. Assume that
  \begin{align*}
  \sigma_0 \, [\State_\mathit{tgt}(R^{(-1)})] \, \sigma_{(1,k-1)}'', \quad\text{and}\quad
  \sigma_{(0,\ell)}' \, [\State_\mathit{tgt}(R^{(\ell)})] \, \sigma_{(1,k-1)}'' \quad\text{for all $\ell \in \{0,\ldots,k-1\}$}.
  \end{align*}
  We have to show that
  \begin{align*}
  \sigma_0 \, [\State_\mathit{tgt}(R^{(-1)})] \, \sigma_{(1,k)}'', \quad\text{and}\quad
  \sigma_{(0,\ell)}' \, [\State_\mathit{tgt}(R^{(\ell)})] \, \sigma_{(1,k)}'' \quad\text{for all $\ell \in \{0,\ldots,k\}$}.
  \end{align*}
  Recall the definitions of $\sigma_{(1,k)}''$ and $\sigma_{(0,\ell)}'$:
  \begin{align*}
    \sigma_{(1,k)}''
    & {} = \sigma_{(1,k)}'[x^{\lintt} : [i \mapsto i(\alpha) : i \in A_1']] \\
    & {} = \sigma_{(1,k-1)}\langle\rho_1'\rangle[x^{\lintt} : [i \mapsto i(\alpha) : i \in A_1']], \\
    \sigma_{(0,\ell)}' & {} = \sigma_{(0,\ell-1)}[x^{\lintt} : [i \mapsto \ell : i \in A_0]].
  \end{align*}

  To prove $\sigma_0 \, [\State_\mathit{tgt}(R^{(-1)})] \, \sigma_{(1,k)}''$,
  pick a variable $y^{\tau\dagger}$ and indices $i_0,i_1$ with $i_0 \, [R^{(-1)}] \, i_1$.
  Then, $i_0 \in A_0$ and $i_1 \in A_1$ and $i_0 \, [R] \, i_1$.
  These imply that
  \begin{align*}
    \extend(\sigma_{(1,k)}''(y^{\tau\dagger}))(i_1)
    & {} =
    \extend(\sigma_{(1,k-1)}\langle\rho_1'\rangle[x^{\lintt} : [i \mapsto i(\alpha) : i \in A_1']](y^{\tau\dagger}))(i_1) \\
    & {} =
    \extend(\sigma_{(1,k-1)}\langle\rho_1'\rangle(y^{\tau\dagger}))(i_1) \\
    & {} = 
    \extend(\sigma_{(1,k-1)}(y^{\tau\dagger}))(i_1) \\
    & {} =
    \extend(\sigma_{(1,k-1)}''(y^{\tau\dagger}))(i_1) \\
    & {} =
    \extend(\sigma_0(y^{\tau\dagger}))(i_0).
  \end{align*}
  The second equality holds by \Cref{lem:update-renaming-extend-appendix} because $i_1 \in A_1$ implies $i_1 \notin A_1'{\uparrow}$.
  The third equality also holds by \Cref{lem:update-renaming-extend-appendix}
  because $\image(\rho_1'){\uparrow} \subseteq A_1'{\uparrow}$ and $i_1 \notin A_1'{\uparrow}$ imply $i_1 \notin \image(\rho_1'){\uparrow}$.
  The fourth equality holds because $\sigma_{(1,k-1)}'' =_{(A_1'{\uparrow})^c} \sigma_{(1,k-1)}$
  by \Cref{lem:updated-indices-cmd} and $i_1 \notin A_1'{\uparrow}$.
  The fifth equality holds because
  \[
    i_0 \, [R^{(-1)}] \, i_1 \quad\text{and}\quad \sigma_0 \, [\State_\mathit{tgt}(R^{(-1)})] \, \sigma_{(1,k-1)}''.
  \]
  To prove that $\sigma'_{(0,\ell)} \, [\State_\mathit{tgt}(R^{(\ell)})] \, \sigma''_{(1,k)}$,
  pick a variable $y^{\tau\dagger}$ and indices $i_0, i_1'$ with $i_0 \, [R^{(\ell)}] \, i_1'$ arbitrarily.
  Then, $i_0 \in A_0$ and $i_1' \in A_1'$ and
  \[
    \exists \, i_1 \in A_1. \, \big( i_1' = i_0 \, [R] \, i_1 \text{ and } i_1' = i_1 \concat [(\alpha,\ell)] \big).
  \]
  If $\ell = 0$, then
  \begin{align*}
    \extend(\sigma_{(1,k)}''(y^{\tau\dagger}))(i_1')
    & {} = \extend(\sigma_{(1,k-1)}\langle\rho_1'\rangle[x^{\lintt} : [i \mapsto i(\alpha) : i \in A_1']](y^{\tau\dagger}))(i_1') \\
    & {} =
    \begin{cases}
      \extend([i \mapsto i(\alpha) : i \in A_1'])(i_1') & \text{if $y^{\tau\dagger} \equiv x^\lintt$} \\
      \extend(\sigma_{(1,k-1)}\langle\rho_1'\rangle(y^{\tau\dagger}))(i_1') & \text{otherwise}
    \end{cases} \\
    & {} =
    \begin{cases}
      0 & \text{if $y^{\tau\dagger} \equiv x^\lintt$} \\
      \extend(\sigma_{(1,k-1)}(y^{\tau\dagger}))(i_1) & \text{otherwise} \\
    \end{cases} \\
    & {} =
    \begin{cases}
      0 & \text{if $y^{\tau\dagger} \equiv x^\lintt$} \\
      \extend(\sigma_{(1,k-1)}''(y^{\tau\dagger}))(i_1) & \text{otherwise} \\
    \end{cases} \\
    & {} =
    \begin{cases}
      0 & \text{if $y^{\tau\dagger} \equiv x^\lintt$} \\
      \extend(\sigma_0(y^{\tau\dagger}))(i_0) & \text{otherwise} \\
    \end{cases} \\
    & {} = \extend(\sigma_0[x^\lintt : [i \mapsto 0 : i \in A_0]](y^{\tau\dagger}))(i_0) \\
    & {} = \extend(\sigma_{(0,0)}'(y^{\tau\dagger}))(i_0).
  \end{align*}
  The second equality holds by \Cref{lem:update-renaming-extend-appendix} because $i_1' \in A_1'{\uparrow}$.
  The third equality also holds by \Cref{lem:update-renaming-extend-appendix} because
  \begin{align*}
    i_1' = i_1 \concat [(\alpha,0)] = \rho_1'(i_1).
  \end{align*}
  The fourth equality holds because $\sigma_{(1,k-1)}'' =_{(A_1'{\uparrow})^c} \sigma_{(1,k-1)}$
  by \Cref{lem:updated-indices-cmd} and $i_1 \notin A_1'{\uparrow}$.
  The fifth equality holds because
  \begin{align*}
  i_0 \, [R^{(-1)}] \, i_1
  \quad\text{and}\quad
  \sigma_0 \, [\State_\mathit{tgt}(R^{(-1)})] \, \sigma_{(1,k-1)}''.
  \end{align*}
  The sixth equality holds by \Cref{lem:update-renaming-extend-appendix} because $i_0 \in A_0{\uparrow}$. 
  If $\ell \ge 1$, then
  \begin{align*}
    \extend(\sigma_{(1,k)}''(y^{\tau\dagger}))(i_1')
    & {} = \extend(\sigma_{(1,k-1)}\langle\rho_1'\rangle[x^{\lintt} : [i \mapsto i(\alpha) : i \in A_1']](y^{\tau\dagger}))(i_1') \\
    & {} =
    \begin{cases}
      \extend([i \mapsto i(\alpha) : i \in A_1'])(i_1') & \text{if $y^{\tau\dagger} \equiv x^\lintt$} \\
      \extend(\sigma_{(1,k-1)}\langle\rho_1'\rangle(y^{\tau\dagger}))(i_1') & \text{otherwise}
    \end{cases} \\
    & {} =
    \begin{cases}
      \ell & \text{if $y^{\tau\dagger} \equiv x^\lintt$} \\
      \extend(\sigma_{(1,k-1)}(y^{\tau\dagger}))(i_1 \concat [(\alpha,\ell-1)]) & \text{otherwise}
    \end{cases} \\
    & {} =
    \begin{cases}
      \ell & \text{if $y^{\tau\dagger} \equiv x^\lintt$} \\
      \extend(\sigma_{(0,\ell-1)}(y^{\tau\dagger}))(i_0) & \text{otherwise}
    \end{cases} \\
    & {} = \extend(\sigma_{(0,\ell-1)}[x^\lintt : [i \mapsto \ell : i \in A_0]](y^{\tau\dagger}))(i_0) \\
    & {} = \extend(\sigma_{(0,\ell)}'(y^{\tau\dagger}))(i_0).
  \end{align*}
  The second equality holds by \Cref{lem:update-renaming-extend-appendix} because $i_1' \in A_1'{\uparrow}$.
  The third equality also holds by \Cref{lem:update-renaming-extend-appendix} because
  \begin{align*}
    i_1' = i_1 \concat [(\alpha,\ell)] = \rho_1'(i_1 \concat [(\alpha,\ell-1)]).
  \end{align*}
  The fifth equality also holds by \Cref{lem:update-renaming-extend-appendix} because $i_0 \in A_0{\uparrow}$.
  The only nontrivial part is the fourth equality.
  The fourth equality follows from that
  \[
    i_0 \, [R^{(\ell-1)}] \, (i_1 \concat [(\alpha,\ell-1)]) \quad\text{and}\quad
    \sigma_{(0,\ell-1)} \, [\State_\mathit{tgt}(R^{(\ell-1)})] \, \sigma_{(1,k-1)}.
  \]
  The second relationship holds due to the following reasons:
  \begin{align*}
    C', D, \sigma_{(0,\ell-1)}', A_0 \Downarrowi \sigma_{(0,\ell-1)}, T_{(0,\ell-1)},
    \qquad
    \overline{C'}, D, \sigma_{(1,k-1)}'', A_1' \Downarrowi \sigma_{(1,k-1)}, T_{(1,k-1)},
  \end{align*}
  $\sigma_{(0,\ell-1)}' \, [\State_\mathit{tgt}(R^{(\ell-1)})] \, \sigma_{(1,k-1)}''$,
  and $A_0 \, [\cP(R^{(\ell-1)})] \, A_1'$,
  which are sufficient conditions to apply the induction hypothesis of the theorem on $C'$.

  \paragraph{\bf Third part.}
  Applying $k = n-1$ to the result of the second part, we have
  \[
    \sigma'_{(0,\ell)} \, [\State_\mathit{tgt}(R^{(\ell)})] \, \sigma''_{(1,n-1)}
    \quad\text{for all $\ell \in \{0,\ldots,n-1\}$}.
  \] 
  By the definition, we have
  \[
    C', D, \sigma_{(0,\ell)}', A_0 \Downarrowi \sigma_{(0,\ell)}, T_{(0,\ell)},
    \quad\text{and}\quad
    \overline{C'}, D, \sigma_{(1,n-1)}'', A_1' \Downarrowi \sigma_{(1,n-1)}, T_{(1,n-1)}.
  \]
  By the induction hypothesis on $C'$ and the fact $A_0 \, [\cP(R^{(\ell)})] \, A_1'$ from the first part, we have
  \begin{align*}
    \sigma_{(0,\ell)} \, [\State_\mathit{tgt}(R^{(\ell)})] \, \sigma_{(1,n-1)}
    \quad\text{and}\quad
    T_{(0,\ell)} \, [\Tensor_\R(R^{(\ell)})] \, T_{(1,n-1)} \quad\text{for all $\ell \in \{0,\ldots,n-1\}$}.
  \end{align*}
  In particular, $\ell = n-1$ gives $\sigma_{(0,n-1)} \, [\State_\mathit{tgt}(R^{(n-1)})] \, \sigma_{(1,n-1)}$.
  By using this, we want to show
  \[
    \sigma_0' \, [\State_\mathit{tgt}(R)] \, \sigma_1'.
  \]
  Recall the definitions of $\sigma_0'$ and $\sigma_1'$:
  \[
    \sigma_0' = \sigma_{(0,n-1)} \quad\text{and}\quad \sigma_1' = \sigma_{(1,n-1)}\langle\rho_1\rangle.
  \]
  Let $y^{\tau\dagger}$ be a variable and $i_0, i_1$ be indices such that $i_0 \, [R] \, i_1$.
  We first assume $i_0 \in A_0{\uparrow}$. Then, $A_0 \, [\cP(R)] \, A_1$ implies $i_1 \in A_1{\uparrow}$.
  This also implies that
  \begin{align*}
    \extend(\sigma_0'(y^{\tau\dagger}))(i_0)
    & {} =
    \extend(\sigma_0'(y^{\tau\dagger}))(\max(A_0 \cap i_0{\downarrow})) \\
    & {} =
    \extend(\sigma_{(0,n-1)}(y^{\tau\dagger}))(\max(A_0 \cap i_0{\downarrow})) \\
    & {} =
    \extend(\sigma_{(1,n-1)}(y^{\tau\dagger}))(\max(A_1 \cap i_1{\downarrow}) \concat [(\alpha,n-1)]) \\
    & {} =
    \extend(\sigma_{(1,n-1)}\langle\rho_1\rangle(y^{\tau\dagger}))(\max(A_1 \cap i_1{\downarrow})) \\
    & {} =
    \extend(\sigma_1'(y^{\tau\dagger}))(i_1).
  \end{align*}
  The first and the last equalities hold because
  $\sigma_0'$ is an $A_0$-state and $\sigma_1'$ is an $A_1$-state by \Cref{lem:preservation-L-states}.
  The third equality holds because
  \begin{align*}
    A_0 \, [\cP(R)] \, A_1
    & {} \implies
    \max(A_0 \cap i_0{\downarrow}) \, [R] \, \max(A_1 \cap i_1{\downarrow}) \\
    & {} \implies
    \max(A_0 \cap i_0{\downarrow}) \, [R^{(-1)}] \, \max(A_1 \cap i_1{\downarrow}) \\
    & {} \implies
    \max(A_0 \cap i_0{\downarrow}) \, [R^{(n-1)}] \, \max(A_1 \cap i_1{\downarrow}) \concat [(\alpha,n-1)].
  \end{align*}
  and $\sigma_{(0,n-1)} \, [\State_\mathit{tgt}(R^{(n-1)})] \, \sigma_{(1,n-1)}$.
  The fourth equality holds by \Cref{lem:update-renaming-extend-appendix} because
  \[
    \rho_1(\max(A_1 \cap i_1{\downarrow}) \concat [(\alpha,n-1)]) = \max(A_1 \cap i_1{\downarrow}).
  \]
  Now, we assume $i_0 \notin A_0{\uparrow}$ and $i_1 \notin A_1{\uparrow}$.
  Then, by \Cref{lem:updated-indices-cmd}, we have
  \begin{align*}
  \sigma_0 =_{(A_0{\uparrow})^c} \sigma_0',
  \quad\text{and}\quad
  \sigma_1 =_{(A_1{\uparrow})^c} \sigma_1'.
  \end{align*}
  This implies that
  \begin{align*}
    \extend(\sigma_0'(y^{\tau\dagger}))(i_0)
    & {} =
    \extend(\sigma_0(y^{\tau\dagger}))(i_0) \\
    & {} =
    \extend(\sigma_1(y^{\tau\dagger}))(i_1) \\
    & {} =
    \extend(\sigma_1'(y^{\tau\dagger}))(i_1).
  \end{align*}
  The second equality holds because
  \begin{align*}
  \sigma_0 \, [\State_\mathit{tgt}(R)] \, \sigma_1
  \quad\text{and}\quad
  i_0 \, [R] \, i_1.
  \end{align*}
  It remains to show the other conclusion of the theorem about $T_0$ and $T_1$:
  \[
    T_0 \, [\Tensor_\R(R)] \, T_1.
  \]
  Note that $\dom(T_0) = A_0$ and $\dom(T_1) = A_1$ by \Cref{lem:updated-indices-cmd}.
  For all $i_0 \in A_0$ and $i_1 \in A_1$, if $i_0 \, [R] \, i_1$, then
  \begin{align*}
    T_0(i_0)
    = \sum_{\ell = 0}^{n-1} T_{(0,\ell)}(i_0)
    = \sum_{\ell = 0}^{n-1} T_{(1,n-1)}(i_1 \concat [(\alpha,\ell)])
    = T_1(i_1),
  \end{align*}
  where the second equality holds since $T_{(0,\ell)} \, [\Tensor_\R(R^{(\ell)})] \, T_{(1,n-1)}$
  and $i_0 \, [R^{(\ell)}] \, (i_1 \concat [(\alpha,\ell)])$ for all $\ell \in \{0,\ldots,n-1\}$.
  \end{itemize}
\end{proof}
\subsection{Soundness of the Parallelising Translation from Source to Target Language}
\label{subsec:parallelising-translation}
In this subsection, we show the soundness of our parallelising translation.
The following lemma states that the command $C$ in the source language
and its parallelising translation $\overline{C}$ in the target language always successfully terminate
where the latter one is run with an antichain $\{[]\}$.

\begin{lemma}
  \label{lem:source-translation-no-error}
  Let $C$ be a command in the source language, and $\overline{C}$ be the parallelising translation of $C$ in the target language.
  Let $D \in \RDB$, $\sigma \in \State_\mathit{src}$, $\overline{\sigma} \in \State_\mathit{tgt}$, and $A \in \FAntiChain$ such that
  \begin{align*}
  \mathcal{S}(\overline{C}) \cap \dom(i) = \emptyset
  \text{ for all $i \in A$}.
  \end{align*}
  Then, there exist $\sigma' \in \State_\mathit{src}$, $\overline{\sigma}' \in \State_\mathit{tgt}$,
  $r \in \R$, and $T \in \Tensor_\R$ such that
  \begin{align*}
    C, D, \sigma \Downarrows \sigma', r
    \quad\text{and}\quad
    \overline{C}, D, \overline{\sigma}, A \Downarrowt \overline{\sigma}', T.
  \end{align*}
\end{lemma}
\begin{proof}
  The proofs are based on structural induction on $C$.
  \begin{enumerate}
    \item the source language\\
      Since every expression and update on an arbitrary state doesn't occur any error,
      we can always find the proper $\sigma', r$ for every syntax which has a rule with no premises by calculating the result of the rule.
      Therefore, what we need is to prove the case for the rules with premises.
      \begin{enumerate}
        \item $C \equiv (C_1;C_2)$:\\
          By the induction hypothesis applying to $C_1$ and $C_2$, there exist $\sigma_1$, $r_1$ such that $C_1, D, \sigma \Downarrows \sigma_1, r_1$,
          and also there exist $\sigma_2$, $r_2$ such that $C_2, D, \sigma_1 \Downarrows \sigma_2, r_2$.
          Therefore, $(C_1;C_2), D, \sigma \Downarrows \sigma_2, r_1+r_2$.
        \item $C \equiv (\code{ifz}\ Z\ C_1\ C_2)$:\\
          By the induction hypothesis applying to $C_1$ and $C_2$, there exist $\sigma_1$, $r_1$ such that $C_1, D, \sigma \Downarrows \sigma_1, r_1$,
          and also there exist $\sigma_2$, $r_2$ such that $C_2, D, \sigma \Downarrows \sigma_2, r_2$.
          If $\db{Z}\sigma = 0$, we can apply the first rule for $\code{ifz}$, so that $(\code{ifz}\ Z\ C_1\ C_2), D, \sigma \Downarrows \sigma_1, r_1$.
          Otherwise, $\db{Z}\sigma \neq 0$, then we can apply the second rule for $\code{ifz}$, so that $(\code{ifz}\ Z\ C_1\ C_2), D, \sigma \Downarrows \sigma_2, r_2$.
        \item $C \equiv (\code{for}\ x\ \code{in}\ \code{range}(n)\ \code{do}\ C')$:\\
          By the induction hypothesis applying to $C'$, there exist $\sigma_i, r_i$ for $1 <= i <= n$ such that $C', D, \sigma_{i-1}[x \mapsto i-1] \Downarrows \sigma_i,r_i$ for all $i$, where $\sigma_0 = \sigma$.\\
          Therefore, $(\code{for}\ x\ \code{in}\ \code{range}(n)\ \code{do}\ C'), D, \sigma_0 \Downarrows \sigma_n, \sum_{k = 1}^{n} r_k$
      \end{enumerate}
    \item the target language\\
      Similar to the source language, every expression and update on an arbitrary state and an arbitrary tensor are well-defined.
      Therefore, there exist $\overline{\sigma}', T$ for every syntax which has a rule with assignment premises, which only introduce new symbol for the result of operation or update (e.g. $T = \dots$).
      \begin{enumerate}
        \item $C \equiv (C_1;C_2)$:\\
          Since $\mathcal{S}(\overline{C_1}), \mathcal{S}(\overline{C_2}) \subseteq \mathcal{S}(\overline{C})$, we can apply the induction hypothesis to premises.
          Therefore, there exist $\overline{\sigma_1}$, $T_1$ such that $\overline{C_1}, D, \overline{\sigma}, A \Downarrowt \overline{\sigma_1}, T_1$,
          and also there exist $\overline{\sigma_2}$, $T_2$ such that $\overline{C_2}, D, \overline{\sigma_1}, A \Downarrowt \overline{\sigma_2}, T_2$.
          Therefore, $\overline{(C_1;C_2)}, D, \overline{\sigma}, A \Downarrowt \overline{\sigma_2}, (T_1 \oplus T_2)$.
        \item $C \equiv (\code{ifz}\ Z\ C_1\ C_2)$:\\
          Since $A_1, A_2 \subseteq A$ and $\mathcal{S}(\overline{C_1}), \mathcal{S}(\overline{C_2}) \subseteq \mathcal{S}(\overline{C})$, we can apply the induction hypothesis to premises.
          Therefore, there exist $\overline{\sigma_1}$, $T_1$ such that $\overline{C_1}, D, \overline{\sigma}, A_1 \Downarrowt \overline{\sigma_1}, T_1$,
          and also there exist $\overline{\sigma_2}$, $T_2$ such that $\overline{C_2}, D, \overline{\sigma_1}, A_2 \Downarrowt \overline{\sigma_2}, T_2$.
          Therefore, $\overline{(\code{ifz}\ Z\ C_1\ C_2)}, D, \overline{\sigma}, A \Downarrowt \overline{\sigma_2}, (T_1 \oplus T_2)$.
        \item $C \equiv (\code{for}\ x\ \code{in}\ \code{range}(n)\ \code{do}\ C_1)$:\\
          To simplify the proof, we introduce sub-commands of the translated code as follows:
          \begin{align*}
            \overline{C_2} & \defeq x^{\intt\dagger} := \code{lookup\_index}(\alpha); \; \overline{C_1}\\
            \overline{C_3} & \defeq \code{shift}(\alpha); \; \overline{C_2} \\
            \overline{C_4} & \defeq \code{loop\_fixpt\_noacc}(n) \{ \; \overline{C_3} \; \}.
          \end{align*}
          Then, $\overline{C} \equiv \code{extend\_index}(\alpha,n) \{ \; \overline{C_4} \; \}$.\\
          We will prove the whole part with top-down approaches.
          Each part is a proof for a statement which has a form as follows:\\
          For all $\overline{\sigma} \in \State_\mathit{tgt}$,
          there exist $\overline{\sigma}', T$ such that $\hat{C}, D, \overline{\sigma}, \hat{A} \Downarrowt \overline{\sigma}', T$,\\
          where $\hat{C}, \hat{A}$ are fixed to specific values.
          \begin{enumerate}
            \item $\hat{C} \equiv (\code{extend\_index}(\alpha,n) \{ \; \overline{C_4} \; \})$ and $\hat{A} \equiv A$ \quad (the target statement)\\
              Since $\mathcal{S}(\overline{C}) \cap \dom(i) = \emptyset$ for all $i \in A$, $\alpha \notin \dom(i)$ for all $i \in A$.
              Therefore, $A'$ and $\rho$ which are introduced in the inference rule of $\code{extend\_index}$ can be defined from $A$.
              Finally, if there exist $\overline{\sigma_4}$ and $T_4$ such that $\overline{C_4}, D, \overline{\sigma}, A' \Downarrowt \overline{\sigma_4}, T_4$,
              we can apply the inference rule of $\code{extend\_index}$ to $\overline{C}$,
              and it results as follows: $(\code{extend\_index}(\alpha,n_+)\ C), D, \overline{\sigma}, A \Downarrowt \overline{\sigma_4}\langle\rho\rangle, T_4'$, where $T_4' = [i \mapsto \sum_{0 \leq k < n_+} T_4(i \concat [(\alpha,k)]) : i \in A]$.
              Thus, what we need to do is to find such $\overline{\sigma_4}$ and $T_4$, and it can be achieved from the next statement.
            \item $\hat{C} \equiv (\code{loop\_fixpt\_noacc}(n) \{ \; \overline{C_3} \; \})$ and $\hat{A} \equiv A'$\\
              There exist two rules for $\code{loop\_fixpt\_noacc}(n)$, but premises of both rules can be derived from the common premises as follows:
              \[
                \sigma_0 = \overline{\sigma} \quad \overline{C_3}, D, \sigma_k, A' \Downarrowt \sigma_{k+1}, T_{k+1} \text{ for all $k \in \{0,\ldots,{n_+}-1\}$}
              \]
              If all states are different with each other, it satisfies the second rule, otherwise it satisfies the first rule.
              Moreover, this common premises can be satisfied from the next statement.
            \item $\hat{C} \equiv (\code{shift}(\alpha); \; \overline{C_2})$ and $\hat{A} \equiv A'$\\
              From the inference rule of $\code{shift}$, $\code{shift}(\alpha), D, \overline{\sigma}, A' \Downarrowt \overline{\sigma}\langle\rho\rangle', T^z_A$, where $\rho$ is from the premise of the rule.
              Then, we need to construct a premise for $\overline{C_2}$ for sequential composition, and it follows from the next statement.
            \item $\hat{C} \equiv (x^{\intt\dagger} := \code{lookup\_index}(\alpha)); \; \overline{C_1}$ and $\hat{A} \equiv A'$\\
              From the definition of $A'$, for any $i' \in A'$, there exist $i \in A, k \in \{0,\ldots,n_+{-}1\}$ such that $i' = i \concat [(\alpha,k)]$.
              Therefore, we can apply the rule of $\code{lookup\_index}$ to the first command.
              Moreover, $\mathcal{S}(\overline{C_1}) \cap \dom(i) \subseteq \mathcal{S}(\overline{C_1}) \cap (\dom(i')\cup\{\alpha\}) = (\mathcal{S}(\overline{C_1}) \cap \dom(i')) \cup (\mathcal{S}(\overline{C_1}) \cap \{\alpha\})$.
              Then, since $\alpha \notin \mathcal{S}(\overline{C_1})$ and $\mathcal{S}(\overline{C_1}) \subseteq \mathcal{S}(\overline{C})$, $\mathcal{S}(\overline{C_1}) \cap \dom(i) = \emptyset$.
              Thus, $A'$ satisfies the condition of induction hypothesis for $\overline{C_1}$, and finally we can execute $\overline{C_1}$ by applying the hypothesis.
          \end{enumerate}
      \end{enumerate}
  \end{enumerate} 
 
\end{proof}

Combining \Cref{thm:target-intermediate-equivalence},
\Cref{prop:sem-equiv}, \Cref{thm:soundness-translation}, and \Cref{lem:source-translation-no-error},
we show the following theorem, which describes a sense in which the parallelising translation of
$C$ in the source language to $\overline{C}$ in the target language preserves the semantics.

\begin{theorem}
  \label{thm:soundness-parallelising-translation}
  Let $C$ be a command in the source language, and $\overline{C}$ be the parallelising translation of $C$ in the target language.
  Let $D \in \RDB$, $\sigma_0 \in \State_{\mathit{src}}$, and $\sigma_1 \in \State_{\mathit{tgt}}$ such that
  \begin{align*}
    \sigma_0 \sim \sigma_1
    \quad\text{and}\quad
    \sigma_1 \text{ is a $\{[]\}$-state.}
  \end{align*}
  Then, there exist $\sigma_0' \in \State_{\mathit{src}}$, $\sigma_1' \in \State_{\mathit{tgt}}$,
  $r \in \R$, and $T \in \Tensor_\R$ such that
  \[
  C, D, \sigma_0 \Downarrows \sigma_0', r,
  \qquad
  \overline{C}, D, \sigma_1, \{[]\} \Downarrowt \sigma_1', T,
  \qquad
  \sigma_0' \sim \sigma_1',
  \quad\text{and}\quad
  r \sim T.
  \]
\end{theorem}
\begin{proof}
  By \Cref{lem:source-translation-no-error}, there exist $\sigma_0' \in \State_{\mathit{src}}$ and $r \in \R$ such that
  \[
  C, D, \sigma_0 \Downarrows \sigma_0', r.
  \]
  Let $C'$ be the command in the target language obtained from $C$
  by replacing all variables $x^\lintt$ with $x^{\lintt\dagger}$.
  Then, $\overline{C} \equiv \overline{C'}$.
  By \Cref{prop:sem-equiv}, since $\sigma_0 \sim \sigma_1$,
  there exist $\sigma_1'' \in \State_{\mathit{tgt}}$ and $T' \in \Tensor_\R$ such that
  \begin{align}
  \label{eqn:translation-soundness-1}
  C', D, \sigma_1, \{[]\} \Downarrowi \sigma_1'', T',
  \qquad
  \sigma_0' \sim \sigma_1'',
  \quad\text{and}\quad
  r \sim T'.
  \end{align}
  Again, by \Cref{lem:source-translation-no-error},
  there exist $\sigma_1' \in \State_{\mathit{tgt}}$ and $T \in \Tensor_\R$ such that
  \begin{align}
  \label{eqn:translation-soundness-2}
  \overline{C'}, D, \sigma_1, \{[]\} \Downarrowi \sigma_1', T.
  \end{align}
  Now, consider \Cref{thm:soundness-translation} where $R = \Delta_\Index$, and $A_0 = A_1 = \{[]\}$,
  and $\sigma_0 = \sigma_1$, and $\sigma_0' = \sigma_1''$, and $T_0 = T'$, and $T_1 = T$.
  Then,\Cref{eqn:translation-soundness-1} and \Cref{eqn:translation-soundness-2} imply that
  \begin{align*}
  \sigma_1'' =_\Index \sigma_1'
  \quad\text{and}\quad
  T' = T.
  \end{align*}
  Also, we conclude that $\sigma_0' \sim \sigma_1'$ and $r \sim T$, as desired.
  Finally, by \Cref{thm:target-intermediate-equivalence} and $\overline{C} \equiv \overline{C'}$,
  \Cref{eqn:translation-soundness-2} is equivalent to
  \begin{align*}
  \overline{C}, D, \sigma_1, \{[]\} \Downarrowt \sigma_1', T.
  \end{align*}
\end{proof}

\begin{corollary}
  \label{cor:main-theorem-soundness}
  Let $C$ be a command in the source language, and $\overline{C}$ be the parallelising translation of $C$ in the target language.
  Let $D \in \RDB$, $\sigma_0 \in \State_{\mathit{src}}$, and $\sigma_1 \in \State_{\mathit{tgt}}$. If  
  \[
      \sigma_1(x^{\tau\dagger}) = 
      [[] \mapsto \sigma_0(x^\tau)]
      \ \text{ for all variables $x^\tau$ in the source language,}
  \]
  then there exist $\sigma_0' \in \State_{\mathit{src}}$, $\sigma_1' \in \State_{\mathit{tgt}}$,
  and $r \in \R$ such that 
  \begin{align*}
    &
    C, D, \sigma_0 \Downarrows \sigma_0', r,
    \quad
    \overline{C}, D, \sigma_1, \{[]\} \Downarrowt \sigma_1', [[] \mapsto r],
    \quad\text{and}
    \quad
    \sigma'_1(x^{\tau\dagger}) = 
    [[] \mapsto \sigma'_0(x^\tau)]
    \ \text{ for all $x^\tau$.}
  \end{align*} 
\end{corollary}
\begin{proof}
  This is a direct consequence of \Cref{thm:soundness-parallelising-translation} by the definition of $\sigma_0 \sim \sigma_1$
  and $r \sim T$ and being a $\{[]\}$-state.
\end{proof}

\clearpage
\section{Benchmark Structures and Method Applicability}
\label{sec:benchmarks-methods}
In this section, we describe the structure of the benchmarks and the applicability of each method, as evaluated in \Cref{sec:eval}.
\Cref{table:model-info} shows the lengths of the loops in each model, and the averaged number of innermost iterations executed to compute the score when each method is applied.
Each method targets different classes of loop structures, which determines its applicability to the benchmarks.

The \mseq method does not apply any vectorisation and executes all iterations sequentially.
Although it is applicable to all benchmarks,
it performs inefficiently due to the large number of innermost iterations in all benchmarks.

The \mplate method applies \textit{plate constructors} to loop levels that do not contain data dependencies and if-then-else statements.
These loop levels are highlighted in light grey in \Cref{table:model-info},
and appear in all benchmarks except \texttt{arm} and \texttt{tcm}.

The \mvmarkov method builds on \mplate by also supporting \textit{vectorised Markov constructors}.
These can handle loop levels that may have data dependencies but are not nested and do not include if-then-else branches.
These loop levels are highlighted in dark grey in \Cref{table:model-info}.
As a result, \mvmarkov is applicable to the \texttt{hmm-*}, \texttt{dmm}, and \texttt{arm} benchmarks.

The \mdiscHMM method combines plate constructors and \textit{discrete HMM constructors},
which can vectorise loop levels corresponding to first-order discrete HMMs, as long as they are not nested.
Consequently, \mdiscHMM is only applicable to the \texttt{hmm-ord1} and \texttt{hmm-neural} benchmarks.

\rshl{The \mhand method manually re-writes vectorised code for all loops in the model.
This method requires specialised expertise in the internal mechanisms of the specific PPL being used,
along with the ability to manipulate tensor-based operations, which can be complex and error-prone (e.g., shape management).
For example, in our experiments using Pyro, which does not provide low-level primitives like $\code{fetch}$ or $\code{score}$ used in our source and target languages,
hand-coded vectorisation necessitated careful manipulation of internal variables used by Pyro's backend algorithms to implement equivalent functionality.
We also consider \mhand inapplicable to \texttt{tcm}, which contains if-then-else statements where the conditions themselves need to be vectorised,
as manually vectorising such conditional constructs is highly non-trivial.}

In contrast, \mours supports vectorisation across all loop structures present in the benchmarks.
It can handle data dependencies, nested structures, and if-then-else statements.
As a result, \mours is applicable to all benchmarks and consistently reduces the number of innermost iterations.

\vspace{10pt}
\begin{table}[h]
\footnotesize
\begin{tabular}{lrrrrrrrc}
\toprule
& \multicolumn{1}{c}{\makecell{loop lengths}} & \multicolumn{1}{c}{\makecell{\# iter of \\ \mseq}}
& \multicolumn{1}{c}{\makecell{\# iter of \\ \mplate}} & \multicolumn{1}{c}{\makecell{\# iter of \\ \mvmarkov}}
& \multicolumn{1}{c}{\makecell{\# iter of \\ \mdiscHMM}} & \multicolumn{1}{c}{\makecell{\rshl{\# iter of} \\ \rshl{\mhand}}}
& \multicolumn{1}{c}{\makecell{\# iter of \\ \mours}} & \multicolumn{1}{c}{\makecell{\rshl{inference} \\ \rshl{engines}}} \\
\midrule
\texttt{hmm-ord1}   & $\colorbox{lightgray!50}{229} \times \colorbox{darkgray!50}{129} \times \colorbox{lightgray!50}{51}$ & 1,506,591 & 129   & 3    & 1 & \rshl{1} & 2    & \rshl{MAP / MCMC} \\
\texttt{hmm-neural} & $\colorbox{lightgray!50}{229} \times \colorbox{darkgray!50}{129}$                                    & 29,541    & 129   & 3    & 1 & \rshl{1} & 2    & \rshl{MAP} \\
\texttt{hmm-ord2}   & $\colorbox{lightgray!50}{50}  \times \colorbox{darkgray!50}{129} \times \colorbox{lightgray!50}{51}$ & 328,950   & 129   & 5    & - & \rshl{1} & 3    & \rshl{MAP / MCMC} \\
\texttt{dmm}        & $\colorbox{lightgray!50}{229} \times \colorbox{darkgray!50}{129}$                                    & 29,541    & 129   & 3    & - & \rshl{1} & 2    & \rshl{SVI} \\
\texttt{nhmm-train} & $\colorbox{lightgray!50}{5}   \times 12 \times {31} \times {24}$                                     & 44,640    & 8,928 & -    & - & \rshl{1} & 8    & \rshl{MAP / MCMC} \\
\texttt{nhmm-stock} & $\colorbox{lightgray!50}{10}  \times 10 \times {4} \times {64}$                                      & 25,600    & 2,560 & -    & - & \rshl{1} & 8    & \rshl{MAP / MCMC} \\
\texttt{arm}        & $\colorbox{darkgray!50}{2000}$                                                                       & 2,000     & -     & 19.2 & - & \rshl{1} & 10.1 & \rshl{SVI / MCMC} \\
\texttt{tcm}        & $10 \times 200$                                                                                      & 2,000     & -     & -    & - & \rshl{-} & 2    & \rshl{SVI / MCMC} \\
\bottomrule
\end{tabular}
\caption{Lengths of loops in each benchmark, and averaged number of innermost iterations executed to compute the score for each method,
\rshl{and the inference engines used in the experiments.}
The averaged number of iterations for \texttt{arm} may not be an integer due to the randomly sampled order.
The loop levels highlighted in light grey are applicable to plate constructors, and those in dark grey are applicable to vectorised Markov constructors.
``-'' indicates that the corresponding method is not applicable to the benchmark.}
\vspace{-2em}
\label{table:model-info}
\end{table}

\clearpage
\definecolor{codegray}{rgb}{0.5,0.5,0.5}
\definecolor{codegreen}{HTML}{6A9955}
\definecolor{codebackground}{rgb}{0.95,0.95,0.95}

\lstset{
    language=Python,
    backgroundcolor=\color{codebackground},
    commentstyle=\color{codegray},
    keywordstyle=\color{blue},
    morekeywords={with},
    stringstyle=\color{purple},
    basicstyle=\ttfamily\footnotesize,
    breaklines=true,
    frame=single,
    frameround=tttt,
    numbersep=5pt,
    showspaces=false,
    showtabs=false,
    tabsize=2,
    emph={hmm_neural_seq, hmm_neural_plate, hmm_neural_vmarkov, hmm_neural_discHMM, hmm_neural_hand, hmm_neural_ours, MAP_step,
          sample, Dirichlet, mask, Categorical, to_event, eye, arange, unsqueeze, repeat, zeros, log, Bernoulli, 
          cat, shape, plate, vectorized_markov, where, DiscreteHMM, ones, Index, State, @vec, vec, vectorize, range,
          trace, get_trace, compute_log_prob, get_params, Adam, Data, TonesGeneratorNN},
    emphstyle=\color{codegreen},
    alsoletter={@},
}

\rshl{
\section{Code Snippet of Benchmarks}
\label{sec:code}
In this section, we provide code snippets of the \texttt{hmm-neural} benchmark using different vectorisation strategies.
Some components of the code (e.g., library imports, neural network definitions) are omitted for brevity.
The complete implementation of the benchmark is available at \url{https://github.com/Lim-Sangho/auto-vectorise-ppl.public}.

\Cref{fig:hmm-neural-seq} illustrates the baseline implementation without vectorisation (\mseq).
This version defines a hidden Markov model (HMM) with neural network emissions (\texttt{tones\_generator}),
which observes sequences of piano tones (\texttt{sequences}).
It uses two nested loops: the outer loop iterates over sequences of music in a batch,
and the inner loop iterates over the time steps within each sequence.
Each \texttt{sample} statement reads the values from the random database using variable names formatted as
\texttt{f"x\_\{i\}\_\{t\}"} and \texttt{f"y\_\{i\}\_\{t\}"},
where \texttt{i} and \texttt{t} denote the sequence and time step indices, respectively.
The first \texttt{sample} statement includes the argument \texttt{infer = \{"enumerate": "parallel"\}},
which instructs Pyro to perform parallel enumeration and variable elimination over the discrete latent variable.

\Cref{fig:hmm-neural-ours} shows our vectorised immplementation (\mours),
which employs the \texttt{@vec.vectorize} decorator and the \texttt{vec.range} constructor
to parallelise the loops while preserving the original sequential structure.
This version requires an additional argument, \texttt{s} of type \texttt{State}, to the model
for managing variable read and write operations,
and an \texttt{Index} operator to support $\mathrm{nan}$-mask based indexing,
where $\mathrm{nan}$ values represent empty indices in a partial map in the state.
These additions are orthogonal to the model structure and independent of tensor shapes or dimensionality,
allowing users to mitigate manual tensor manipulations
typically required in many existing vectorisation approaches.

\Cref{fig:hmm-neural-plate,fig:hmm-neural-vmarkov} present the implementations using the \emph{plate constructors} (\mplate),
and the combination of \emph{plate} and \emph{vectorised Markov constructors} (\mvmarkov), respectively.
These versions require explicit tensor shape manipulations (e.g., through slice indexing and \texttt{unsqueeze})
to align with the semantics of the respective constructors.
Also, these strategies are only valid when loop iterations are conditionally independent
or when data dependencies are explicitly specified by the user
(e.g., via the \texttt{history} argument in \texttt{vectorized\_markov}).
In the \texttt{hmm-neural} benchmark, the outer loop can be vectorised using \texttt{plate},
while the inner loop requires \texttt{vectorized\_markov} with an argument \texttt{history=1}
due to the dependencies on previous time steps.

\Cref{fig:hmm-neural-discHMM} shows the version implemented using \emph{plate constructors}
in combination with \emph{discrete HMM constructors} (\mdiscHMM).
This version leverages Pyro's \texttt{DiscreteHMM} distribution,
which is designed for discrete hidden Markov models.
It requires substantial restructuring of the original sequential code,
including explicit tensor shape management through operations such as \texttt{unsqueeze} and \texttt{cat}.

Finally, \Cref{fig:hmm-neural-hand} presents hand-vectorised implementation (\mhand).
Similar to the previous case, this version requires extensive manual modifications,
such as giving a non-standard argument (i.e., \texttt{infer=\{"\_do\_not\_trace": True, "is\_auxiliary": True\}})
to the \texttt{sample} statement to suppress errors caused by duplicate sampling variable names,
and explicitly handling tensor shapes through operations such as \texttt{unsqueeze}, \texttt{cat}, and \texttt{repeat}.
}

\begin{figure}[t]
\begin{lstlisting}
def hmm_neural_seq(sequences, lengths, tones_generator, hidden_dim=16):
    num_sequences, max_length, data_dim = sequences.shape
    probs_x = sample(
        "probs_x",
        Dirichlet(0.9 * eye(hidden_dim) + 0.1).to_event(1),
    )

    for i in range(num_sequences):
        x = 0
        y = zeros(data_dim)
        for t in range(max_length):
            with mask(mask=(t < lengths[i])):
                x = sample(
                    f"x_{i}_{t}",
                    Categorical(logits=probs_x[x].log()),
                    infer={"enumerate": "parallel"},
                )
                y = sample(
                    f"y_{i}_{t}",
                    Bernoulli(logits=tones_generator(x, y)).to_event(1),
                    obs=sequences[i, t],
                )
\end{lstlisting}
\vspace{-10pt}
\caption{\rshl{A code snippet of the \texttt{hmm-neural} benchmark with no vectorisation (\mseq).}}
\vspace{10pt}
\label{fig:hmm-neural-seq}
\end{figure}

\begin{figure}[t]
\begin{lstlisting}
@vec.vectorize
def hmm_neural_ours(s: State, sequences, lengths, tones_generator, hidden_dim=16):
    num_sequences, max_length, data_dim = sequences.shape
    probs_x = sample(
        "probs_x",
        Dirichlet(0.9 * eye(hidden_dim) + 0.1).to_event(1),
    )

    for i in vec.range("sequences", num_sequences, vectorized=True):
        s.x = 0
        s.y = zeros(data_dim)
        for t in vec.range("lengths", max_length, vectorized=True):
            with mask(mask=(t < lengths[i])):
                s.x = sample(
                    "x",
                    Categorical(logits=Index(probs_x)[s.x].log()),
                    infer={"enumerate": "parallel"},
                )
                s.y = sample(
                    "y",
                    Bernoulli(logits=tones_generator(s.x, s.y)).to_event(1),
                    obs=Index(sequences)[i, t],
                )
\end{lstlisting}
\vspace{-10pt}
\caption{\rshl{A code snippet of the \texttt{hmm-neural} benchmark implemented with our proposed vectorisation constructors (\mours).}}
\vspace{10pt}
\label{fig:hmm-neural-ours}
\end{figure}

\begin{figure}[t]
\begin{lstlisting}
def hmm_neural_plate(sequences, lengths, tones_generator, hidden_dim=16):
    num_sequences, max_length, data_dim = sequences.shape
    probs_x = sample(
        "probs_x",
        Dirichlet(0.9 * eye(hidden_dim) + 0.1).to_event(1),
    )

    with plate("sequences", num_sequences):
        x = 0
        y = zeros(data_dim)
        for t in range(max_length):
            with mask(mask=(t < lengths)):
                x = sample(
                    f"x_{t}",
                    Categorical(logits=probs_x[x].log()),
                    infer={"enumerate": "parallel"},
                )
                y = sample(
                    f"y_{t}",
                    Bernoulli(logits=tones_generator(x, y)).to_event(1),
                    obs=sequences[:, t],
                )
\end{lstlisting}
\vspace{-10pt}
\caption{\rshl{A code snippet of the \texttt{hmm-neural} benchmark implemented with \emph{plate constructors} (\mplate).}}
\vspace{10pt}
\label{fig:hmm-neural-plate}
\end{figure}

\begin{figure}[t]
\begin{lstlisting}
def hmm_neural_vmarkov(sequences, lengths, tones_generator, hidden_dim=16):
    num_sequences, max_length, data_dim = sequences.shape
    probs_x = sample(
        "probs_x",
        Dirichlet(0.9 * eye(hidden_dim) + 0.1).to_event(1),
    )

    with plate("sequences", num_sequences, dim=-2) as batch:
        batch = batch.unsqueeze(-1)
        x = 0
        y = zeros(data_dim)
        for t in vectorized_markov(None, "lengths",
                                   max_length, dim=-1, history=1):
            with mask(mask=(t < lengths.unsqueeze(-1))):
                x = sample(
                    f"x_{t}",
                    Categorical(logits=probs_x[x].log()),
                    infer={"enumerate": "parallel"},
                )
                y = sample(
                    f"y_{t}",
                    Bernoulli(logits=tones_generator(x, y)).to_event(1),
                    obs=sequences[batch, t],
                )
\end{lstlisting}
\vspace{-10pt}
\caption{\rshl{A code snippet of the \texttt{hmm-neural} benchmark implemented with \emph{plate constructors} and \emph{vectorised Markov constructors} (\mvmarkov).}}
\vspace{10pt}
\label{fig:hmm-neural-vmarkov}
\end{figure}

\begin{figure}[t]
\begin{lstlisting}
def hmm_neural_discHMM(sequences, lengths, tones_generator, hidden_dim=16):
    num_sequences, max_length, data_dim = sequences.shape
    probs_x = sample(
        "probs_x",
        Dirichlet(0.9 * eye(hidden_dim) + 0.1).to_event(1),
    )

    with plate("sequences", num_sequences):
        mask = (arange(max_length) < lengths.unsqueeze(-1))
        x = arange(hidden_dim)
        y = cat([zeros(num_sequences, 1, data_dim), sequences[:, :-1]], dim=1)  
        init_logits = cat([ones(1), zeros(hidden_dim - 1)]).log()  
        trans_logits = where(mask.unsqueeze(-1).unsqueeze(-1), probs_x.log(), 0)  
        obs_dist = Bernoulli(
            logits=tones_generator(x, y.unsqueeze(-2))
        ).to_event(1)  
        obs_dist = obs_dist.mask(mask.unsqueeze(-1))  
        hmm_dist = DiscreteHMM(init_logits, trans_logits, obs_dist)  
        sample("y", hmm_dist, obs=sequences)  
\end{lstlisting}
\vspace{-10pt}
\caption{\rshl{A code snippet of the \texttt{hmm-neural} benchmark implemented with \emph{plate constructors} and \emph{discrete HMM constructors} (\mdiscHMM).}}
\vspace{10pt}
\label{fig:hmm-neural-discHMM}
\end{figure}

\begin{figure}[t]
\begin{lstlisting}
def hmm_neural_hand(sequences, lengths, tones_generator, hidden_dim=16):
    num_sequences, max_length, data_dim = sequences.shape
    probs_x = sample(
        "probs_x",
        Dirichlet(0.9 * eye(hidden_dim) + 0.1).to_event(1)
    )

    with mask(mask=(arange(max_length) < lengths.unsqueeze(-1))):
        x_prev = sample(
            "x",
            Categorical(logits=zeros(hidden_dim)),
            infer={"enumerate": "parallel",
                    "_do_not_trace": True,
                    "is_auxiliary": True}
        )
        x_prev = cat([zeros(1),
                      x_prev.repeat(1, num_sequences, max_length-1)],
                      dim=-1)

        x_curr = sample(
            "x",
            Categorical(logits=probs_x[x_prev].log()),
            infer={"enumerate": "parallel"}
        )

        y_prev = cat([zeros(num_sequences, 1, data_dim),
                      sequences[:, :-1]],
                      dim=-2)

        sample(
            "y",
            Bernoulli(logits=tones_generator(x_curr, y_prev)).to_event(1),
            obs=sequences
        )
\end{lstlisting}
\vspace{-10pt}
\caption{\rshl{A code snippet of the \texttt{hmm-neural} benchmark implemented with hand-coded vectorisation (\mhand).}}
\vspace{10pt}
\label{fig:hmm-neural-hand}
\end{figure}

\begin{figure}[t]
\begin{lstlisting}
def MAP_step(model, optim, sequence, lengths, tones_generator):
    model_trace = trace(model).get_trace(sequence, lengths, tones_generator)
    log_prob = model_trace.compute_log_prob()
    params = model_trace.get_params()
    optim(log_prob, params)  # Update params to maximise log_prob

optim = Adam({"lr": 1e-1})
sequence, lengths = Data()
tones_generator = TonesGeneratorNN()
model = hmm_neural_ours  # or hmm_neural_seq, hmm_neural_plate, etc.

for step in range(num_steps):
    MAP_step(model, optim, sequence, lengths, tones_generator)
\end{lstlisting}
\vspace{-10pt}
\caption{\rshl{A code snippet of a training step for MAP inference used in all versions of the \texttt{hmm-neural} benchmark.}}
\vspace{10pt}
\label{fig:hmm-neural-inference}
\end{figure}

\clearpage
\section{Additional Results from Experimental Evaluation}
\label{sec:eval-appendix}
\rshl{
In this section, we provide detailed results from the experimental evaluation on \Cref{rq:time-memory} in \Cref{sec:eval}.
\begin{itemize}
\item \Cref{table:time-svi} shows the averaged time to compute the sum of scores (score time) and to run the single entire training step (total time) of each model in SVI or MAP inference.
\item \Cref{table:memory-svi} shows the maximum GPU memory usage during the execution of each training step in SVI or MAP inference.
\item \Cref{table:steps-mcmc} shows the averaged time to execute each MCMC step for each model.
\item \Cref{table:memory-mcmc} shows the maximum GPU memory usage during the execution of each MCMC step.
\end{itemize}
All tables also include the ratios over \mours and standard deviations.}

Our results show that reducing the number of innermost iterations leads to greater time savings.
For example, when comparing \mours with \mplate,
in the case of \texttt{hmm-ord2}, \mplate executes 129 innermost iterations,
whereas \mours reduces this to just 3, achieving a $43\times$ reduction in iteration count,
which results in $4.5\times$ speedup in score time and $1.8\times$ speedup in total time
\rshl{during SVI or MAP inference, and $1.8\times$ speedup in time per MCMC step}.
In contrast, for \texttt{nhmm-train}, the number of innermost iterations decreases from 8,928 (\mplate) to 8 (\mours),
resulting in a reduction of over $1,116\times$, leading to significantly greater time savings:
$145.8\times$ speedup in score time and $78.4\times$ speedup in total time
\rshl{during SVI or MAP inference, and $65.6\times$ speedup in time per MCMC step}.

Additionally, we observe that \mours achieves memory reductions in some benchmarks compared to \mseq,
especially in those with relatively low memory demands, such as \texttt{nhmm-stock}, \texttt{arm}, and \texttt{tcm}.
This is because, in these benchmarks, the \mseq method often operates on small-sized vectors within each innermost iteration.
When the GPU allocates memory for these vectors, it tends to allocate many small chunks, which increases memory overhead.
In contrast, \mours allocates larger memory chunks at once due to its vectorised structure, resulting in lower memory overhead.

\begin{table}[t]
\rshl{
\scriptsize
\begin{tabular}{lrrrrrr|rrrrrr}
\toprule
& \multicolumn{6}{c}{SVI/MAP score time (s) ($\downarrow$)} & \multicolumn{6}{c}{SVI/MAP total time (s) ($\downarrow$)} \\
& \multicolumn{1}{c}{\mseq} & \multicolumn{1}{c}{\mplate}
& \multicolumn{1}{c}{\mvmarkov} & \multicolumn{1}{c}{\mdiscHMM}
& \multicolumn{1}{c}{\mhand} & \multicolumn{1}{c}{\mours}
& \multicolumn{1}{c}{\mseq} & \multicolumn{1}{c}{\mplate}
& \multicolumn{1}{c}{\mvmarkov} & \multicolumn{1}{c}{\mdiscHMM}
& \multicolumn{1}{c}{\mhand} & \multicolumn{1}{c}{\mours} \\
\midrule
\texttt{hmm-ord1}   & 669.097 & 0.118  & 0.054 & 0.006 & 0.006 & 0.021 & 1066.060 & 0.192  & 0.083 & 0.018 & 0.035 & 0.075 \\
\texttt{hmm-neural} & 35.715  & 0.158  & 0.056 & 0.010 & 0.009 & 0.028 & 60.474   & 0.275  & 0.082 & 0.027 & 0.042 & 0.062 \\
\texttt{hmm-ord2}   & 131.594 & 0.124  & 0.165 & -     & 0.009 & 0.027 & 205.457  & 0.282  & 0.280 & -     & 0.127 & 0.158 \\
\texttt{dmm}        & 42.092  & 0.199  & 0.076 & -     & 0.004 & 0.016 & 129.339  & 0.957  & 0.524 & -     & 0.427 & 0.496 \\
\texttt{nhmm-train} & 91.287  & 18.523 & -     & -     & 0.015 & 0.127 & 121.993  & 25.021 & -     & -     & 0.198 & 0.319 \\
\texttt{nhmm-stock} & 29.925  & 3.058  & -     & -     & 0.005 & 0.089 & 57.417   & 6.247  & -     & -     & 0.022 & 0.139 \\
\texttt{arm}        & 1.129   & -      & 0.177 & -     & 0.003 & 0.031 & 2.354    & -      & 0.199 & -     & 0.009 & 0.043 \\
\texttt{tcm}        & 1.653   & -      & -     & -     & -     & 0.029 & 3.837    & -      & -     & -     & -     & 0.045 \\
\midrule
& \multicolumn{6}{c}{SVI/MAP score time ratio ($\ \cdot \ /$ \mours) ($\downarrow$)} & \multicolumn{6}{c}{SVI/MAP total time ratio ($\ \cdot \ /$ \mours) ($\downarrow$)} \\
& \multicolumn{1}{c}{\mseq} & \multicolumn{1}{c}{\mplate} & \multicolumn{1}{c}{\mvmarkov}
& \multicolumn{1}{c}{\mdiscHMM} & \multicolumn{1}{c}{\mhand} & \multicolumn{1}{c}{\mours}
& \multicolumn{1}{c}{\mseq} & \multicolumn{1}{c}{\mplate} & \multicolumn{1}{c}{\mvmarkov}
& \multicolumn{1}{c}{\mdiscHMM} & \multicolumn{1}{c}{\mhand} & \multicolumn{1}{c}{\mours} \\
\midrule
\texttt{hmm-ord1}   & 31399.3 & 5.6   & 2.6 & 0.3 & 0.3 & 1 & 14302.9 & 2.6  & 1.1 & 0.2 & 0.5 & 1 \\
\texttt{hmm-neural} & 1291.6  & 5.7   & 2.0 & 0.4 & 0.3 & 1 & 972.4   & 4.4  & 1.3 & 0.4 & 0.7 & 1 \\
\texttt{hmm-ord2}   & 4824.4  & 4.5   & 6.0 & -   & 0.3 & 1 & 1300.4  & 1.8  & 1.8 & -   & 0.8 & 1 \\
\texttt{dmm}        & 2657.3  & 12.6  & 4.8 & -   & 0.3 & 1 & 260.7   & 1.9  & 1.1 & -   & 0.9 & 1 \\
\texttt{nhmm-train} & 718.5   & 145.8 & -   & -   & 0.1 & 1 & 382.2   & 78.4 & -   & -   & 0.6 & 1 \\
\texttt{nhmm-stock} & 337.0   & 34.4  & -   & -   & 0.1 & 1 & 414.1   & 45.1 & -   & -   & 0.2 & 1 \\
\texttt{arm}        & 36.2    & -     & 5.7 & -   & 0.1 & 1 & 54.8    & -    & 4.6 & -   & 0.2 & 1 \\
\texttt{tcm}        & 57.7    & -     & -   & -   & -   & 1 & 85.6    & -    & -   & -   & -   & 1 \\
\midrule
& \multicolumn{6}{c}{SVI/MAP score time std. (s)} & \multicolumn{6}{c}{SVI/MAP total time std. (s)} \\
& \multicolumn{1}{c}{\mseq} & \multicolumn{1}{c}{\mplate}
& \multicolumn{1}{c}{\mvmarkov} & \multicolumn{1}{c}{\mdiscHMM}
& \multicolumn{1}{c}{\mhand} & \multicolumn{1}{c}{\mours}
& \multicolumn{1}{c}{\mseq} & \multicolumn{1}{c}{\mplate}
& \multicolumn{1}{c}{\mvmarkov} & \multicolumn{1}{c}{\mdiscHMM}
& \multicolumn{1}{c}{\mhand} & \multicolumn{1}{c}{\mours} \\
\midrule
\texttt{hmm-ord1}   & 53.727 & 0.018 & 0.004 & 0.000 & 0.000 & 0.001 & 50.168 & 0.018 & 0.005 & 0.001 & 0.001 & 0.002 \\
\texttt{hmm-neural} & 0.387  & 0.005 & 0.004 & 0.000 & 0.000 & 0.001 & 3.242  & 0.019 & 0.005 & 0.001 & 0.001 & 0.002 \\
\texttt{hmm-ord2}   & 1.919  & 0.019 & 0.008 & -     & 0.000 & 0.000 & 2.344  & 0.021 & 0.010 & -     & 0.008 & 0.005 \\
\texttt{dmm}        & 0.555  & 0.031 & 0.002 & -     & 0.000 & 0.000 & 9.477  & 0.046 & 0.009 & -     & 0.004 & 0.005 \\
\texttt{nhmm-train} & 1.036  & 0.223 & -     & -     & 0.000 & 0.001 & 1.132  & 0.234 & -     & -     & 0.003 & 0.003 \\
\texttt{nhmm-stock} & 2.025  & 0.056 & -     & -     & 0.000 & 0.001 & 2.498  & 0.076 & -     & -     & 0.002 & 0.002 \\
\texttt{arm}        & 0.064  & -     & 0.065 & -     & 0.001 & 0.010 & 0.143  & -     & 0.066 & -     & 0.001 & 0.010 \\
\texttt{tcm}        & 0.064  & -     & -     & -     & -     & 0.000 & 0.220  & -     & -     & -     & -     & 0.000 \\
\bottomrule
\end{tabular}
}
\caption{\rshl{Averaged score time (s) and total time (s) for each training step,
their ratios over $\mours$, and standard deviations in SVI or MAP inference.
We report the average over 5 runs $\times$ 9 training steps for \mseq, and 5 runs $\times$ 990 training steps for other methods.
Lower is better.}
}
\label{table:time-svi}
\end{table}

\begin{table}[t]
\rshl{
\small
\begin{tabular}{lrrrrrr}
\toprule
& \multicolumn{6}{c}{SVI/MAP GPU memory usage (GB) ($\downarrow$)} \\
& \multicolumn{1}{c}{\mseq} & \multicolumn{1}{c}{\mplate}
& \multicolumn{1}{c}{\mvmarkov} & \multicolumn{1}{c}{\mdiscHMM}
& \multicolumn{1}{c}{\mhand} & \multicolumn{1}{c}{\mours} \\
\midrule
\texttt{hmm-ord1}   & 5.053 & 0.700 & 1.582 & 0.277 & 0.888 & 1.774 \\
\texttt{hmm-neural} & 0.723 & 0.518 & 1.300 & 0.483 & 1.002 & 1.043 \\
\texttt{hmm-ord2}   & 1.902 & 0.868 & 5.229 & -     & 2.289 & 2.467 \\
\texttt{dmm}        & 0.871 & 0.622 & 1.201 & -     & 0.679 & 0.739 \\
\texttt{nhmm-train} & 3.342 & 2.851 & -     & -     & 3.899 & 3.989 \\
\texttt{nhmm-stock} & 0.260 & 0.057 & -     & -     & 0.039 & 0.040 \\
\texttt{arm}        & 0.015 & -     & 0.011 & -     & 0.000 & 0.002 \\
\texttt{tcm}        & 0.022 & -     & -     & -     & -     & 0.001 \\
\midrule
& \multicolumn{6}{c}{SVI/MAP GPU memory usage ratio ($\ \cdot \ /$ \mours) ($\downarrow$)} \\
& \multicolumn{1}{c}{\mseq} & \multicolumn{1}{c}{\mplate}
& \multicolumn{1}{c}{\mvmarkov} & \multicolumn{1}{c}{\mdiscHMM}
& \multicolumn{1}{c}{\mhand} & \multicolumn{1}{c}{\mours} \\
\midrule
\texttt{hmm-ord1}   & 2.849  & 0.395 & 0.892 & 0.156 & 0.501 & 1 \\
\texttt{hmm-neural} & 0.693  & 0.496 & 1.246 & 0.463 & 0.961 & 1 \\
\texttt{hmm-ord2}   & 0.771  & 0.352 & 2.120 & -     & 0.928 & 1 \\
\texttt{dmm}        & 1.179  & 0.841 & 1.626 & -     & 0.919 & 1 \\
\texttt{nhmm-train} & 0.838  & 0.715 & -     & -     & 0.977 & 1 \\
\texttt{nhmm-stock} & 6.487  & 1.420 & -     & -     & 0.963 & 1 \\
\texttt{arm}        & 9.842  & -     & 7.313 & -     & 0.193 & 1 \\
\texttt{tcm}        & 28.872 & -     & -     & -     & -     & 1 \\
\midrule
& \multicolumn{6}{c}{SVI/MAP GPU memory usage (GB) std.} \\
& \multicolumn{1}{c}{\mseq} & \multicolumn{1}{c}{\mplate}
& \multicolumn{1}{c}{\mvmarkov} & \multicolumn{1}{c}{\mdiscHMM}
& \multicolumn{1}{c}{\mhand} & \multicolumn{1}{c}{\mours} \\
\midrule
\texttt{hmm-ord1}   & 0.000 & 0.000 & 0.000 & 0.000 & 0.000 & 0.000 \\
\texttt{hmm-neural} & 0.000 & 0.000 & 0.000 & 0.000 & 0.000 & 0.000 \\
\texttt{hmm-ord2}   & 0.000 & 0.000 & 0.000 & -     & 0.000 & 0.000 \\
\texttt{dmm}        & 0.006 & 0.001 & 0.001 & -     & 0.001 & 0.001 \\
\texttt{nhmm-train} & 0.003 & 0.000 & -     & -     & 0.000 & 0.000 \\
\texttt{nhmm-stock} & 0.001 & 0.000 & -     & -     & 0.000 & 0.000 \\
\texttt{arm}        & 0.000 & -     & 0.001 & -     & 0.000 & 0.000 \\
\texttt{tcm}        & 0.000 & -     & -     & -     & -     & 0.000 \\
\bottomrule
\end{tabular}
}
\caption{\rshl{Averaged GPU memory usage (GB) for each training step,
its ratio over $\mours$, and standard deviation in SVI or MAP inference.
We report the average over 5 runs $\times$ 9 training steps for \mseq, and 5 runs $\times$ 990 training steps for other methods.
Lower is better.}
}
\label{table:memory-svi}
\end{table}

\begin{table}[t]
\rshl{
\small
\begin{tabular}{lrrrrrr}
\toprule
& \multicolumn{6}{c}{MCMC time per step (s) ($\downarrow$)} \\
& \mseq & \mplate & \mvmarkov & \mdiscHMM & \mhand & \mours \\
\midrule
\texttt{hmm-ord1} & $\mathrm{nan}$ & 1.872 & 0.844 & 0.223 & 0.399 & 0.786 \\
\texttt{hmm-ord2} & 1868.995 & 2.826 & 2.734 & - & 1.238 & 1.545 \\
\texttt{nhmm-train} & 1232.830 & 245.348 & - & - & 2.512 & 3.738 \\
\texttt{nhmm-stock} & 510.519 & 52.833 & - & - & 0.560 & 1.439 \\
\texttt{arm} & 21.952 & - & 0.489 & - & 0.049 & 0.191 \\
\texttt{tcm} & 28.087 & - & - & - & - & 0.417 \\
\midrule
& \multicolumn{6}{c}{MCMC time per step ratio ($\ \cdot \ /$ \mours) ($\downarrow$)} \\
& \mseq & \mplate & \mvmarkov & \mdiscHMM & \mhand & \mours \\
\midrule
\texttt{hmm-ord1} & $\mathrm{nan}$ & 2.383 & 1.074 & 0.284 & 0.507 & 1 \\
\texttt{hmm-ord2} & 1209.649 & 1.829 & 1.769 & - & 0.801 & 1 \\
\texttt{nhmm-train} & 329.769 & 65.628 & - & - & 0.672 & 1 \\
\texttt{nhmm-stock} & 354.730 & 36.710 & - & - & 0.389 & 1 \\
\texttt{arm} & 114.806 & - & 2.557 & - & 0.256 & 1 \\
\texttt{tcm} & 67.277 & - & - & - & - & 1 \\
\midrule
& \multicolumn{6}{c}{MCMC time per step (s) std.} \\
& \mseq & \mplate & \mvmarkov & \mdiscHMM & \mhand & \mours \\
\midrule
\texttt{hmm-ord1} & $\mathrm{nan}$ & 0.093 & 0.020 & 0.016 & 0.020 & 0.027 \\
\texttt{hmm-ord2} & 84.462 & 0.176 & 0.100 & - & 0.067 & 0.096 \\
\texttt{nhmm-train} & 38.036 & 6.592 & - & - & 0.036 & 0.081 \\
\texttt{nhmm-stock} & 30.262 & 2.965 & - & - & 0.035 & 0.078 \\
\texttt{arm} & 0.902 & - & 0.069 & - & 0.004 & 0.018 \\
\texttt{tcm} & 0.859 & - & - & - & - & 0.027 \\
\bottomrule
\end{tabular}
}
\caption{
  \rshl{
    Average time per MCMC step (s) in 1 hour, its ratio over $\mours$, and standard deviation.
    Each reported number is averaged over 8 runs $\times$ averaged number of steps executed in each one-hour run.
    A dash (-) indicates that the method is not applicable to the benchmark.
    A value of $\mathrm{nan}$ indicates it fails to execute a single MCMC step within 1 hour.
    Lower is better.
  }
}
\label{table:steps-mcmc}
\end{table}

\begin{table}[t]
\rshl{
\small
\begin{tabular}{lrrrrrr}
\toprule
& \multicolumn{6}{c}{MCMC GPU memory usage (GB) ($\downarrow$)} \\
& \mseq & \mplate & \mvmarkov & \mdiscHMM & \mhand & \mours \\
\midrule
\texttt{hmm-ord1} & $\mathrm{nan}$ & 0.7000 & 1.2663 & 0.2464 & 0.6681 & 1.3834 \\
\texttt{hmm-ord2} & 1.8299 & 0.8659 & 3.8962 & - & 2.0342 & 2.1895 \\
\texttt{nhmm-train} & 3.3428 & 2.8506 & - & - & 3.6922 & 3.7635 \\
\texttt{nhmm-stock} & 0.2446 & 0.0517 & - & - & 0.0344 & 0.0358 \\
\texttt{arm} & 0.0103 & - & 0.0090 & - & 0.0002 & 0.0022 \\
\texttt{tcm} & 0.0139 & - & - & - & - & 0.0005 \\
\midrule
& \multicolumn{6}{c}{MCMC GPU memory usage (GB) ratio ($\ \cdot \ /$ \mours) ($\downarrow$)} \\
& seq & plate & vmarkov & discHMM & manual & ours \\
\midrule
\texttt{hmm-ord1} & $\mathrm{nan}$ & 0.5060 & 0.9153 & 0.1781 & 0.4830 & 1 \\
\texttt{hmm-ord2} & 0.8358 & 0.3955 & 1.7795 & - & 0.9291 & 1 \\
\texttt{nhmm-train} & 0.8882 & 0.7574 & - & - & 0.9811 & 1 \\
\texttt{nhmm-stock} & 6.8400 & 1.4457 & - & - & 0.9607 & 1 \\
\texttt{arm} & 4.7768 & - & 4.1773 & - & 0.0775 & 1 \\
\texttt{tcm} & 29.5086 & - & - & - & - & 1 \\
\midrule
& \multicolumn{6}{c}{MCMC GPU memory usage (GB) std.} \\
& seq & plate & vmarkov & discHMM & manual & ours \\
\midrule
\texttt{hmm-ord1} & $\mathrm{nan}$ & 0.0000 & 0.0000 & 0.0000 & 0.0000 & 0.0000 \\
\texttt{hmm-ord2} & 0.0000 & 0.0000 & 0.0000 & - & 0.0000 & 0.0000 \\
\texttt{nhmm-train} & 0.0000 & 0.0000 & - & - & 0.0000 & 0.0000 \\
\texttt{nhmm-stock} & 0.0000 & 0.0000 & - & - & 0.0000 & 0.0000 \\
\texttt{arm} & 0.0000 & - & 0.0001 & - & 0.0000 & 0.0010 \\
\texttt{tcm} & 0.0000 & - & - & - & - & 0.0000 \\
\bottomrule
\end{tabular}
}
\caption{
  \rshl{
  Average GPU memory usage (GB) for each MCMC step in 1 hour, its ratio over $\mours$, and standard deviation.
  Each reported number is averaged over 8 runs $\times$ averaged number of steps executed in each one-hour run.
  A value of $\mathrm{nan}$ indicates it fails to execute a single MCMC step within 1 hour.
  A dash (-) indicates that the method is not applicable to the benchmark.
  Lower is better.
  }
}
\label{table:memory-mcmc}
\end{table}

}

\end{document}